\RequirePackage{xcolor}
\RequirePackage{silence}
\AddToHook{package/hyperref/before}{\RequirePackage{setspace}}
\documentclass{jpc} %

\keywords{differential privacy, network statistics, social networks, connectedness indices, privacy--utility tradeoff}

\usepackage{natbib}
\usepackage{booktabs}
\usepackage{mathtools}
\usepackage{mathrsfs}
\usepackage{dsfont}
\usepackage[export]{adjustbox}
\usepackage{tikz}
\usetikzlibrary{calc, fit, backgrounds, positioning, shapes.geometric, arrows.meta, shadows}
\usepackage{tikz-network}
\usepackage{algorithm}
\usepackage{algpseudocode}
\usepackage{placeins}
\usepackage{float}
\usepackage{aliascnt} %
\usepackage[capitalise]{cleveref}
\usepackage{subcaption}

\newcommand{\cA}{\mathcal{A}}
\newcommand{\cB}{\mathcal{B}}
\newcommand{\cM}{\mathcal{M}}

\usepackage{chngcntr} %

\newif\ifshowrevisions
\showrevisionsfalse
\DeclareRobustCommand{\rev}[1]{%
  \ifshowrevisions{\color{teal}#1}\else#1\fi}

\newcommand{\defeq}{\vcentcolon=}

\newcommand{\E}{\operatorname{\mathbf{E}}}

\newcommand{\1}{\mathbf 1}

\newcommand{\toP}{\xrightarrow{p}}
\DeclareMathOperator{\Cov}{Cov}
\DeclareMathOperator{\Var}{Var}

\let\cite\citet

\newaliascnt{theorem}{thm}
\newtheorem{theorem}[theorem]{Theorem}
\aliascntresetthe{theorem}
\newaliascnt{lemma}{thm}
\newtheorem{lemma}[lemma]{Lemma}
\aliascntresetthe{lemma}
\newaliascnt{proposition}{thm}
\newtheorem{proposition}[proposition]{Proposition}
\aliascntresetthe{proposition}
\newaliascnt{definition}{thm}
\newtheorem{definition}[definition]{Definition}
\aliascntresetthe{definition}
\newaliascnt{corollary}{thm}

\aliascntresetthe{corollary}
\newaliascnt{assumption}{thm}
\newtheorem{assumption}[assumption]{Assumption}
\aliascntresetthe{assumption}
\newaliascnt{remark}{thm}
\newtheorem{remark}[remark]{Remark}
\aliascntresetthe{remark}
\newaliascnt{observation}{thm}

\aliascntresetthe{observation}
\newaliascnt{claim}{thm}

\aliascntresetthe{claim}
\newaliascnt{axiom}{thm}

\aliascntresetthe{axiom}
\newaliascnt{hypothesis}{thm}

\aliascntresetthe{hypothesis}
\newaliascnt{postulate}{thm}

\aliascntresetthe{postulate}
\newaliascnt{puzzle}{thm}

\aliascntresetthe{puzzle}

\crefname{theorem}{Theorem}{Theorems}
\Crefname{theorem}{Theorem}{Theorems}
\crefname{proposition}{Proposition}{Propositions}
\Crefname{proposition}{Proposition}{Propositions}
\crefname{assumption}{Assumption}{Assumptions}
\Crefname{assumption}{Assumption}{Assumptions}
\crefformat{theorem}{\rev{Theorem}~#2#1#3}
\Crefformat{theorem}{Theorem~#2#1#3}
\crefformat{proposition}{Proposition~#2#1#3}
\Crefformat{proposition}{\rev{Proposition}~#2#1#3}
\crefformat{assumption}{\rev{Assumption}~#2#1#3}
\Crefformat{assumption}{Assumption~#2#1#3}

\makeatletter
\DeclareRobustCommand{\(}{\relax\ifmmode\@badmath\else$\fi}
\DeclareRobustCommand{\)}{\relax\ifmmode$\else\@badmath\fi}
\makeatother

\begin{document}

\title[DP for Network Connectedness Indices]{Differential Privacy for Network Connectedness Indices}

\author[T.~A.~Rutter]{Tom A.\ Rutter}
\address{Department of Economics, Stanford University}
\email{trutter@stanford.edu}

\author[Y.~Liu]{Yuxin Liu}
\address{Swanson School of Engineering, University of Pittsburgh}
\email{yul435@pitt.edu}

\author[M.~A.~Rahimian]{M.\ Amin Rahimian}
\address{Swanson School of Engineering, University of Pittsburgh}
\email{rahimian@pitt.edu}

\begin{abstract}
  
\noindent Researchers increasingly use data on social and economic networks to study a range of social science questions, but releasing statistics derived from networks can raise significant privacy concerns. We show how to release network connectedness indices that quantify assortative mixing across node attributes under edge-adjacent differential privacy. \rev{Directly adding noise to connectedness indices based on global sensitivity can result in poor accuracy} both because connectedness indices have high global sensitivity and because a single node's attribute can potentially be an input to connectedness in thousands of cells, leading to poor composition. Our method, which is straightforward to apply, first \rev{randomizes} node attributes, then analytically debiases downstream statistics, and finally \rev{adds noise} to protect the presence or absence of individual edges. We prove consistency and asymptotic normality of our estimators for both discrete and continuous labels \rev{and evaluate their accuracy in simulations. In 46 village networks with about 100 to 350 households, the private estimates capture village-to-village differences in the share of social ties that historically disadvantaged households have outside their group.}

\vspace{-1pt}
\end{abstract}

\maketitle

\section{Introduction}

    Social scientists interested in a range of social phenomena are increasingly using data on the \textit{networks} in which individuals, firms, financial institutions and nations are embedded. (See, for example, \cite{easley2010networks}, \cite{jackson_human_2019}, and \cite{goyal2023networks} for overviews.) The use of network data has provided valuable insights on the spread of contagious diseases, the role of peers and mentors in shaping people's life outcomes, the resilience of global supply chains, and the nature of systemic risk in the financial system. The use of such data, however, comes with concerns about the privacy protections provided to individuals and organizations in the dataset. The naive publication of aggregated statistics from these datasets entails risks that information may be leaked about the presence (or lack thereof) of particular connections in the dataset, or even about the individual characteristics of people or businesses represented in the data. 

    In response to broader concerns about guaranteeing privacy for individuals and organizations represented in datasets used by social scientists, the literature on differential privacy, starting from \cite{dwork2006calibrating}, has developed tools for allowing particular statistics to be released from datasets while simultaneously providing a mathematical guarantee that, without making any assumptions on the external knowledge someone might have, it is impossible for anyone to learn too much about a given individual or organization in the data just from seeing the aggregated statistic. 
    
    This literature was inspired by the failure of traditional data protection techniques, such as the removal of personally identifiable information from datasets, to provide adequate protection to individuals whose records are contained in the data. For example, \cite{narayanan2008robust} were able to cross-reference anonymized Netflix movie ratings with public IMDB profiles to identify specific users in the Netflix database, despite the fact that the Netflix movie ratings dataset did not in itself contain any personally identifiable information. Even aggregated datasets can reveal surprising amounts of information about particular individuals in the dataset if only naive data protection techniques, such as masking small cells or swapping some individuals between cells, are used. For example, \cite{dick2023confidence} provide reconstruction attacks on the 2010 US Census that allow them to reliably learn the characteristics of many \textit{individuals} from the aggregated data alone. The reason their attack works is that the Census released many tables containing aggregate statistics, so differencing the statistics in particular tables (which aggregate over different, but not mutually exclusive, sets of individuals) allows them to hone in on the data for particular individuals. The increasing availability of compute power and access to artificial intelligence will dramatically reduce the cost of carrying out future attacks on public data as well as the mathematical expertise required to execute them successfully. 
    
    The key idea in the differential privacy literature, which provides a mathematical guarantee that no possible attack can reveal much information about any one data point in the dataset, is to infuse random statistical noise into the aggregated statistic so that an outsider is not able to separate the influence of a given individual on a statistic versus the influence of the random noise. Typically this noise is applied on the basis of the \textit{sensitivity} of the statistic to a change in any particular observation in the dataset. In the case of the Census reconstruction attack, the statistical noise introduced completely defeats the logic that differencing different releases allows an attacker to get at the information of a particular individual in the dataset. 
    
    For statistics constructed from networks, it can often be the case that the sensitivity of the statistic covers the entirety of the parameter space, limiting the fidelity of the privatized statistic to its true value while maintaining a particular privacy guarantee. Indeed, we show that this is the case for the class of network statistics we consider in this paper, so an alternative approach is required.

    In this paper, we provide an approach for releasing a particular class of network statistics with a formal privacy guarantee that maintains high accuracy even for relatively small realistic networks. The statistics we consider are \textit{network connectedness indices}, which describe the proportion of connections from nodes with a given characteristic that extend to nodes with either the same characteristic or a different characteristic. For example, what proportion of the friends of white individuals are also white? What proportion of the suppliers of US firms are based in China? What proportion of bank loans are directed to hedge funds?  

    The application we focus on, which motivated this paper, concerns the proportion of friendships that cut across the income distribution. In a widely cited pair of studies, \cite{chettyetal2022I, chettyetal2022II} publicly released a range of social capital indices constructed from Facebook and Instagram data for counties, zip codes, universities, and schools in the US. In particular, \cite{chettyetal2022I} released a measure of the fraction of high-socioeconomic status (SES) friends among low-SES individuals in an area, which they single out as the strongest predictor of economic mobility in the US. (They use data from \cite{chetty2026opportunity} and define the economic mobility of an area as the average adult income rank of a child born to parents at the 25th percentile of the national income distribution.) The statistics they publicly released have been widely downloaded and used in follow-on research. Although the authors of these studies used tools from the differential privacy literature to reduce the risk of privacy loss, they did not provide a formal privacy guarantee for the connectedness indices they released. Subsequent work by \cite{harris2025social} extended these measures to the United Kingdom, and while they also infused noise into the data before making it public, they similarly did not provide a formal privacy guarantee. \rev{\ifshowrevisions\hypersetup{citecolor=teal}\fi Researchers have not just used large-scale platform data for calculating network connectedness indices---\cite{power2026social} similarly consider connectedness across the wealth distribution in 46 communities using hand-collected survey data.} 

    In this paper, we provide a simple-to-use method for publicly releasing network connectedness indices, such as economic connectedness, with a formal privacy guarantee. \rev{We apply our method to 46 village networks with 99 to 354 households collected by \cite{banerjee_diffusion_2013}. For each household from a historically disadvantaged caste, we measure the share of its social ties to households outside that group. At a privacy budget of $\varepsilon=8$, the private estimates closely track the observed village averages.} 

    The key logic behind our approach is to first inject random noise into the characteristics of nodes in the data, such that downstream statistics are calculated from a database of node characteristics that already satisfies a privacy guarantee. This is crucial in allowing us to sidestep consideration of the sensitivity of our statistics of interest to a change in node characteristics in ``worst-case'' networks, such as star networks, which do not occur in practice. It also allows us to sidestep the issue that a change in the attribute of a node can potentially affect connectedness in a large number of cells, leading to poor composition. Instead, by privatizing node characteristics and then assuming that all downstream statistics are computed from these sanitized node-level data, we only have to consider the sensitivity of our statistics to changes in the edge set, for which we show the sensitivity is much lower. Additionally, a particular edge (or lack thereof) can only change the value of connectedness in at most two cells, leading to much better composition. 

    Injecting random noise into the node characteristics will often bias the downstream statistics, so we provide formulae for analytic corrections that can also be applied downstream to debias these statistics while maintaining their privacy guarantees by the postprocessing theorem \citep[Proposition 2.1]{dwork2014algorithmic}. %

    We note, crucially, that the database of individual characteristics is usually not of interest and so does not have to be (and usually should not be) released. Although not releasing this intermediate dataset does not reduce the \textit{mathematical} privacy loss from data release, since we make no assumptions about the background knowledge of outsiders, in practice it provides an additional layer of privacy protection, since realistic outsiders do not have the extreme level of information that usually constitutes the worst case for a researcher trying to provide a privacy guarantee \citep{bassily_coupled-worlds_2013}.

    Our paper proceeds as follows. Section 2 provides an overview of relevant literature on differential privacy and the network statistics we consider. Section 3 outlines the definitions of privacy we use and network connectedness indices. Section 4 presents our method for binary labels. Section 5 extends our framework to continuous labels. Section 6 presents simulations that evaluate the accuracy of our methods. Section 7 demonstrates the empirical performance of our methods on real-world labeled networks, and \rev{Section} 8 concludes.

    \section{Literature Review}

    Our paper builds on the vast literature on differential privacy that began with \cite{dwork2006calibrating}. \rev{Differential privacy compares the distributions of a mechanism's outputs on neighboring datasets and does not require independent records \citep{dwork2014algorithmic}. When records are correlated, a release can reveal information about one person through the records of others \citep{kifer_no_2011,ghosh2016inferential}. For example, \citet{liu2016dependence} found that knowledge of friendships improved estimates of users' location patterns from differentially private cluster centers and other users' check-ins.} 
    
    \rev{Our paper builds in particular on the literature studying differentially private analyses of graphs. \cite{nissim_smooth_2007} first studied the differentially private analysis of graphs and developed methods for privately releasing certain graph statistics such as triangle counts. Subsequent work emphasized the difference between edge and node differential privacy \citep{hay2009accurate}, and developed private algorithms for other graph statistics such as subgraph counts \citep{karwa2014private}, as well as private algorithms that satisfy the stronger notion of node differential privacy \citep{kasiviswanathan2013analyzing}. \cite{raskhodnikova2016differentially} and \cite{mueller2024differentially} provide surveys of this literature. One difference between our setting and much of this literature is that the statistics we consider depend on both the edge structure of the graph and the characteristics of nodes. Our paper therefore starts from the concept of a labeled network provided by \cite{blocki2013differentially}, and uses their notion of edge-adjacent privacy.} 

    \rev{A related literature studies differential privacy for graph neural networks. This is also a setting where both graph structure and node attributes can contain sensitive information that should be protected. \citet{sajadmanesh2023gap} perturb neighborhood aggregations to provide edge-level and node-level privacy guarantees for graph neural network training and inference. \citet{sajadmanesh2024progap} extend this approach using progressive training. Other contributions in this literature include \cite{sajadmanesh2021locally} and \cite{chien2023differentially}. See \cite{zhang2024privacygnn} for a survey. While this literature focuses on privately learning predictive models from graph data, our objective is to release a particular class of statistics---network connectedness indices---from network data with a privacy guarantee.}

    Our main contribution is to the practice of economic data release, and so our work builds on \cite{chettyetal2022I, chettyetal2022II} who developed the notion of economic connectedness and released these statistics for granular US geographies, schools, and universities. The central innovation in those papers is showing that measures of network connectedness constructed by combining information about who is connected to whom with information about the characteristics of the connected individuals can significantly outperform measures based only on the edge structure of the network (such as the clustering coefficient) or measures based only on node characteristics that ignore the connections between nodes (such as variables related to a neighborhood's income composition) in explaining variation in economic outcomes such as intergenerational mobility. This insight that both the structure of the connections and the characteristics of connected individuals matter directly informs the notion of edge-adjacent privacy we use from \cite{blocki2013differentially}, and allows us to provide a guarantee that we protect both the presence (or absence) of edges and the characteristics of nodes. 

    Our paper also relates to a literature on differential privacy for dependent data. Dependent differential privacy and its generalizations address pairwise and arbitrary correlations, though they often rely on dependence parameters that are difficult to estimate without explicit probabilistic models \citep{liu2016dependence, zhao2017dependent}. A more flexible alternative is the Pufferfish framework \citep{kifer2014pufferfish}, which allows for customized privacy definitions—such as Blowfish and Bayesian DP—by explicitly modeling secrets, discriminative pairs, and data generation distributions. In the context of graph data, Pufferfish implementations like the Wasserstein and Markov quilt mechanisms have demonstrated superior utility over group DP by leveraging Bayesian networks to model dependencies \citep{song2017pufferfish}. Complementary approaches include inferential privacy, which uses Markov chains to quantify correlation \citep{ghosh2016inferential}, and zero-knowledge privacy, which provides stronger guarantees for specific graph properties like connectivity \citep{gehrke2011towards}. 

    Especially in graph settings, a common approach is to project onto a graph with a low maximum degree, to decrease the sensitivity of statistics \citep{kasiviswanathan2013analyzing}. In many settings this technique works very well, but for the applications we are interested in, degree heterogeneity is an important feature of the data that can be lost when projecting onto a graph with a low maximum degree. (For example, \cite{harris2025social} show that degree variation by SES in the UK is an important feature of networks when considering connectedness indices---imposing a low degree bound in this particular setting could dramatically change estimates of connectedness.) See \cite{hehir2025interpreting} for a review of various notions of differential privacy for networks and their interpretation. 
\section{Privacy and Connectedness Indices}

    We follow \cite{dwork2006calibrating} and the differential privacy literature in adopting a definition of privacy that requires the probability that any statistical output being produced from two close datasets to be similar. 

    \begin{definition}[$\varepsilon$-differential privacy]\onehalfspacing
        A randomized mechanism $\mathcal{M}$ is $\varepsilon$-differentially private if\rev{,} for all adjacent datasets $D$ and $D'$ and all \rev{measurable sets $S$ of possible outputs,} 
        \begin{align*}
            \rev{\Pr(\mathcal{M}(D) \in S)
            \leq e^\varepsilon \Pr(\mathcal{M}(D') \in S).}
        \end{align*}
    \end{definition}
    An intuitive interpretation of this definition is that an adversary with unrestricted information about the rest of the dataset cannot become much more confident that the true dataset is $D$ as opposed to $D'$, or vice versa, after observing the output of $\mathcal{M}$. 

    \rev{We will often make use of the following relaxation of $\varepsilon$-differential privacy, following \cite{dwork2014algorithmic}.

    \begin{definition}[$(\varepsilon,\delta)$-Differential Privacy]
    A randomized mechanism $\mathcal{M}$ is $(\varepsilon,\delta)$-differentially private if, for all adjacent datasets $D$ and $D'$ and all measurable sets $S$ of possible outputs,
    \[
    \Pr(\mathcal{M}(D)\in S)
    \leq e^\varepsilon \Pr(\mathcal{M}(D')\in S)+\delta.
    \]
    \end{definition}

    Note that this definition collapses to $\varepsilon$-differential privacy when $\delta = 0$.} We next define a labeled network.

\begin{definition}[Labeled Network]\onehalfspacing
A labeled network is a triple $(\mathcal{V}, \mathcal{E}, \mathbf{L})$, where $\mathcal{V}$ is the set of vertices, $\mathcal{E}$ is the set of edges, and 
\[
\mathbf{L}= (l_i)_{i\in\mathcal V}, \qquad l_i\in\rev{\mathcal{L}},
\]
is the node label vector that assigns each node a label from \rev{a specified label set $\mathcal{L}$. We use $\mathcal{L}=\{a,b\}$ for binary labels and $\mathcal{L}=[0,1]$ for the bounded continuous labels in Section~5. We assume throughout that the network has no self-loops.}
\end{definition}

    Note that this definition is analogous to the definition of a \textit{social network} in \cite{blocki2013differentially}. A pair of possible labels for a node could be, for example, ``high SES'' or ``low SES''.
    For notation we define $e_{ij} = 1$ if $(i,j) \in \mathcal{E}$ (that is, if there is an edge between $i$ and $j$) and $e_{ij} = 0$ otherwise. We denote node $i$'s degree as: 
    \begin{align*}
        d_i = \sum_{j \in \mathcal{V}} e_{ij}
    \end{align*}
    and the neighborhood of node $i$ as the set of nodes $i$ is connected to: 
    \begin{align*}
        N(i) = \{ j \in \mathcal{V} : e_{ij} = 1 \}
    \end{align*}

\begin{definition}[\rev{Cross-Type and Same-Type Connectedness Indices}]\onehalfspacing
Take a labeled network $(\mathcal V,\mathcal E,\mathbf L)$, where $\mathbf L=(l_i)_{i\in\mathcal V}$ and $l_i\in\{a,b\}$. Define the induced partition of nodes by
\[
\mathcal A:=\{i\in\mathcal V:l_i=a\},\qquad \mathcal B:=\{i\in\mathcal V:l_i=b\}.
\]
The \textit{cross-type connectedness index} from $\mathcal A$ to $\mathcal B$ is defined as
\[
C^{\mathcal A\to\mathcal B}:=\frac{1}{\#(\mathcal A)}\sum_{i\in\mathcal A}\frac{\sum_{j\in\mathcal B}e_{ij}}{d_i},
\]
\rev{and the \textit{same-type connectedness index} for group $\mathcal A$ is defined as}
\[
C^{\mathcal A\to\mathcal A}:=\frac{1}{\#(\mathcal A)}\sum_{i\in\mathcal A}\frac{\sum_{j\in\mathcal A}e_{ij}}{d_i}.
\]
\end{definition}

    \cite{chettyetal2022I, chettyetal2022II} produce three different types of cross-type connectedness index for US counties, zip codes, schools, and universities in the US: 
    \begin{enumerate}
        \item economic connectedness (the average fraction of high-SES friends among low-SES individuals).
        \item language connectedness (the average fraction of friends who use English as their primary language among individuals who do not use English as their primary language).
        \item age connectedness (the average fraction of friends age 35--44 among individuals age 25--34).  
    \end{enumerate}
    \cite{bailey2025cross} produce a cross-type connectedness index measuring the average fraction of female friends among male individuals.
    
    For notation, we let: 
    \begin{align*}
        \rho_i \defeq \frac{\sum_{j \in \mathcal{B}} e_{ij}}{d_i}
    \end{align*}
    from here on, so that a cross-type connectedness index can be written: 
    \begin{align*}
        C^{\mathcal{A}\to\mathcal{B}} = \frac{1}{\#(\mathcal{A})}\sum_{i \in  \mathcal{A}} \rho_i 
    \end{align*}
    We provide a simple example illustrating the cross-type connectedness index in Figure~\ref{fig:example}. In Table~\ref{tab:notation}, we provide a summary of our notation. (We provide a more extensive list of our notation in Table~\ref{tab:notation-full}.)  
    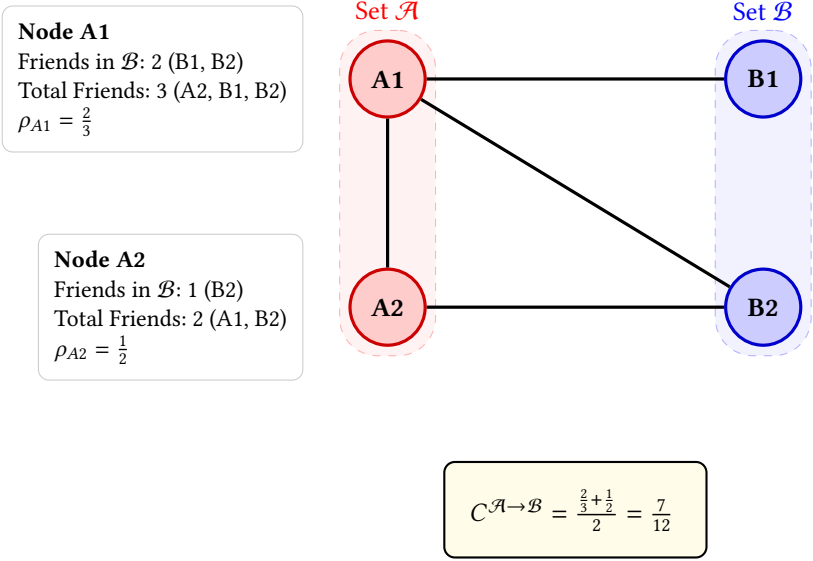
\begin{figure}
        \centering
        \begin{tikzpicture}[
    node_a/.style={circle, draw=red!80!black, fill=red!20, very thick, minimum size=1cm, font=\bfseries},
    node_b/.style={circle, draw=blue!80!black, fill=blue!20, very thick, minimum size=1cm, font=\bfseries},
    friendship/.style={very thick, black},
    annot_box/.style={rectangle, draw=gray!40, fill=white, rounded corners, font=\small, align=left, inner sep=6pt}
]

\node[node_a] (A1) at (0, 3) {A1};
\node[node_a] (A2) at (0, 0) {A2};

\node[node_b] (B1) at (5, 3) {B1};
\node[node_b] (B2) at (5, 0) {B2};

\draw[friendship] (A1) -- (B1);
\draw[friendship] (A1) -- (B2);
\draw[friendship] (A1) -- (A2); %
\draw[friendship] (A2) -- (B2);

\begin{pgfonlayer}{background}
    \node[fit=(A1)(A2), draw=red!30, dashed, fill=red!5, rounded corners=0.5cm, 
          label={[red]above:Set $\cA$}] (SetA) {};
    \node[fit=(B1)(B2), draw=blue!30, dashed, fill=blue!5, rounded corners=0.5cm, 
          label={[blue]above:Set $\cB$}] (SetB) {};
\end{pgfonlayer}

\node[annot_box, left=0.6cm of A1] (calc1) {
    \textbf{Node A1}\\
    Friends in $\cB$: 2 (B1, B2)\\
    Total Friends: 3 (A2, B1, B2)\\
    $\rho_{A1} = \frac{2}{3}$
};

\node[annot_box, left=0.6cm of A2] (calc2) {
    \textbf{Node A2}\\
    Friends in $\cB$: 1 (B2)\\
    Total Friends: 2 (A1, B2)\\
    $\rho_{A2} = \frac{1}{2}$
};

\node[draw, thick, fill=yellow!10, rounded corners, below=1.5cm of A2, xshift=2.5cm, align=center, inner sep=10pt] {
    \textbf{$C^{\cA\to\cB} = \frac{\frac{2}{3} + \frac{1}{2}}{2} = \frac{7}{12}$}
};

\end{tikzpicture}
        \caption{An example of how to calculate the cross-type connectedness index.}
        \label{fig:example}
    \end{figure}
    
    \rev{We adopt the convention that $\rho_i$ is zero for an individual with no connections.\footnote{Such a convention is natural in these settings, where researchers are often interested in cross-group connections because they give individuals access to resources and information they would otherwise not be able to access.} Formally, in the case where some nodes have zero connections, we define
    \[
    a_{ij}:=
    \begin{cases}
    e_{ij}/d_i, & d_i>0,\\
    0, & d_i=0,
    \end{cases}
    \]
    so that $\rho_i=\sum_j a_{ij}\mathbf 1\{l_j=b\}$ is defined for every node. In particular note that $\sum_j a_{ij}=\mathbf 1\{d_i>0\}$.}

    \begin{table}[ht]
    \centering
    \caption{Summary of Notation}
        \label{tab:notation}
        \begin{tabular}{@{}ll@{}}
        \toprule
        \textbf{Symbol} & \textbf{Description} \\ \midrule
        $\mathcal{V}$ & Set of vertices (nodes) in the network \\
        $\mathcal{E}$ & Set of edges (friendships/connections) between nodes \\
       $\mathbf{L}=(l_i)_{i\in\mathcal V}$ & Node label vector with $l_i\in\{a,b\}$ for each node $i$ \\

        $e_{ij}$ & Binary indicator; $e_{ij}=1$ if an edge exists between $i$ and $j$, else $0$ \\
        $d_i$ & Degree of node $i$, $\#(j \in \mathcal{V}:e_{ij} = 1)$. \\ 
        $N(i)$ & Neighborhood of node $i$, $\{j \in \mathcal{V}:e_{ij} = 1\}$. \\ 
        $\mathcal{A}, \mathcal{B}$ & Partitions of $\mathcal{V}$ based on labels ($a, b$) \\
        $\#( \cdot )$ & Cardinality operator (number of elements in a set) \\
        $\rho_i$ & Individual $i$'s fraction of friends belonging to the target group \\
        $C^{\mathcal{A}\to\mathcal{B}}$ & Cross-type connectedness index (average $\rho_i$ for individuals in group $\mathcal{A}$) \\ \bottomrule
        \end{tabular}
    \end{table}

    We often will want to compute cross-type connectedness indices for a particular \textit{cell}, where a cell could be, for example, a county, a zip code, or a school. In that case, if we let $s$ denote the cell (a set of users, such as a school) we can consider the statistics: 
    \begin{align*}
        C_s ^{\mathcal{A}\to\mathcal{B}} = \frac{1}{\#(\mathcal{A} \cap s)}\sum_{i \in  \mathcal{A} \cap s} \frac{\sum_{j \in \mathcal{B}} e_{ij}}{d_i}
    \end{align*}
    or: 
    \begin{align*}
        C_s ^{\mathcal{A}\to\mathcal{B}} = \frac{1}{\#(\mathcal{A} \cap s)}\sum_{i \in  \mathcal{A} \cap s} \frac{\sum_{j \in \mathcal{B} \cap s} e_{ij}}{\sum_{j \in \mathcal{V} \cap s}  e_{ij}}
    \end{align*}
    with the difference being whether we consider only friendships from individuals in the cell to other individuals in the cell, or whether we consider connections from individuals in the cell to any other individual. \rev{We consider a node's cell membership to be public.}

    For our privacy guarantee, we will follow the notion of \textit{edge-adjacent DP} from \cite{blocki2013differentially}, which allows us to protect both the characteristics of individual nodes in the network (such as an individual's income or race) and the presence or absence of particular edges in the network. This is $\varepsilon$-DP when considering labeled networks as adjacent if they satisfy the following property: 
\begin{definition}[Edge-Adjacent Labeled Networks]\label{def:ealn}
Two labeled networks $(\mathcal{V}, \mathcal{E}, \mathbf{L})$ and $(\mathcal{V}, \mathcal{E}', \mathbf{L}')$ are said to be \emph{edge-adjacent} if:
\begin{itemize}
\item the edge sets $\mathcal{E}$ and $\mathcal{E}'$ differ by at most one edge, and
\item the label vectors $\mathbf{L}=(l_i)_{i\in\mathcal V}$ and $\mathbf{L}'=(l_i')_{i\in\mathcal V}$ differ in at most one component, i.e., there exists at most one node $u\in\mathcal V$ such that $l_u\neq l_u'$.
\end{itemize}
\end{definition}

    That is, two networks are edge adjacent if they differ in the presence of at most one edge and the characteristic of at most one node. The name edge-adjacent is somewhat unfortunate since it does allow for node characteristics to vary, but we stick with this terminology for consistency with the prior literature. The notion of edge-adjacent DP has been employed by \cite{jorgensen2016publishing} and \cite{Chen2019PublishingCA}. We note that our notion of edge-adjacency is slightly stronger than that stated by \cite{blocki2013differentially}, since we allow for a simultaneous change in \textit{both} an edge and the characteristics of one node, as opposed to a change in \textit{either} an edge or the characteristics of one node. 

    Note that our adjacency notion does not treat the \emph{presence} of a node in the dataset as private. Requiring node presence privacy would move us toward node-level DP (node-DP), under which two graphs are adjacent if one can be obtained from the other by adding/removing a vertex together with (potentially) all incident edges. In social network settings, this stronger notion is typically infeasible for connectedness statistics: the removal of a single high-degree vertex can drastically change cross-type connectivity, forcing any DP mechanism to inject noise at a scale that overwhelms the signal. Figure~\ref{fig:star-net} provides a simple illustration: changing the center node (or equivalently, protecting its presence/incident edges under node-DP) can shift the red-to-blue connectedness from $0$ to $1$, so the sensitivity spans the entire parameter space and does not diminish with network size. 

    \begin{figure}
    \centering
    \resizebox{\linewidth}{!}{\begin{tikzpicture}[
    node_a/.style={circle, draw=red!80!black, fill=red!20, very thick, minimum size=0.9cm, inner sep=1pt, font=\bfseries\sffamily\color{red!50!black}},
    node_b/.style={circle, draw=blue!80!black, fill=blue!20, very thick, minimum size=0.9cm, inner sep=1pt, font=\bfseries\sffamily\color{blue!50!black}},
    friendship/.style={very thick, black},
    annot_box/.style={rectangle, draw=gray!40, fill=white, rounded corners, font=\small\sffamily, align=center, inner sep=8pt, drop shadow, minimum width=5cm}
]

\def\radius{2.5}   %
\def\nodes{6}      %
\def\offset{10}    %

\begin{scope}[local bounding box=LeftNet]
    \node[node_a] (L_c) at (0, 0) {R};
    
    \foreach \i in {1,...,\nodes} {
        \node[node_a] (L_p\i) at ({360/\nodes * (\i - 1)}:\radius) {R};
        \draw[friendship] (L_c) -- (L_p\i);
    }
\end{scope}

\begin{scope}[xshift=\offset cm, local bounding box=RightNet]
    \node[node_b] (R_c) at (0, 0) {B};
    
    \foreach \i in {1,...,\nodes} {
        \node[node_a] (R_p\i) at ({360/\nodes * (\i - 1)}:\radius) {R};
        \draw[friendship] (R_c) -- (R_p\i);
    }
\end{scope}

\draw[-{Stealth[length=12pt, width=8pt]}, ultra thick, gray!40, shorten >= 1cm, shorten <= 1cm] 
    (LeftNet.east) -- (RightNet.west) 
    node[midway, above, font=\sffamily\bfseries\small, text=black, align=center] {Change the label\\of the central node\\};

\coordinate (BoxLevel) at (0, -4.5);

\node[annot_box] at (0, 0 |- BoxLevel) {
    \textbf{Network G}\\[8pt]
    Every Red node has \textbf{zero} Blue friends.\\[8pt]
    Red-to-Blue Connectedness = \textbf{0}
};

\node[annot_box] at (\offset, 0 |- BoxLevel) {
    \textbf{Network G'}\\[8pt]
    Every Red node has \textbf{only} Blue friends.\\[8pt]
    Red-to-Blue Connectedness = \textbf{1}
};

\end{tikzpicture}}
    \caption{A star-network illustration of why node-level privacy notions can be too strong for connectedness statistics: modifying the center node (or protecting its presence/incident edges under node-DP) can move red-to-blue connectedness from $0$ to $1$, implying sensitivity that spans $[0,1]$ and does not diminish with network size.}
    \label{fig:star-net}
\end{figure}
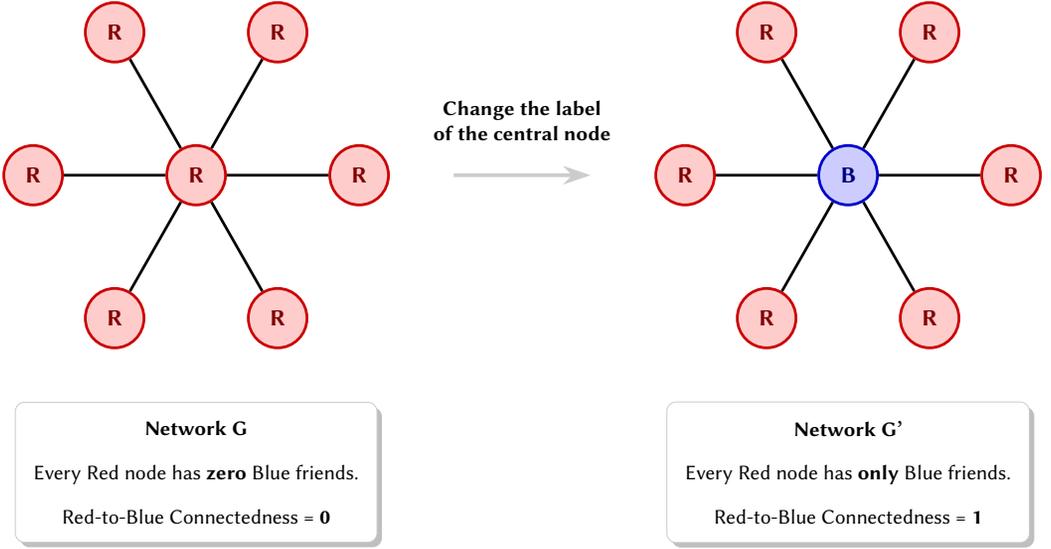

Accordingly, we follow the modeling perspective advocated by \citet{Chen2019PublishingCA}: on many platforms the existence of user profiles (and hence the vertex set $\mathcal{V}$) is effectively public, whereas personal characteristics recorded on profiles are sensitive. Our goal is therefore to protect node \emph{attributes} and individual \emph{edges} (via edge-adjacent DP), rather than node presence, which would render the connectedness indices of interest non-actionable.

\subsection{A Composition Theorem for Implementing Edge-Adjacent Differential Privacy}\label{sec:edge-adj-dp}

In this section, we provide a method that allows the release of connectedness indices with a formal privacy guarantee.

The core logic behind our approach is as follows. A connectedness index risks releasing information about two characteristics of a given node in the network— the node's label and its friendships— and this is reflected in the notion of edge-adjacent DP we use. The approach used by \cite{chettyetal2022I,chettyetal2022II} to protect economic connectedness involves computing the economic connectedness metric and then applying noise to the aggregated statistic on the basis of its local sensitivity to changes in the node's characteristic and friendships.\footnote{The notion of local sensitivity is from \cite{nissim_smooth_2007}. However, while the smooth sensitivity approach developed by \cite{nissim_smooth_2007} is formally private, the local-sensitivity-based method used by \cite{chettyetal2022I, chettyetal2022II}, which leverages the maximum observed sensitivity approach of \cite{chetty2019practical}, is not formally private because the maximum observed sensitivity itself is a function of the private data and so can leak information about outliers.}

We show that applying noise to the node characteristics and then adjusting for the bias this creates in the ensuing estimate of economic connectedness, followed by applying noise on the basis of sensitivity to changes in a node's friendship network, is a better approach that allows for a formal privacy guarantee while injecting small amounts of noise.\footnote{It also sidesteps a key difficulty that would otherwise be embedded in this problem, which is that the label of a given node is potentially an input to the connectedness index for \textit{every} single cell in the dataset (since a node can potentially be connected to a node in every other cell). This causes composition issues for approaches based on the global sensitivity of change to a node's label, since privacy losses compose in the product of the number of cells and number of statistics released. Our approach only composes in the number of statistics released. This is because once the node labels are privatized, by the postprocessing theorem, any number of subsequent cell-level indices can be derived from these private labels without further depleting the privacy budget reserved for node characteristics.}

Our approach rests on the following theorem (proved in \cref{app:proof:thm:edgeadj-composition-epsdelta}), which formalizes the composition principle underlying this construction. It states that applying a private mechanism to node attributes, and then a separate edge-DP mechanism for the aggregate statistic that takes as an input the privatized attributes, composes in a way that satisfies edge-adjacent DP.

\begin{theorem}[$(\varepsilon,\delta)$ composition under edge-adjacency] \label{thm:edgeadj-composition-epsdelta}
Let $(\mathcal V,\mathcal E,\mathbf L)$ and $(\mathcal V,\mathcal E',\mathbf L')$ be edge-adjacent labeled networks. \rev{Let $\mathcal Y$ be a measurable output space.} \rev{Let $\mathcal M_1:\mathbf L\mapsto\hat{\mathbf L}$ be $(\varepsilon_\ell,\delta_\ell)$-differentially private with respect to changing a single node attribute}, and for every fixed $\hat{\mathbf L}$ let $\mathcal{M}_2(\,\cdot\,,\hat{\mathbf L}):(\mathcal{V},\mathcal{E})\mapsto \rev{\mathcal Y}$ be $(\varepsilon_e,\delta_e)$-differentially private with respect to changing a single edge. Define the composed mechanism by
\[
\mathcal M(\mathcal V,\mathcal E,\mathbf L):=\mathcal M_2(\mathcal V,\mathcal E,\mathcal M_1(\mathbf L)).
\]
Then $\mathcal M$ is $(\varepsilon_\ell+\varepsilon_e,\delta_\ell+\rev{e^{\varepsilon_\ell}}\delta_e)$ edge-adjacent differentially private.
\end{theorem}

\section{Differential Privacy for Binary Labels: High and Low Status Connectedness}%

We now make use of the general composition framework established in 
\Cref{thm:edgeadj-composition-epsdelta}. 
In particular, we first construct an $\varepsilon_\ell$-DP mechanism 
$\mathcal{M}_1$ for node attributes, and subsequently apply an 
$\varepsilon_e$ edge-DP mechanism to the aggregate statistic 
computed from the privatized labels. 

Our first step is therefore to protect node attributes under $\varepsilon_\ell$-DP using randomized response~\rev{\ifshowrevisions\hypersetup{citecolor=teal}\fi\citep{warner1965randomized}}. Specifically, for each node $i$, the categorical label $l_i \in \{a,b\}$ is independently randomized by flipping its value with probability 
\[
p = \frac{1}{1+e^{\varepsilon_\ell}},
\]
yielding a private label $\hat{l}_i$.

\begin{fact}
\label{fact:randomized-response}
Let each node's attribute be $l_v \in \{a,b\}$, and let 
$\mathbf{L}=(l_v)_{v\in\mathcal V}$ denote the label vector.
Define the randomized mechanism $\mathcal{M}_1$ that independently replaces each node's attribute according to
\[
\mathcal{M}_1(l_v) =
\begin{cases}
l_v, & \text{with probability } 1-p = \dfrac{e^{\varepsilon_\ell}}{1+e^{\varepsilon_\ell}} \\[1em]
\text{the opposite value in } \{a,b\}, & \text{with probability } 
p = \dfrac{1}{1+e^{\varepsilon_\ell}}
\end{cases}
\]
Then $\mathcal{M}_1$ satisfies $\varepsilon_\ell$-differential privacy with respect to changes in a single node's attribute.
\end{fact}
\begin{proof}
    This property of randomized response is well known in the DP literature. See, for example, Result 1 of \cite{wang2016using} for a proof of this property.
\end{proof}

    One potential issue with this approach is that randomized response induces attenuation bias in the connectedness metric. When using an above- vs below-median split, randomly flipping node labels mechanically pushes observed group composition toward $1/2$, even when true homophily is strong. For example, if $p=\frac{1}{2}$, then regardless of the true level of homophily in society, in expectation half of every node’s neighbors will appear to belong to group $b$ under the privatized labels.

We therefore require a correction mechanism that recovers unbiased estimates of cross-type connectedness using only the privatized labels $\hat{l}_i$, without accessing the nonprivate labels $l_i$. Proposition~\ref{prop:mvue} (proved in \cref{app:proof:prop:mvue}) constructs a minimum-variance unbiased estimator (MVUE) that corrects for this attenuation bias using only the randomized-response outputs.

\begin{proposition}[MVUE for individual connectedness]%
\label{prop:mvue}
For each node $i$, define the true individual connectedness
\[
\rho_i 
:= 
\sum_{j\in\mathcal V} a_{ij}\,\1\{l_j=b\},
\qquad 
    a_{ij}:=
    \begin{cases}
    e_{ij}/d_i, & d_i>0,\\
    0, & d_i=0,
    \end{cases}
\]
Let privatized labels $\hat l_j$ be generated by randomized response with flipping probability
\[
p=\frac{1}{1+e^{\varepsilon_\ell}},
\]
and define the observed proxy
\[
\hat{\rho}_i 
:= 
\sum_{j\in\mathcal V} a_{ij}\,\1\{\hat l_j=b\}.
\]
Then the unique minimum-variance unbiased estimator (MVUE) of $\rho_i$ based only on the privatized labels $(\hat l_j)_{j\in\mathcal V}$ is
\begin{equation}
\tilde{\rho}_i 
:= 
\frac{\hat{\rho}_i-p\rev{\1\{\ d_i > 0 \}}}{1-2p}.
\label{eq:rho_tilde_node}
\end{equation}
\end{proposition}

\subsection{Consistency of Debiased Private Estimation for Network Connectedness}

Proposition~\ref{prop:mvue} establishes that $\tilde{\rho}_i$ is unbiased (indeed MVUE) for the individual connectedness $\rho_i$ given the privatized labels.
However, our target estimand is a \emph{cell-level connectedness index}, which aggregates $\rho_i$ over a \emph{latent} subset of nodes (e.g., $A\cap s$) and is implemented via a \emph{ratio} (a H\'ajek-type normalization).
Unbiasedness of the individual components does not automatically imply that the resulting ratio estimator is unbiased or well-behaved in finite samples, and in particular it does not control the probability that the random denominator is close to zero.
Moreover, the final released statistic additionally includes edge-DP noise calibrated to an edge-sensitivity that shrinks with the population size.

Define the debiased weight for membership in $\mathcal{A}$:
\begin{equation}
w_i := \frac{\1\{\hat l_i=a\}-p}{1-2p}.
\label{eq:w}
\end{equation}
Let
\[
S_{0,n} := \sum_{i\in\mathcal V} w_i,\qquad S_{1,n} := \sum_{i\in\mathcal V} w_i\,\tilde{\rho}_i,
\]
where $\tilde{\rho}_i$ is the debiased proxy for $\rho_i$ in \cref{eq:rho_tilde_node}. \rev{Here, $S_{0,n}$ estimates the number of nodes in $\mathcal A$, while $S_{1,n}$ estimates their total connectedness $\sum_{i\in\mathcal A}\rho_i$. Dividing by the estimated group size gives an estimate of the average connectedness, with both sums constructed from privatized labels. This normalization gives the H\'ajek-type ratio estimator:}
\begin{equation}
\rev{\frac{S_{1,n}}{S_{0,n}}.}
\label{eq:hajek}
\end{equation}

\rev{To ensure that the estimator is defined, with a positive denominator, for every realization of the randomized-response mechanism, let $B_n:=\max\{S_{0,n},1\}$. We define the stabilized H\'ajek-type estimator}
\begin{equation}
\rev{\widetilde C^{\mathcal A\to\mathcal B}:=\frac{S_{1,n}}{B_n}.}
\label{eq:hajek-connectedness}
\end{equation}
\rev{Under \cref{ass:piA}, $B_n=S_{0,n}$ with probability tending to one, so this stabilization does not affect the asymptotic behavior of the estimator.}

We next establish a large-sample guarantee: under mild regularity conditions, the debiased H\'ajek estimator is \emph{consistent}, and the additional edge-DP noise is asymptotically negligible, so the overall private release converges in probability to the true connectedness index. To state the result, we first introduce two regularity conditions. 
The first ensures that the target cell remains asymptotically non-negligible, so that the random denominator in the H\'ajek normalization does not go to zero. 
The second is a bounded-degree condition, which controls local dependence and ensures that both the estimator and its edge-sensitivity remain uniformly well behaved.

\begin{assumption}[Non-vanishing fraction of type-$A$ nodes]
There exists $\pi_\cA \in (0,1]$ such that $\#(\mathcal{A})/n \to \pi_\cA$ as $n\to\infty$.
\label{ass:piA}
\end{assumption}

\begin{assumption}[Bounded degree]
There exists $\Delta<\infty$ such that for all $n$ and all $i\in\mathcal V$,
\begin{equation}
d_i \le \Delta.
\label{eq:bounded-degree}
\end{equation}
\label{ass:bounded-degree}
\end{assumption}

In order to implement the second stage of the construction in 
Section~\ref{sec:edge-adj-dp}, we must calibrate the edge-DP noise 
added to the aggregate statistic. 
Since the Laplace mechanism requires an upper bound on the 
edge-sensitivity of the released quantity, we therefore compute 
the sensitivity of $S_{1,n}$ to the addition or removal of a single edge.

Importantly, at this stage the privatized labels $\hat{\mathbf L}$ 
have already been generated via $\mathcal M_1$. 
\rev{Hence, the weights $(w_i)_{i\in\mathcal V}$ and their sum $S_{0,n}$ are fixed, and we only need to control how $S_{1,n}$ changes when the edge set varies.}
The following lemma establishes a uniform bound on this edge-sensitivity (proved in \cref{app:proof:thm:consistency_dp}, along with other auxiliary lemmas for consistency of the private network connectedness estimator).

\begin{lemma}[Edge-sensitivity of $S_{1,n}$]
\label[lemma]{lem:S1_edge_sens}
\rev{Suppose $p\in(0,1/2)$.} Fix $(\hat l_i)_{i\in\mathcal V}$, so that $(w_i)_{i\in\mathcal V}$ is fixed. 
Let $\mathcal E$ and $\mathcal E'$ differ by the addition or removal of a single edge $(u,v)$.
Assume $\hat\rho_i(\mathcal E)$ is computed row-wise as $\hat\rho_i(\mathcal E)=\sum_j a_{ij}(\mathcal E)\,\1\{\hat l_j=b\}$ with 
$a_{ij}(\mathcal E)=e_{ij}/d_i$ \rev{if $d_i>0$ and $a_{ij}(\mathcal E)=0$ if $d_i=0$, as in \cref{prop:mvue}} (so only the normalized weight rows for $u$ and $v$ change when adding or removing $(u,v)$).
Then,
\[
\bigl|S_{1,n}(\mathcal E)-S_{1,n}(\mathcal E')\bigr|\;\le\;\frac{2(1-p)}{(1-2p)^2}.
\]
\end{lemma}

We next study the asymptotic behavior of the proposed estimator.
Although $\widetilde C^{\mathcal{A}\to \mathcal{B}}$ is constructed from unbiased components,
it is a ratio estimator and the final release additionally includes
edge-DP noise. It is therefore not immediate that the private estimator
remains well behaved as $n$ grows.

The following theorem establishes consistency and shows that
the added edge-DP noise is asymptotically negligible (proved in \cref{app:proof:thm:consistency_dp}).

\begin{theorem}[Consistency of the debiased private estimator]
\label{thm:consistency_dp}
Suppose Assumptions~\ref{ass:piA} and~\ref{ass:bounded-degree} hold, with $p\in(0,\tfrac12)$.
Let $\widetilde C^{\mathcal A\to\mathcal B}=S_{1,n}/B_n$ be the stabilized H\'ajek-type estimator in \eqref{eq:hajek-connectedness}, where $B_n=\max\{S_{0,n},1\}$. Define the edge-DP release
\[
\widehat C_{\mathrm{DP}}^{\mathcal A\to\mathcal B}
:=\widetilde C^{\mathcal A\to\mathcal B}+Z_n,
\qquad
Z_n\mid S_{0,n}\sim
\mathrm{Lap}\left(
0,\frac{2(1-p)}{(1-2p)^2\varepsilon_e B_n}
\right).
\]
Then, for any fixed $\varepsilon_e>0$,
\[
\widehat C^{\mathcal{A}\to \mathcal{B}}_{\mathrm{DP}} \;\toP\; C^{\mathcal{A}\to \mathcal{B}},
\; \text{as } n\to\infty.
\]
\end{theorem}

   We summarize our approach in Algorithm \ref{alg:dp-connectedness}.

\begin{algorithm}[H]
\caption{Differentially Private Cross-Type Connectedness Index}
\label{alg:dp-connectedness}
\begin{algorithmic}[1]
\Require Network $(\mathcal{V},\mathcal{E})$, node labels $l_i\in\{a,b\}$, privacy parameter $\varepsilon_\ell,\varepsilon_e$.
\Ensure Differentially private connectedness index $\widehat C^{\mathcal{A}\to\mathcal{B}}_{\mathrm{DP}}$

\State $p \gets \dfrac{1}{1+e^{\varepsilon_\ell}}$

\For{each node $i\in\mathcal{V}$} \Comment{Randomized response on labels}
    \State $\hat{l}_i \gets 
    \begin{cases}
    \text{the opposite value in }\{a,b\}, & \text{with probability } p,\\
    l_i, & \text{with probability } 1-p
    \end{cases}$
\EndFor

\For{each node $i\in\mathcal{V}$} \Comment{Compute debiased weights and individual connectedness}
    \State $\hat{\rho}_i \gets 
    \begin{cases}
    \dfrac{1}{d_i}\sum_{j \in N(i)} \1\{\hat{l}_j=b\}, & d_i>0,\\
    0, & d_i=0
    \end{cases}$
    \State $w_i \gets \dfrac{\1\{\hat{l}_i=a\}-p}{1-2p}$
    \State $\tilde\rho_i \gets
    \dfrac{\hat\rho_i-p\mathbf 1\{d_i>0\}}{1-2p}$
\EndFor

\State $S_{0,n} \gets \sum_{i\in\mathcal{V}} w_i$
\State $S_{1,n} \gets \sum_{i\in\mathcal{V}} w_i\,\tilde{\rho}_i$
\State $B_n \gets \max\{S_{0,n},1\}$
\State $\widehat C^{\mathcal{A}\to\mathcal{B}}_{\mathrm{DP}}
\gets \dfrac{{S}_{1,n}}{B_n} + Z_n,\quad
Z_n \sim \mathrm{Lap}\!\left(0,\dfrac{2(1-p)}{(1-2p)^2\,\varepsilon_e\, B_n}\right)$
\Comment{Privatized H\'ajek estimator}

\State \Return $\widehat C^{\mathcal{A}\to\mathcal{B}}_{\mathrm{DP}}$

\end{algorithmic}
\end{algorithm}

Based on Algorithm~\ref{alg:dp-connectedness}, we now state the overall privacy guarantee of the proposed procedure.
The result follows directly from \Cref{thm:edgeadj-composition-epsdelta},
together with the $\varepsilon_\ell$-DP property of the randomized response
and the $\varepsilon_e$ edge-DP calibration of the Laplace mechanism (proved in \cref{app:proof:thm:privacy-guarantee}).

\begin{theorem}[Privacy guarantee of Algorithm~\ref{alg:dp-connectedness}]\label{thm:privacy-guarantee}
Algorithm~\ref{alg:dp-connectedness} satisfies
$(\varepsilon_\ell+\varepsilon_e)$ edge-adjacent differential privacy
with respect to the input $(\mathcal{V},\mathcal{E},\mathbf{L})$.
\end{theorem}

As established in \cref{thm:privacy-guarantee}, Algorithm~\ref{alg:dp-connectedness} provides a clear decomposition of the overall privacy budget: the label-privatization step contributes $\varepsilon_\ell$, while the final edge-private release contributes $\varepsilon_e$. By sequential composition, the entire procedure therefore satisfies $(\varepsilon_\ell+\varepsilon_e)$ edge-adjacent differential privacy.

\rev{\Cref{thm:consistency_dp} shows that the released estimator is consistent; we now describe its fluctuations. Consider a sequence of labeled networks $(\mathcal V_n,\mathcal E_n,\mathbf L_n)$, $n\geq1$, with $|\mathcal V_n|=n$, and apply \cref{alg:dp-connectedness} to each of them with fixed $\varepsilon_\ell,\varepsilon_e>0$. The networks and their true labels are treated as fixed; if they are random, the results below hold conditionally on the realized sequence. All probabilities, expectations, and covariances are taken with respect to the randomized response and the Laplace noise of \cref{alg:dp-connectedness}. For the $n$th network, write $\theta_n:=C^{\mathcal A\to\mathcal B}$ for the true index, $\hat\theta_n:=\widetilde C^{\mathcal A\to\mathcal B}$ and $\hat\theta_n^{\mathrm{DP}}:=\widehat C_{\mathrm{DP}}^{\mathcal A\to\mathcal B}$ for the stabilized and released estimators, $Z_n:=\hat\theta_n^{\mathrm{DP}}-\hat\theta_n$ for the Laplace noise, $T_n:=(S_{0,n},S_{1,n})^\top$, and $\Sigma_n:=n^{-1}\Cov(T_n)$.}

\begin{theorem}[Asymptotic normality and perturbation order]
\label{thm:relative-variance-orders}
Suppose \rev{that} \cref{ass:piA,ass:bounded-degree} hold \rev{and that, as $n\to\infty$, $\theta_n\to\theta$ and $\Sigma_n\to\Sigma$.} Then $\Var(Z_n)=\Theta(n^{-2})$ and
\[
\sqrt{n}(\hat\theta_n-\theta_n)\Rightarrow\mathcal N(0,\sigma^2),
\qquad
\sqrt{n}(\hat\theta_n^{\mathrm{DP}}-\theta_n)\Rightarrow\mathcal N(0,\sigma^2),
\]
\rev{where $\sigma^2:=\zeta^\top\Sigma\zeta/\pi_\cA^2$ and $\zeta:=(-\theta,1)^\top$.}
\end{theorem}

As established in \cref{thm:relative-variance-orders}, the dominant source of uncertainty comes from the label randomization of $\hat\theta_n$, rather than from the additional edge-private perturbation. Since the variance of the Laplace noise is of order $n^{-2}$, it does not affect the first-order asymptotic distribution of the released estimator.

\begin{remark}
\rev{The conditions $\theta_n\to\theta$ and $\Sigma_n\to\Sigma$ concern only the observed networks. For each node, consider its labeled neighborhood of radius three: the subnetwork induced by the nodes within graph distance three of it, rooted at that node, together with their labels. The quantities $\#(\mathcal A)$, $\sum_{i\in\mathcal A}\rho_i$, and $n\Sigma_n$ are sums over nodes of bounded terms that depend only on these neighborhoods (see \cref{lem:local-covariance} for $n\Sigma_n$). Under \cref{ass:bounded-degree}, there are finitely many such neighborhoods up to isomorphism, so both conditions hold whenever the proportions of nodes with each neighborhood type converge. %
}
\end{remark}

    \section{Differential Privacy for Continuous Labels: Regressing Average Friend's Rank on Own Rank}%

    While our results so far have focused on network connectedness indices with binary labels, we expect that our method of applying noise directly to sensitive attributes and calculating statistics from private input vectors will be helpful for a range of statistics relating to labeled networks. We illustrate this with applications of our approach to statistics being developed for release in forthcoming papers that make use of continuous bounded labels for nodes in the network, as opposed to the discrete labels we considered before. Specifically, these papers construct statistics relating to the average friend label among nodes with a given rank or among nodes falling within a given set of ranks. More generally, our techniques will be useful for providing differentially private answers to questions such as whether less liquid firms tend to also rely on less liquid suppliers, whether green firms tend to source materials from other green firms, or whether students display homophily on test scores.  

    \rev{Our motivation in constructing a differentially private estimator with a continuous label is to provide a formally private way to release the statistics considered by \cite{johnston2026economic}. For a range of countries, they assign each Facebook user in their sample an SES rank between 0 and 100 (corresponding to a continuous percentile). For each user, they then compute the average SES rank of that user's friends. Within regions, they then run a user-level regression of average friend SES rank on a user's own SES rank. They then construct two statistics from each of these regressions. The first measure, \textit{Upward Economic Connectedness}, is the regression's predicted average SES rank of friends for users at the 25th percentile of the national income distribution.\footnote{In principle, one could nonparametrically estimate average friend rank for individuals close to the 25th percentile. Economists have generally tended to prefer predictions from linear approximations since they are more interpretable and can be more precise for smaller cells where there are not many individuals close to a particular percentile. See \cite{chetty2014land} and \cite{chetty2026opportunity} for discussion of this point.} The second measure, \textit{EC Divide},  is the slope of the regression line. For our purposes, both statistics can be constructed from the slope and intercept from regression of average friend rank on a user's own rank, so we now turn to the problem of privately releasing the slope and intercept from such a regression. (The privacy of subsequent measures and statistics constructed from the regression line then follows via the postprocessing theorem.)}

    \subsection{Continuous Labels and Regression-Based Connectedness Estimates} 

    \rev{We again consider a labeled network $(\mathcal V,\mathcal E,\mathbf L)$, but now we let each node $i$'s attribute $x_i$ be an element of a connected and bounded interval in $\mathbb{R}$, which without loss of generality we can take to be $[0,1]$. In the \cite{johnston2026economic} paper cited above, $x_i$ is node $i$'s SES percentile in the national distribution.}

    Let $y_i$ denote the average friend rank of individual $i$. That is: 
    \begin{align*}
        y_i = \frac{\sum_{j \in N(i)} x_j}{d_i}.
    \end{align*}
    \rev{Throughout this section, friend averages $y_i$ are defined as zero when $d_i=0$. We are interested in a regression of $y_i$ on $x_i$, and specifically we want private estimators of the intercept $\alpha$ and slope $\beta$ from this regression. As before, our notion of adjacency for labeled networks is edge-adjacency from \cite{blocki2013differentially} (Definition~\ref{def:ealn}) such that we consider two labeled networks adjacent if they differ in at most one node attribute ($x_i$) and one edge. We will follow a similar conceptual approach to the case with binary labels, where we first apply noise to node characteristics, then infuse more noise to protect the presence of edges, and finally analytically debias downstream statistics. We detail our approach formally in Algorithm~\ref{alg:noisystats-bounded}. Algorithm~\ref{alg:noisystats-bounded} adapts the \texttt{NoisyStats} algorithm of \cite{alabi2022differentially} to regression on a labeled network.}

    \begin{algorithm}[H]
    \caption{Differentially Private Regression for Continuous Labels}
    \label{alg:noisystats-bounded}
    \begin{algorithmic}[1]
    \Require Labeled network $(\mathcal{V},\mathcal{E},\mathbf{L})$
    with \rev{$n\ge2$} nodes and attributes $x_i\in[0,1]$
    \Require Privacy parameters $(\varepsilon_\ell,\delta_\ell)$
    and $\varepsilon_e$
    \Ensure \rev{Differentially private regression coefficients
    $(\tilde{\alpha},\tilde{\beta})$, or $\perp$}
    
    \State Set $\Delta=1$, $\lambda=\Delta/\varepsilon_\ell$,
    \[
    A=\lambda\log\left(
    1+\frac{e^{\varepsilon_\ell}-1}{2\delta_\ell}
    \right),
    \qquad
    B=2\lambda(1-e^{-A/\lambda})
    \]
    \State \rev{Let $\sigma_z^2$ denote the variance of the
    truncated-Laplace distribution} with density
    $f(z)=B^{-1}e^{-|z|/\lambda}$ on $[-A,A]$
    
    \For{$i\in\mathcal{V}$}
        \State Sample $z_i$ \rev{independently} from $f$
        and set $\hat{x}_i\gets x_i+z_i$
    \EndFor
    
    \For{$i\in\mathcal{V}$}
        \If{$d_i>0$}
            \State $\hat{y}_i\gets
            \dfrac{1}{d_i}\sum_{j\in N(i)}\hat{x}_j$
        \Else
            \State $\hat{y}_i\gets 0$
        \EndIf
    \EndFor
    
    \State Compute
    $\bar{\hat{x}}=\dfrac{1}{n}\sum_i\hat{x}_i$,
    \quad
    $\bar{\hat{y}}=\dfrac{1}{n}\sum_i\hat{y}_i$
    
    \State Compute
    \[
    \textup{nvar}(\hat{x})=\sum_i(\hat{x}_i-\bar{\hat{x}})^2,
    \qquad
    \textup{ncov}(\hat{x},\hat{y})
    =\sum_i(\hat{x}_i-\bar{\hat{x}})(\hat{y}_i-\bar{\hat{y}})
    \]
    
    \State \rev{Set the covariance sensitivity bound:}
    \[
    \rev{\Delta_1=
    \left(\max_i|\hat{x}_i|\right)
    \max_{u\ne v}
    \left(
    |\hat{x}_u-\bar{\hat{x}}|+|\hat{x}_v-\bar{\hat{x}}|
    \right)}
    \]
    
    \State \rev{Sample
    $L_1\sim\mathrm{Lap}(0,2\Delta_1/\varepsilon_e)$}
    
    \If{\rev{$\textup{nvar}(\hat{x})>(n-1)\sigma_z^2$}}
        \State \rev{$\hat{\beta}
        =\dfrac{\textup{ncov}(\hat{x},\hat{y})+L_1}{\textup{nvar}(\hat{x})}$}
    
        \State \rev{Set
        $\Delta_2=\dfrac{2}{n}\max_i|\hat{x}_i|$}
    
        \State \rev{Sample independently
        $L_2\sim\mathrm{Lap}(0,2\Delta_2/\varepsilon_e)$}
    
        \State \rev{$\hat{\alpha}
        =(\bar{\hat{y}}-\hat{\beta}\bar{\hat{x}})+L_2$}
    
        \State \rev{Apply the errors-in-variables correction:}
        \[
        \rev{\tilde{\beta}
        =\hat{\beta}
        \frac{\textup{nvar}(\hat{x})}
        {\textup{nvar}(\hat{x})-(n-1)\sigma_z^2}}
        \]
        \[
        \rev{\tilde{\alpha}
        =\hat{\alpha}-(\tilde{\beta}-\hat{\beta})\bar{\hat{x}}}
        \]
    
        \State \Return \rev{$(\tilde{\alpha},\tilde{\beta})$}
    \Else
        \State \Return $\perp$
    \EndIf
    
    \end{algorithmic}
    \end{algorithm}

    The result below shows that our mechanism is differentially private under the edge-adjacent neighboring relation.

    \begin{theorem}[Privacy guarantee of Algorithm~\ref{alg:noisystats-bounded}]
    \label{thm:privacy-guarantee-continuous}
    Algorithm~\ref{alg:noisystats-bounded} satisfies
    $(\varepsilon_\ell+\varepsilon_e,\delta_\ell)$ edge-adjacent differential privacy
    with respect to the input $(\mathcal{V},\mathcal{E},\mathbf{L})$.
    \end{theorem}
    
    The proof follows by analyzing the two stages of the algorithm separately.
    The truncated Laplace perturbation ensures
    $(\varepsilon_\ell,\delta_\ell)$-differential privacy for node attributes,
    while the edge-dependent estimation step satisfies
    $\varepsilon_e$ edge-adjacent differential privacy under the
    global sensitivity bound.
    The result then follows from the composition
    property established in \Cref{thm:edgeadj-composition-epsdelta}.

    More specifically, our approach has two components. First, since $x_j \in [0,1]$, we privatize each individual characteristic by adding truncated Laplace noise, obtaining a privacy-protected rank $\hat{x}_j$. Following \citet{geng2020tight}, this mechanism is calibrated to the desired privacy level and keeps the perturbed characteristic bounded, unlike the standard Laplace mechanism with unbounded support. This first step protects the sensitive node labels themselves.

    Second, we compute the regression using the privatized covariates and the corresponding friend-rank outcomes, \rev{and add noise to the regression sufficient statistics following \cite{alabi2022differentially}.} The second step adds privacy protection at the regression stage; combined with the first-step perturbation of the node labels, this yields released regression coefficients that satisfy edge-adjacent differential privacy.
    
    Applying bounded noise in the first step is especially useful here, since the privacy guarantees for bounded linear regression require the covariates and outcomes to remain in a bounded range \citep{alabi2022differentially}, and the scale of the downstream noise added by Algorithm~\ref{alg:noisystats-bounded} depends on those bounds.
    
    From here, we can form an average friend rank for each individual $i$ that does not leak information about any of the $x_j$ as: 
    \begin{align*}
        \widehat{\text{AFR}}_i = \frac{\sum_{j \in N(i)} \hat{x}_j}{d_i}
    \end{align*}
    Note that these private friend ranks may also not lie between $[0,1]$ since $\hat{x}_j$ may not lie between $[0,1]$. \rev{However, because the noise is drawn from a bounded domain, the privatized friend ranks remain bounded. Conditional on the privatized labels, \cref{alg:noisystats-bounded} sets its noise scales using bounds on how much a single edge change can alter $\operatorname{ncov}(\hat x,\hat y)$ and $\bar{\hat y}$. The released coefficients are then computed from these two noisy statistics and from quantities that depend only on the privatized labels, following the sufficient-statistics approach of \citet{alabi2022differentially}.}

    Note, however, that adding noise to the $x$-variable in a regression in this way will bias the regression coefficients. To correct for this bias, suppose
    \[
    y_i=\alpha+\beta x_i+\nu_i,
    \qquad
    \E[\nu_i\mid x_i]=0,
    \]
    and let
    \[
    \hat y_i = y_i+\eta_i,
    \qquad
    \hat x_i = x_i+u_i,
    \]
    where
    \[
    \E[\eta_i\mid y_i]=0,
    \qquad
    \E[u_i\mid x_i]=0,
    \qquad
    \E[u_i^2\mid x_i]=\sigma^2.
    \]
    If $\beta^*$ denotes the slope coefficient from the regression of $\hat y$ on $\hat x$, then the standard errors-in-variables correction is
    \begin{equation}
    \tilde{\beta}
    =
    \beta^*
    \frac{\left(\frac{1}{n-1}\right)\sum_i (\hat x_i-\bar{\hat x})^2}
    {\left(\frac{1}{n-1}\right)\sum_i (\hat x_i-\bar{\hat x})^2-\sigma^2}
    \label{eq:regression-debiasing}
    \end{equation}
    \rev{which, by \cref{lem:eiv-correction}, is consistent for the true regression coefficient $\beta$ when the observations are independent and satisfy the assumptions of that lemma. The lemma also shows that the corresponding intercept estimator}
    \[
    \tilde{\alpha}=\bar{\hat y}-\tilde{\beta}\bar{\hat x}
    \]
    \rev{is consistent for $\alpha$. See also \citet{greene2012econometric} for a discussion of measurement error.} Algorithm~\ref{alg:noisystats-bounded} applies this correction to the privatized regression coefficients\rev{, and its released intercept includes Laplace noise $L_2$ for edge privacy. Network observations can share privatized labels. Under its stated assumptions, \cref{prop:noisystats} establishes consistency of the released slope in this setting (see Appendix~\ref{app:prop:debias-reg-coefficient-consistent} for the proof).}

        \rev{For a network with $n$ nodes, let $v_n:=n^{-1}\sum_i(x_i-\bar x_n)^2$, where $\bar x_n:=n^{-1}\sum_ix_i$, denote the variance of the true labels, and let $\beta_n$ denote the slope of the regression of the true friend averages $y_i$ on the true labels $x_i$, including an intercept.}

    \begin{proposition}\label{prop:noisystats}
        \rev{Suppose that \cref{ass:bounded-degree} holds, that the privacy parameters of \cref{alg:noisystats-bounded} are fixed, and that $v_n\to v>0$ and $\beta_n\to\beta$. Then the probability that \cref{alg:noisystats-bounded} returns $\perp$ tends to zero, and, on the event that it does not, the slope $\tilde\beta$ it returns satisfies $\tilde\beta\xrightarrow{p}\beta$.}
    \end{proposition}

    This proposition shows that the privacy mechanism does not affect the large-sample validity of the slope estimator. In particular, after accounting for the perturbation introduced by the truncated-Laplace mechanism and the additional noise injected by Algorithm~\ref{alg:noisystats-bounded}, the final debiased estimator remains consistent for the true regression coefficient. %

        \rev{Let $\hat x_i$, $\hat y_i$ and $\sigma_z^2$ be as in \cref{alg:noisystats-bounded}, and set
    \[
    \xi_{i,n}:=(\hat x_i,\hat y_i,\hat x_i^2,\hat x_i\hat y_i)^\top,
    \qquad
    M_n:=\sum_{i=1}^n \xi_{i,n},
    \qquad
    \bar M_n:=M_n/n,
    \qquad
    \mu_n:=\E[\bar M_n],
    \] so that $\bar M_n$ is the vector of the four sample moments of the privatized data. For $m=(m_1,m_2,m_3,m_4)$, let $h(m):=(m_4-m_1m_2)/(m_3-m_1^2-\sigma_z^2)$, and let $\tilde\beta_n:=h(\bar M_n)$ be the debiased slope computed from $\bar M_n$. %
    Let $\tilde\beta_n^{\mathrm{DP}}$ be the debiased slope returned by \cref{alg:noisystats-bounded}, which, unlike $\tilde\beta_n$, includes the Laplace noise $L_1$ that the algorithm adds for edge privacy,
    and let $\beta_n$ be the slope of the true friend averages on the true labels, defined before \cref{prop:noisystats}.} The next result shows that our estimator is asymptotically normal.

    \begin{theorem}[Asymptotic normality of the private debiased regression estimator]
    \label{thm:regression-asymp-normal}
     \rev{Suppose that \cref{ass:bounded-degree} holds, that the privacy parameters of \cref{alg:noisystats-bounded} are fixed, and that $\mu_n\to\mu$ and $n^{-1}\Cov(M_n)\to\Sigma$ for some deterministic $\mu=(\mu_1,\mu_2,\mu_3,\mu_4)^\top\in\mathbb R^4$ and $\Sigma\in\mathbb R^{4\times4}$ with $\mu_3-\mu_1^2-\sigma_z^2>0$. Then
    \[
    \sqrt n\,(\tilde\beta_n-\beta_n)\Rightarrow\mathcal N(0,\tau^2),
    \qquad
    \tau^2:=\nabla h(\mu)^\top\Sigma\,\nabla h(\mu).
    \]
    Moreover, the probability that \cref{alg:noisystats-bounded} returns $\perp$ tends to zero, and, on the event that it does not, the Laplace noise $L_1$ is asymptotically negligible, $\sqrt n\,(\tilde\beta_n^{\mathrm{DP}}-\tilde\beta_n)\xrightarrow{p}0$; hence $\sqrt n\,(\tilde\beta_n^{\mathrm{DP}}-\beta_n)\Rightarrow\mathcal N(0,\tau^2)$.}
    \end{theorem}
    Theorem~\ref{thm:regression-asymp-normal} shows that, as in the discrete-label case, the extra perturbation introduced by Algorithm~\ref{alg:noisystats-bounded} is asymptotically negligible at the root-\(n\) scale, and the estimator is asymptotically normal.
    
    We apply our method to data from the Twitch livestreaming platform in \Cref{sec:twitch} and to Amazon products frequently purchased together in \Cref{sec:appx-amazon}, and show that our private estimator performs well in recovering the true regression slope in each case. 
    
    \subsection{Mean Average Friend Ranks}

    \rev{In particular, when thinking of practical use, many downstream statistics of interest to social scientists can be constructed from the regression line. For example, as well as the measures considered by \cite{johnston2026economic}, we can construct a version of the \textit{Average Alter Wealth} measure considered by \cite{power2026social} as detailed below. We will refer to this measure as the mean average friend rank.}  

    \begin{definition}[Mean Average Friend Rank]
        \rev{Let $\mathcal{A}\subseteq[0,1]$ be a publicly specified set of ranks of positive length. Using the coefficients $\tilde{\alpha}$ and $\tilde{\beta}$ returned by Algorithm~\ref{alg:noisystats-bounded} for the network or cell of interest, define the mean average friend rank for $\mathcal{A}$ by} 
        \begin{align}
            \textup{MAFR}^\mathcal{A} = \rev{\frac{\int_\mathcal{A} (\tilde{\alpha} + \tilde{\beta} x) \, \mathrm{d}x}{\int_\mathcal{A} 1 \, \mathrm{d}x}}.
            \label{eq:MAFR}
        \end{align}
    \end{definition}
    \rev{For an individual of rank $x$, the fitted value $\tilde{\alpha}+\tilde{\beta}x$ predicts the average rank of their friends. Equation~\eqref{eq:MAFR} averages these predictions uniformly over the ranks in $\mathcal{A}$. For the bottom quartile, $\mathcal{A}=[0,1/4]$, this gives $\tilde{\alpha}+\tilde{\beta}/8$.}

 \section{Accuracy on Simulated Networks}

    \subsection{Binary Labels on Erd\H{o}s-R\'{e}nyi and Stochastic Block Models}

    To test our method, and see how its effectiveness varies depending on characteristics of the network and the privacy parameter, we apply our method to networks simulated from Erd\H{o}s-R\'{e}nyi and stochastic block models.
    
    First, we consider how the accuracy of our private estimates of connectedness (relative to the true connectedness index in the simulated network) vary as we alter the amount of homophily in the networks we simulate from stochastic block models. Figure~\ref{fig:mse-vs-homophily} shows that low levels of homophily do not significantly affect the accuracy of our private economic connectedness statistics versus the baseline case of no homophily. However, extremely high levels of homophily worsen the accuracy of our private index. This is because stochastic block models with higher degrees of homophily feature greater levels of clustering, and so the variance resulting from flips in a node's status is greater since it is likely to affect the values of $\rho_i$ for more nodes within that set. However, for all degrees of homophily, our private estimators converge to the true value of connectedness in the network as we increase $\varepsilon$, as displayed in Figure~\ref{fig:mse-vs-priv-by-homophily}.   

    In Figure~\ref{fig:accuracy-cellcomp}, we show how, for a given cell size, the accuracy of our private index varies as we change the composition of the cell. Intuitively, for a given cell size, our private index is less accurate when the cell has fewer low-SES individuals (using the names of the two sets of nodes from \cite{chettyetal2022I}), both because we are averaging the noise infused by swapping the SES of alters included in $\tilde{\rho}_i$ over fewer individuals, but also because the scale of the noise injected to protect edges is decreasing in the number of low-SES individuals. Finally, holding the \rev{expected} degree fixed in \cref{fig:accuracy-networksize}, accuracy improves rapidly with network size: The mean squared error of the private connectedness estimator decreases as $n$ grows across all privacy budgets, consistent with the shrinking edge-sensitivity and averaging of label noise.

    \begin{figure}
        \centering
        \begin{subfigure}[b]{0.48\textwidth}
            \includegraphics[width=\linewidth]{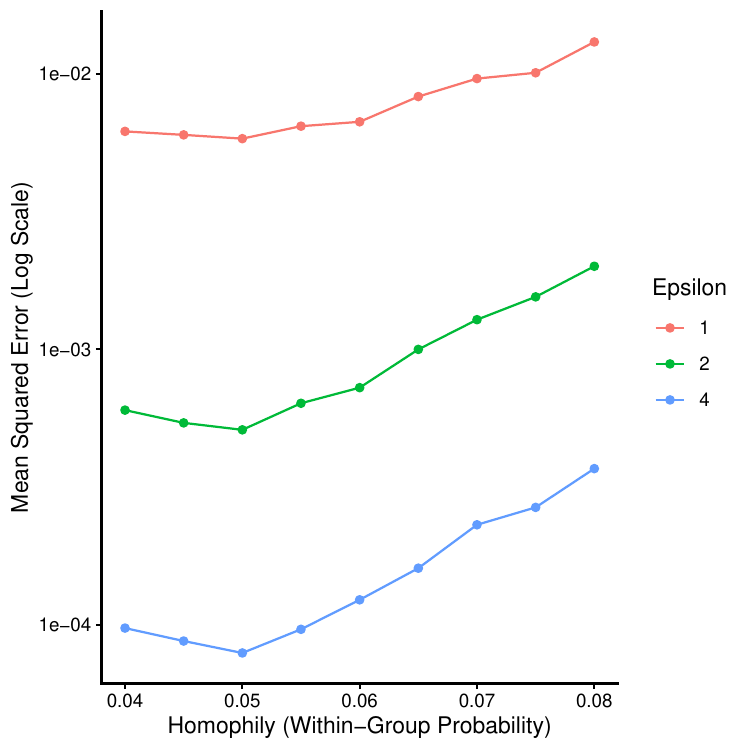}
            \caption{Mean squared error vs homophily.}
            \label{fig:mse-vs-homophily}
        \end{subfigure}
        \begin{subfigure}[b]{0.48\textwidth}
            \includegraphics[width=\linewidth]{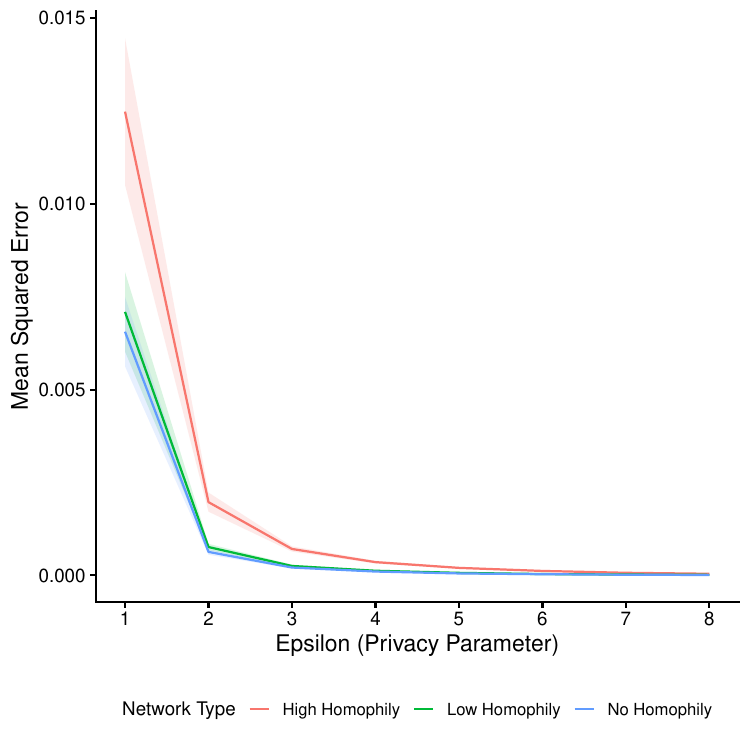}
            \caption{Mean squared error vs privacy.}
            \label{fig:mse-vs-priv-by-homophily}
        \end{subfigure}
        
        \caption{\textbf{Mean squared error vs homophily by privacy.} Panel (a): This figure illustrates the impact of homophily on the accuracy of the differentially private mechanism across \rev{three} privacy budgets ($\varepsilon \in \rev{\{1, 2, 4\}}$). In each case we split our privacy budget equally between $\varepsilon_e$ and $\varepsilon_\ell$. Our simulations are based on networks of \rev{2,000} nodes generated by a stochastic block model with two equally sized groups. To isolate the effect of community structure from network density, \rev{we keep the expected degree close to 80.} This is achieved by sweeping the within-group connection probability from 0.04 to 0.08 while simultaneously decreasing the between-group probability from 0.04 to 0 such that the sum of the within-group connection probability and the between-group connection probability is 0.08. Results are averaged over 1,125 simulation samples per data point (75 fixed graphs $\times$ 15 coupled noise seeds). The vertical axis utilizes a log scale. Panel (b): This figure plots the MSE of our differentially private mechanism as a function of the privacy budget ($\varepsilon$) for three simulated network scenarios: \textit{No Homophily}, \textit{Low Homophily}, and \textit{High Homophily}. Accuracy is calculated based on 500 samples per epsilon value, utilizing 50 pre-generated fixed graphs and 10 coupled privacy noise processes per graph to isolate the privacy-induced variance. Shaded regions represent 95\% confidence intervals. The privacy budget is split equally between $\varepsilon_e$ and $\varepsilon_\ell$. Each simulated network consists of $N=2,000$ nodes. We simulate networks using a stochastic block model with two equally sized groups. For the \textit{no homophily} case, we set the connection probability to 0.04 both within and across groups. For the \textit{low homophily} case, we set the within-group connection probability to be 0.06, and the between-group connection probability to be 0.02. For the \textit{high homophily} case, we set the within-group connection probability to be 0.08 and the between-group connection probability to be 0.}
        \label{fig:mse-homophily}
    \end{figure}

    We also show how splitting the privacy budget between $\varepsilon_e$ and $\varepsilon_\ell$ affects the MSE. For a given value of $\varepsilon_\ell + \varepsilon_e$, we see in Figure~\ref{fig:mse-varying-epsilone} that as the network grows, it becomes optimal (in terms of minimizing the MSE) to increase $\varepsilon_\ell$ relative to $\varepsilon_e$. This is consistent with \Cref{thm:relative-variance-orders}, which shows that asymptotically the variance stemming from the noise applied to labels becomes dominant.

    In Figure~\ref{fig:accuracy-networksize}, we plot MSE against the number of nodes in the network. \rev{At each privacy budget, MSE decreases as network size increases with the expected degree held at 20. This pattern is consistent with the large-sample results of \cref{thm:consistency_dp,thm:relative-variance-orders}, although the simulations do not impose the uniform bound on maximum degree in \cref{ass:bounded-degree}.} 

    \begin{figure}
        \centering
        \begin{subfigure}[b]{0.48\textwidth}
            \includegraphics[width=\linewidth]{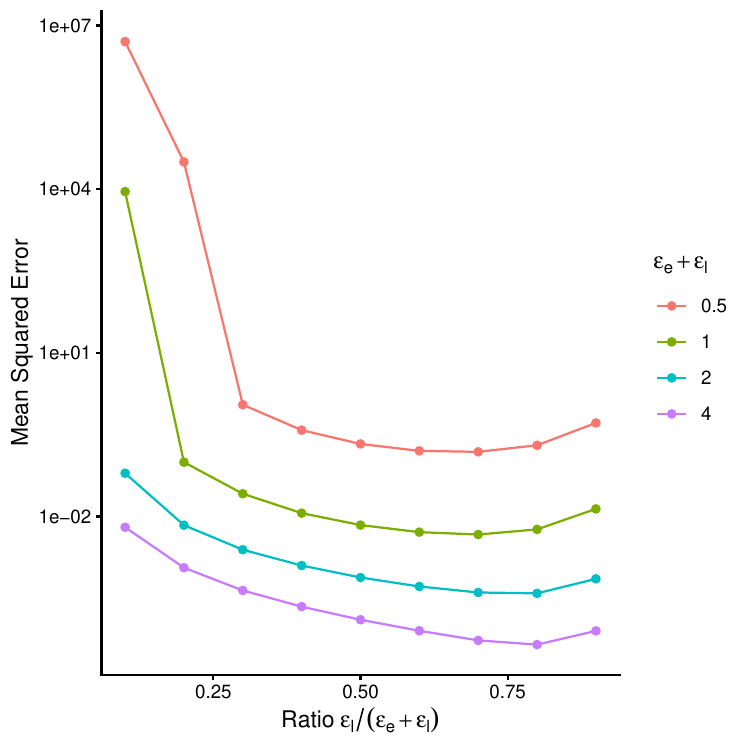}
            \caption{2,000 nodes.}
            \label{fig:mse-vs-size-2000}
        \end{subfigure}
        \begin{subfigure}[b]{0.48\textwidth}
            \includegraphics[width=\linewidth]{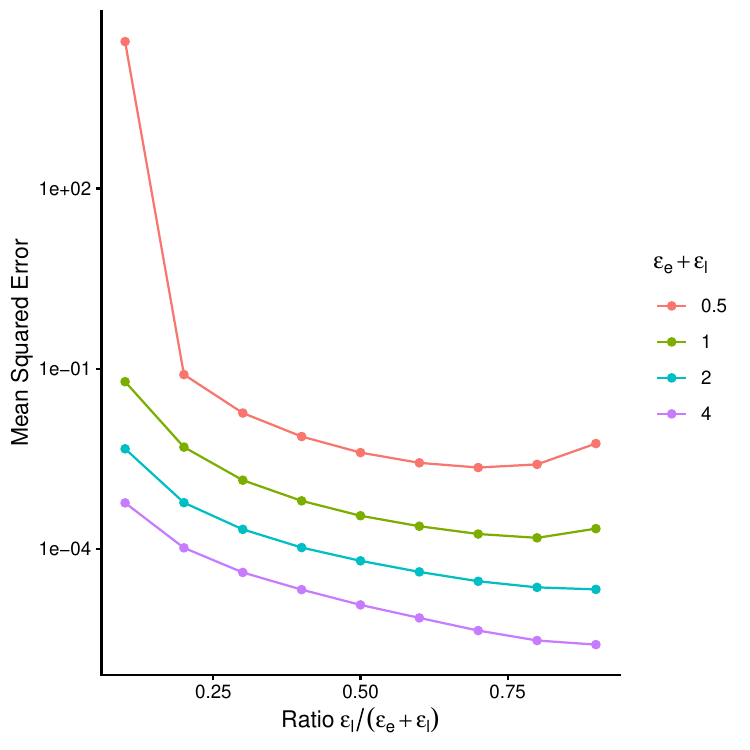}
            \caption{20,000 nodes.}
            \label{fig:mse-vs-split-20000}
        \end{subfigure}
        \caption{MSE holding $\varepsilon_e + \varepsilon_\ell$ constant.}
        \label{fig:mse-varying-epsilone}
    \end{figure}

\subsection{Continuous Labels on Graphons}

We also test our method for continuous labels on simulated networks to understand how its effectiveness varies with characteristics of the network. \rev{For the continuous-label simulations, we set $\delta_\ell=0.001$.}

We construct networks of $N$ nodes and for each node $i$ we sample $x_i$ from the uniform distribution on $[0,1]$. We then fill edges with probability:\footnote{As a result of this setup, these simulated graphs can be considered as being sampled from a graphon.}
\begin{align*}
    \Pr\{e_{ij}=1 \mid x_i,x_j\}
    =
    \frac{\Bar{d}}{(N-1)\left(\frac{2}{h}-\frac{2}{h^2}(1-e^{-h})\right)}e^{-h|x_i-x_j|}
\end{align*}
where $\Bar{d}$ is the target average degree of the network. The fraction in front is a scaling parameter that keeps the average degree of the network at $\Bar{d}$ in expectation. The term $e^{-h|x_i-x_j|}$ with homophily parameter $h \geq 0$ captures how much more likely nodes with values of $x$ close to each other are to have an edge between them versus nodes whose values of $x$ are far apart.\footnote{Note that the limit of this expression as $h \to 0^{+}$ is $\frac{\Bar{d}}{N-1}$. In our simulations, we sometimes set $h=0$, and in this case we set the corresponding edge probability between $i$ and $j$ to $\frac{\Bar{d}}{N-1}$.}

Figure~\ref{fig:continuous-sim-setup} summarizes two basic features of the simulation environment. Panel~\subref{fig:r2-homophily} shows that higher levels of the homophily parameter $h$ correspond to higher $R^2$ values in the regression of average friend rank on own rank. Panel~\subref{fig:reg-priv-allocation} shows how the split between $\varepsilon_e$ and $\varepsilon_\ell$ affects the mean squared error of the differentially private regression slope estimate for varying network sizes. As in the binary case, as the network grows, it is optimal to set $\varepsilon_\ell$ higher relative to $\varepsilon_e$, consistent with the results in \Cref{thm:regression-asymp-normal} showing that the noise to protect edges becomes asymptotically negligible compared to the effect of the noise infused into labels. 

\begin{figure}[htbp]
    \centering
    \begin{subfigure}{0.48\linewidth}
        \centering
        \includegraphics[width=\linewidth]{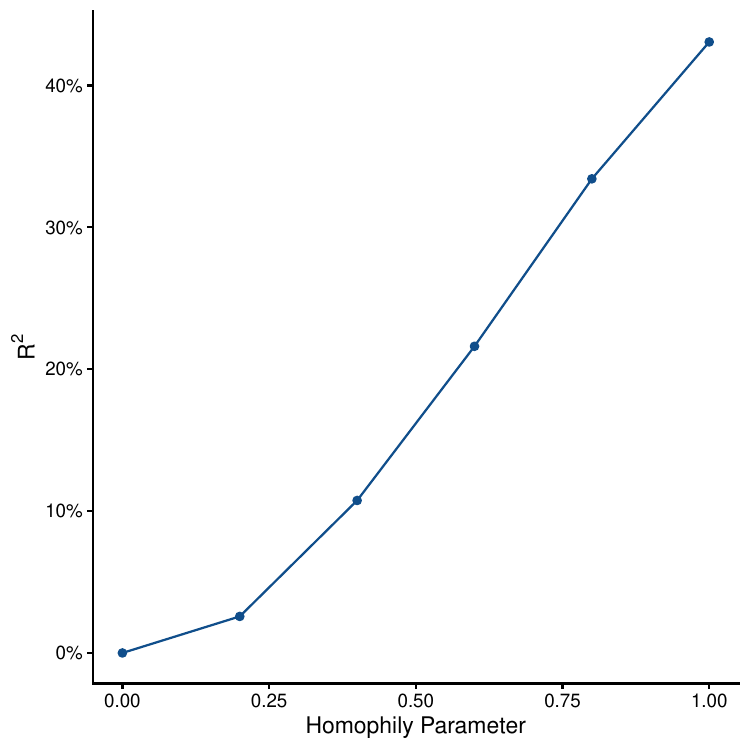}
        \caption{$R^2$ from the regression of average friend rank on own rank versus homophily.}
        \label{fig:r2-homophily}
    \end{subfigure}
    \hfill
    \begin{subfigure}{0.48\linewidth}
        \centering
        \includegraphics[width=\linewidth]{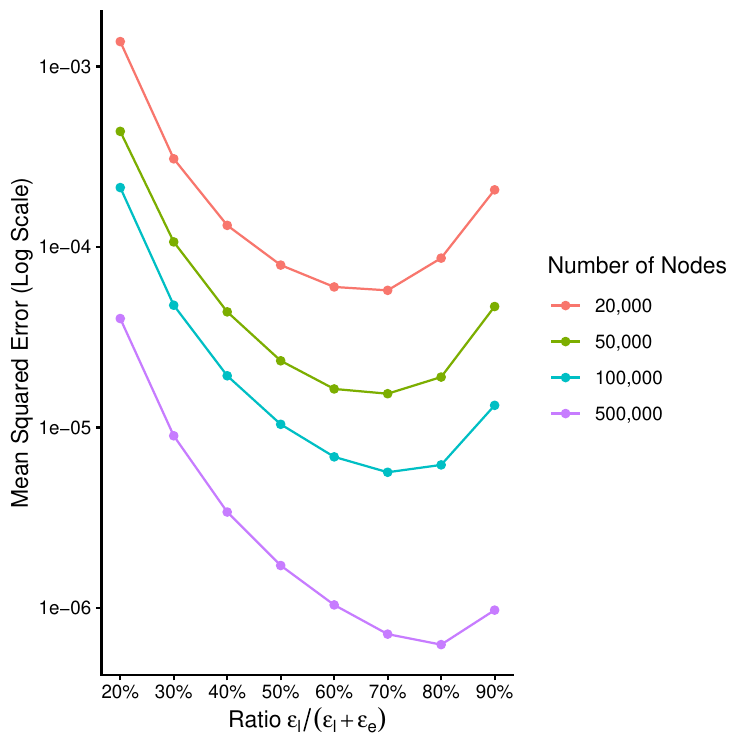}
        \caption{MSE of the differentially private regression slope estimate across privacy-budget allocations.}
        \label{fig:reg-priv-allocation}
    \end{subfigure}
    \caption{\textbf{Simulations for continuous labels.} Panel~(\subref{fig:r2-homophily}) plots the $R^2$ from regressions of average friend rank on own rank in networks generated by sampling node labels uniformly from $[0,1]$ and filling edges with probability $\frac{20}{99{,}999\left(\frac{2}{h}-\frac{2}{h^2}(1-e^{-h})\right)}e^{-h|x_i-x_j|}$, for varying levels of $h$. Each simulated network contains 100{,}000 nodes and has \rev{expected} degree 20. Panel~(\subref{fig:reg-priv-allocation}) shows results for four network sizes with homophily $h=0.8$ and expected degree $\bar d=20$, averaged over 3,000 simulations per point.}
    \label{fig:continuous-sim-setup}
\end{figure}

Figure~\ref{fig:continuous-sim-mse} shows that the mean squared error of both the regression slope and the MAFR for \rev{the bottom quartile of ranks computed from the fitted line} decreases with network size, but is relatively invariant to the degree of homophily in the network. Panel~\subref{fig:slope-mse} reports results for the regression slope, while Panel~\subref{fig:quartile-mse} reports results for bottom-quartile MAFR. \rev{For the regression slope, the decrease in MSE is consistent with the large-sample results of \cref{prop:noisystats,thm:regression-asymp-normal}. These results also assume a uniformly bounded maximum degree, while the graphon simulations hold the expected degree at 20.}

\begin{figure}[htbp]
    \centering
    \begin{subfigure}{0.48\linewidth}
        \centering
        \includegraphics[width=\linewidth]{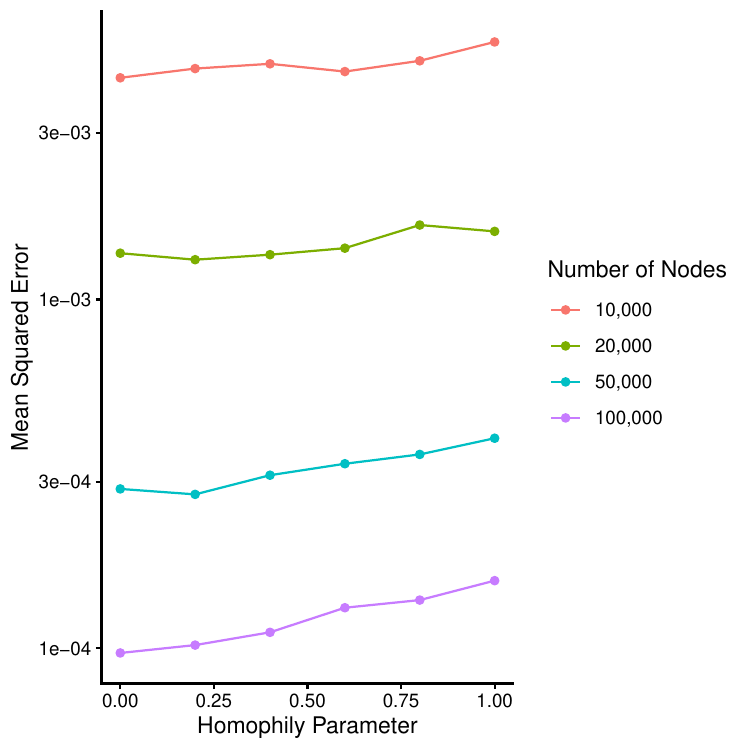}
        \caption{MSE for the regression slope by network size and homophily.}
        \label{fig:slope-mse}
    \end{subfigure}
    \hfill
    \begin{subfigure}{0.48\linewidth}
        \centering
        \includegraphics[width=\linewidth]{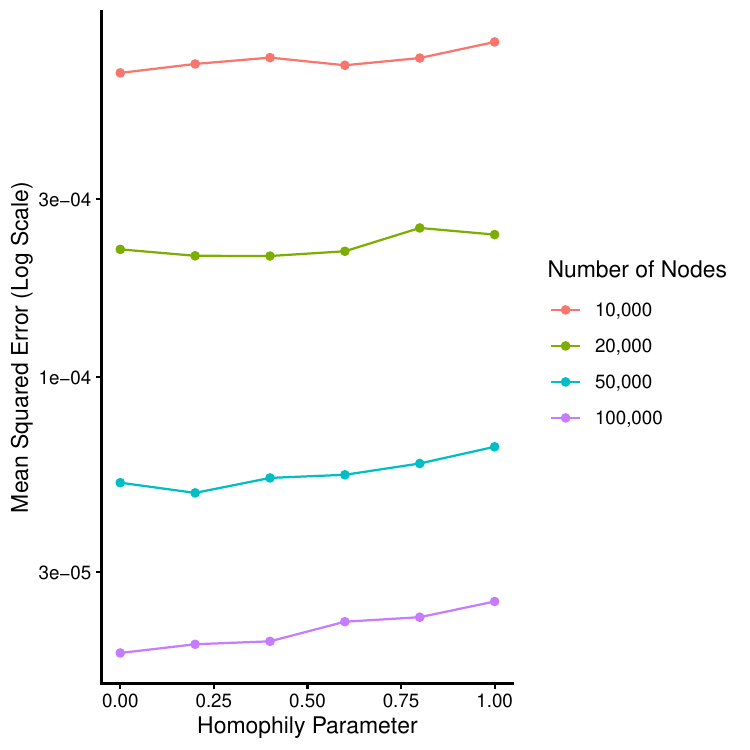}
        \caption{MSE for bottom-quartile MAFR by network size and homophily.}
        \label{fig:quartile-mse}
    \end{subfigure}
    \caption{\textbf{Estimation accuracy in the continuous-label simulations.} In both panels, we generate networks by sampling node labels uniformly from $[0,1]$ and filling edges with probability $\frac{20}{(N-1)\left(\frac{2}{h}-\frac{2}{h^2}(1-e^{-h})\right)}e^{-h|x_i-x_j|}$, for varying levels of $h$. Panel~\subref{fig:slope-mse} shows the mean squared error of our private regression slope estimate from the regression of average friend rank on own rank. Panel~\subref{fig:quartile-mse} shows the mean squared error of MAFR for \rev{the bottom quartile of ranks computed from the fitted line}. In both panels, we set $\varepsilon_e=\varepsilon_\ell=2$.}
    \label{fig:continuous-sim-mse}
\end{figure}

   \section{Empirical Applications}
\subsection{Diffusion Networks in Rural India}\label{sec:india-villages}

We demonstrate the empirical performance of our approach by applying it to the data of \cite{banerjee_diffusion_2013}---a widely used dataset containing social networks drawn from rural villages in Karnataka (a state in the south of India). These data are an ideal application and testing ground for our approach since they contain a set of relatively small, sparse, yet diverse social networks. If our approach works well on these data, we expect it to work well on most social science datasets.

The data contains information on the caste of households in the dataset. To create a binary label set, we group households into one set---which we call \textit{historically disadvantaged households}---if they identify as belonging to a scheduled caste, a scheduled tribe, or being a minority, and assign them to the other set (\textit{historically non-disadvantaged households}) if they report belonging to a forward caste or belonging to a backward caste but not a scheduled caste. We consider two households to be connected if they report a connection in any of the network layers households were asked about.

Fifty villages in the data contain information on the caste of at least some households. We drop four villages (villages 37, 41, 43, and 56) which contain six or fewer historically disadvantaged households. \rev{We exclude these villages because the target connectedness index averages over at most six historically disadvantaged households. With so few target-group households, there is little averaging of the label perturbations, and a small estimated group-size denominator can amplify both label noise and the added edge noise.} This leaves us with 46 villages to work with. We present summary statistics for our sample in Table~\ref{tab:summary-caste}. Table~\ref{tab:summary-caste} shows that in general there is a great deal of homophily on the basis of caste. Figure~\ref{fig:india-caste-results} provides three complementary illustrations of this pattern and of the empirical performance of our privacy procedure. Panel~\subref{fig:india-boxplot} shows an illustrative example for village 60, which is the largest village.

From this data, we construct cross-caste connectedness indices by calculating the average fraction of connections to historically non-disadvantaged households among historically disadvantaged households in each village. We use a privacy parameter $\varepsilon$ of 8, matching the parameter used by \cite{chettyetal2022I} for Facebook data and \cite{chetty2026opportunity} for US tax data. We calculate noisy private estimates of cross-caste connectedness \rev{1,000} times for each village, and report the standard deviation of those estimates for each village in the rightmost column, finding that it is generally low relative to the estimate of cross-caste connectedness. Across villages, the signal standard deviation is 0.12 and the mean noise standard deviation is \rev{0.03}, yielding a signal variance to noise variance ratio of \rev{17.3}. As a result, even after incorporating our formal privacy procedure, most of the variation in the dataset reflects true underlying variation in cross-caste connectedness across villages, with the variation introduced by our privacy procedure constituting a relatively small proportion of the total variation in the data. Panel~\subref{fig:india-scatter} of Figure~\ref{fig:india-caste-results} presents an example of true versus private cross-caste connectedness, showing that the true and private data are very highly correlated; larger villages tend to fall closer to the 45-degree line. Panel~\subref{fig:india-hist} shows the distribution of the correlation coefficients between true and private cross-caste connectedness across villages, illustrating that these correlations are generally very high for $\varepsilon = 8$, but begin to degrade as we move toward $\varepsilon = 4$.

\begin{figure}[htbp]
    \centering

    \begin{subfigure}{0.48\linewidth}
        \centering
        \includegraphics[width=\linewidth]{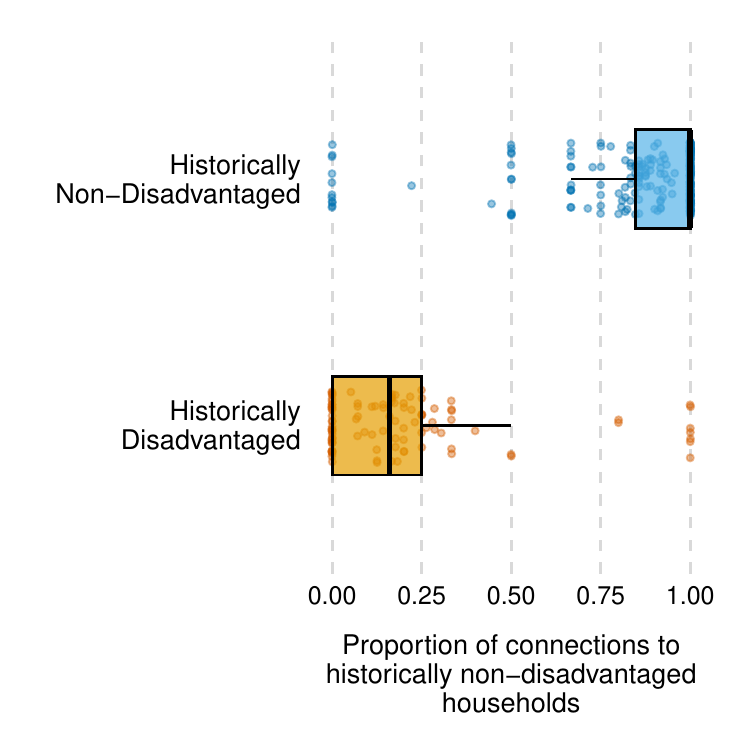}
        \caption{Proportion of connections to historically non-disadvantaged households by household status in village 60.}
        \label{fig:india-boxplot}
    \end{subfigure}

    \vspace{0.75em}

    \begin{subfigure}{0.48\linewidth}
        \centering
        \includegraphics[width=\linewidth]{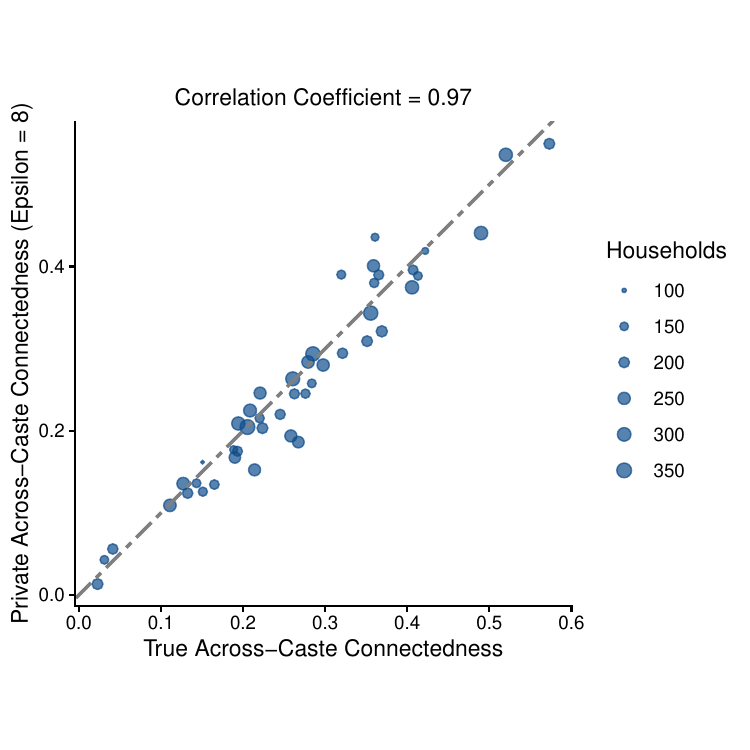}
        \caption{Example of true versus private cross-caste connectedness for one simulation at $\varepsilon = 8$.}
        \label{fig:india-scatter}
    \end{subfigure}
    \hfill
    \begin{subfigure}{0.48\linewidth}
        \centering
        \includegraphics[width=\linewidth]{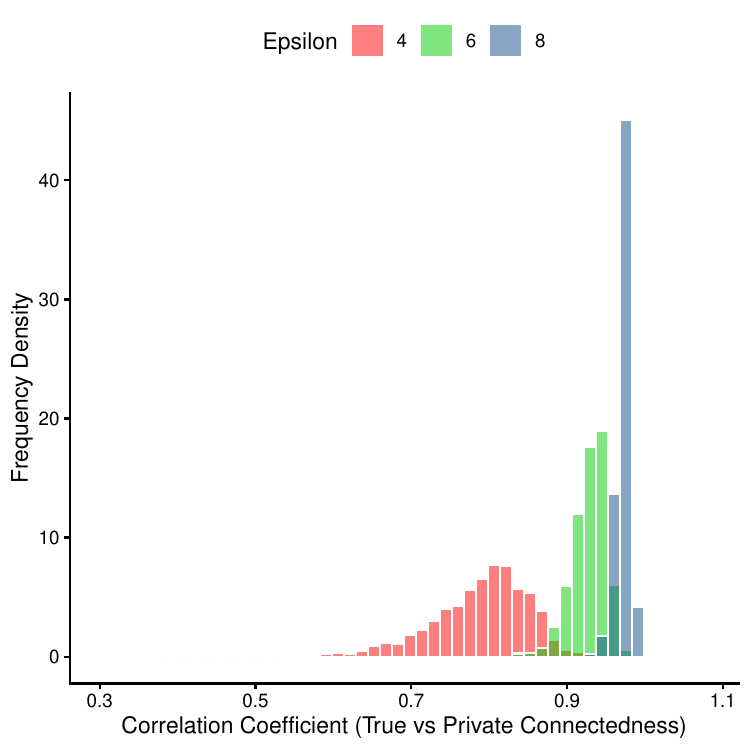}
        \caption{Distribution of correlation coefficients between true and private cross-caste connectedness across simulations, varying $\varepsilon$.}
        \label{fig:india-hist}
    \end{subfigure}

    \caption{\textbf{Cross-caste connectedness in the Karnataka village networks and the performance of the differentially private estimator.} Panel~\subref{fig:india-boxplot} displays the distribution, in village 60, of the fraction of each household's connections that are to historically non-disadvantaged households, separately by household status. Panel~\subref{fig:india-scatter} displays a village-level scatter plot of true and differentially private cross-caste connectedness for one particular simulation at $\varepsilon=8$; the dotted line is the 45-degree line through the origin. Panel~\subref{fig:india-hist} displays the distribution of Pearson correlation coefficients between true and differentially private cross-caste connectedness across simulation iterations at $\varepsilon \in \{4,6,8\}$. We allocate the privacy budget equally between $\varepsilon_e$ and $\varepsilon_\ell$.}
    \label{fig:india-caste-results}
\end{figure}

\subsection{Mutual Follows on Twitch}\label{sec:twitch}

We also apply our method for continuous labels to data from the livestreaming service Twitch, collected by \cite{rozemberczki2021twitch}. The dataset consists of around 168,000 Twitch users. Edges in the data exist if a pair of nodes both follow each other on the platform, and there are around 6.8 million edges in the dataset. We split users by language into 21 different groups. \rev{This allows us to assess how well our private estimator preserves variation in connectedness across groups, analogous to the release of connectedness statistics for different geographic areas or schools.} Each user in the dataset is associated with a view count, and we assign each user a percentile rank \rev{within their language} on the basis of this view count. We then assign each user an average friend rank on the basis of these ranks, and estimate mean average friend rank from the regression of average friend rank on own rank.

In general, users on Twitch with higher view counts tend to mutually follow users with \textit{lower} view counts, as shown for the subnetwork of English-language users in Panel~\subref{fig:twitch-binscatter} of Figure~\ref{fig:twitch-results}.

We conduct a similar exercise to our setting with binary labels, where we calculate our private estimate of the regression slope for each language and compare our private estimates to the true regression slopes. Panel~\subref{fig:twitch-scatter} of Figure~\ref{fig:twitch-results} shows one simulation of this exercise with $\varepsilon = 8$. Panel~\subref{fig:twitch-hist} shows the distribution of correlations between private and true slope coefficients across simulations for different values of $\varepsilon$. We conclude that our approach works well for \rev{$\varepsilon$ of 8} on networks with the size and structure of the Twitch data collected by \cite{rozemberczki2021twitch}.

\begin{figure}[htbp]
    \centering

    \begin{subfigure}{0.48\linewidth}
        \centering
        \includegraphics[width=\linewidth]{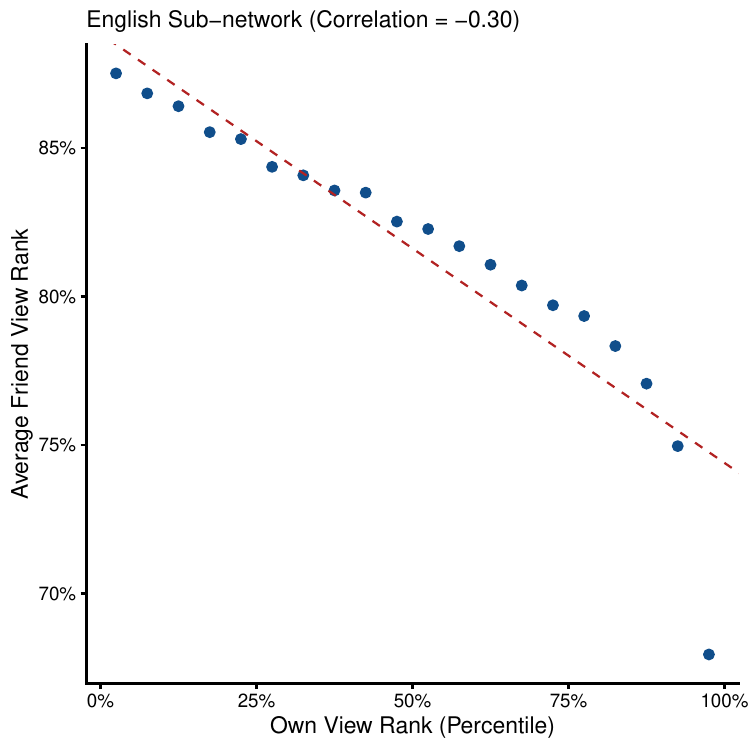}
        \caption{Average friend rank versus own rank for the English-language subnetwork.}
        \label{fig:twitch-binscatter}
    \end{subfigure}

    \vspace{0.75em}

    \begin{subfigure}{0.48\linewidth}
        \centering
        \includegraphics[width=\linewidth]{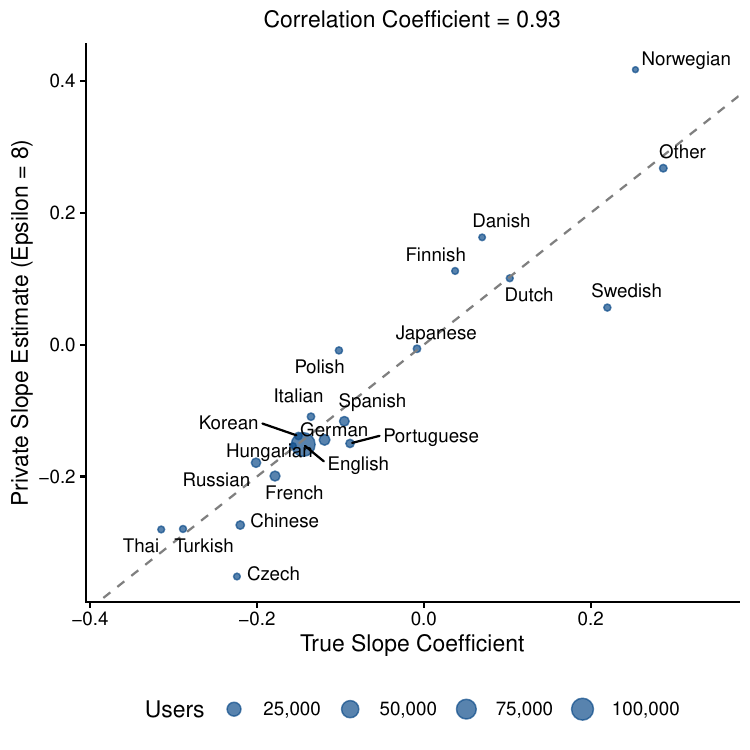}
        \caption{Private versus true slope coefficients across languages for one simulation at $\varepsilon = 8$.}
        \label{fig:twitch-scatter}
    \end{subfigure}
    \hfill
    \begin{subfigure}{0.48\linewidth}
        \centering
        \includegraphics[width=\linewidth]{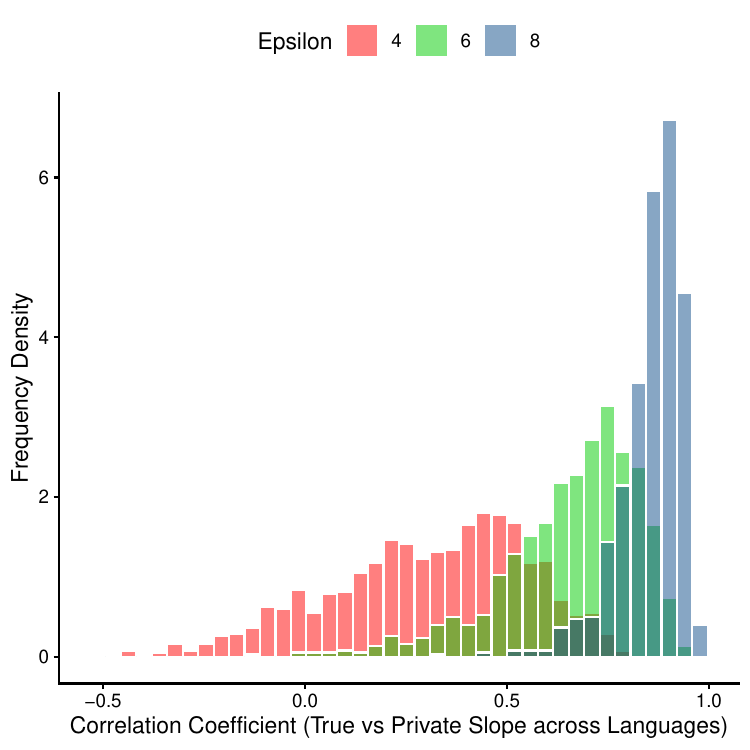}
        \caption{Distribution of correlations between private and true slope coefficients across simulations, by $\varepsilon$.}
        \label{fig:twitch-hist}
    \end{subfigure}

    \caption{\textbf{Connectedness by view-count rank in the Twitch network and the performance of the differentially private estimator.} Panel~\subref{fig:twitch-binscatter} plots average friend rank against own rank for the English-language subnetwork. Panel~\subref{fig:twitch-scatter} shows true and differentially private estimates of the regression slope across languages for one simulation at $\varepsilon = 8$; the dashed line marks equality between the private and observed slopes. Panel~\subref{fig:twitch-hist} shows the distribution of correlations between true and differentially private slope coefficients across simulation iterations for $\varepsilon \in \{4,6,8\}$, dropping any languages in a simulation if the private estimator does not return a value for that language. \rev{We use $\delta_\ell = 0.001$.}}
    \label{fig:twitch-results}
\end{figure}

In Appendix~\ref{sec:appx-amazon}, we provide another application of our method for continuous labels to data on Amazon products co-purchased together.

\section{Discussion}

     In this paper, we develop a practical method for releasing network connectedness indices under edge-adjacent differential privacy. We study settings with binary labels and continuous labels. Our approach proceeds in two steps: we first privatize node attributes directly, and then analytically debias the resulting statistics before adding a second layer of noise to protect the presence or absence of individual edges. We provide formal privacy guarantees for both procedures, establish consistency and asymptotic normality of the resulting estimators, and show in simulations and empirical applications that the method provides privacy while retaining utility even in relatively small networks. 

     Network connectedness indices allow policymakers to glean granular and hyper-local insights about the structure of connections in their area and in the institutions they govern, and our approach allows for network connectedness indices to be released while ensuring that the privacy of individuals contained in the data is not compromised. More generally, our work can make it easier to release statistics from a range of social and economic networks, or datasets that have a graph-like structure, in order to answer a range of research questions in the social sciences. 

     Our labels-first approach has two important advantages. First, it sidesteps the issue that network connectedness statistics have high global sensitivities on worst-case graphs. Second, it also sidesteps composition issues stemming from the fact that the attribute of one node can be an input to a large number of cells. Since these two issues are features of many network statistics, we think our labels-first method is a promising approach for publishing a wide range of statistics derived from labeled networks with a privacy guarantee. 
     
     \rev{Most of our applications have focused on considering one particular node attribute. Often, practitioners will be interested in releasing a set of connectedness indices covering multiple attributes for each node. (For example, \cite{chettyetal2022I} and \cite{harris2025social} construct measures of economic connectedness, age connectedness, and language connectedness, while \cite{power2026social} construct measures of economic connectedness across a number of different wealth metrics.) In these cases, our results extend naturally---researchers can run our DP mechanism on each connectedness measure in turn, and the privacy guarantee will compose additively in the standard way. In particular, this method would provide a guarantee even when all of the attributes of a single node can change simultaneously.}
     
     Extending our labels-first approach to multinomial attributes via multinomial randomized response is a natural next step. (For example, rather than just employed or unemployed, an individual may be in full-time work, part-time work, or not working. Alternatively, instead of being married or single, an individual's marital status may be married, never married, or divorced; or an individual's educational attainment may be a high school degree, a college degree, or a postgraduate degree.)  
     
     More broadly, however, there may be many network statistics that stubbornly resist attempts to privacy-protect them. \citet{chandrasekhar2024non}, for example, suggest that results related to the extent of a diffusion process (such as the spread of a virus) can often be highly sensitive to any small change in initial node characteristics or edge sets. Future work can address private estimation of network diffusion indices (e.g., measuring the proportion of infected neighbors to evaluate risk exposure in different localities) for dynamic network processes, which can be extended to higher-order network neighborhoods (e.g., exposures within two-hops or over hyper edges).  %

     \section*{Acknowledgments}

     We are grateful to Jordan Awan, Sahana Subramanyam, Sam Thau, anonymous reviewers, the organizers and participants at the 2025 Theory and Practice of Differential Privacy, the 2025 Google Research Conference on Ads Privacy, the Pitt Department of Statistics Seminar, and the 2026 Privacy and Public Policy Conference for helpful comments. Rutter acknowledges funding from the Knight-Hennessy scholarship. %

     \section*{\rev{Code and Data Availability}}

\rev{The R code and saved results for the simulations and empirical analyses
are available in the
\href{https://github.com/TomRutter42/Privacy-for-Connectedness-Indices}
{public replication repository}. The README identifies the required input
files and explains how to reproduce the figures and table. The version used
in this article is archived on Zenodo at
\href{https://doi.org/10.5281/zenodo.24386798}{10.5281/zenodo.24386798}.
Zenodo retains deposits for the lifetime of the repository.}

\rev{The Karnataka village network data collected by
\citet{banerjee_diffusion_2013} are available from
\href{https://hdl.handle.net/1902.1/21538}{Harvard Dataverse}
under the terms listed for the dataset. Harvard Dataverse preserves
deposited files for long-term access.}

\rev{The Stanford Large Network Dataset Collection provides public downloads
of the \href{https://snap.stanford.edu/data/twitch_gamers.html}
{Twitch Gamers Social Network}, the
\href{https://snap.stanford.edu/data/com-Amazon.html}
{Amazon product co-purchasing network}, and the
\href{https://snap.stanford.edu/data/amazon-meta.html}
{Amazon product metadata} used to obtain sales ranks.
These dataset pages do not specify a retention period.}

\bibliographystyle{abbrvnat-jpc}
\bibliography{refs_amin}

\appendix
\counterwithin{figure}{section}
\counterwithin{table}{section}
\numberwithin{figure}{section}
\numberwithin{table}{section}

\setcounter{figure}{0}
\setcounter{table}{0}

\section{Proofs}

\subsection{Proof of \texorpdfstring{\Cref{thm:edgeadj-composition-epsdelta}}{thmedgeadjcompositionepsdelta}: \texorpdfstring{$(\varepsilon,\delta)$}{eps-delta} composition under edge-adjacency}\label{app:proof:thm:edgeadj-composition-epsdelta}

\begin{proof}
Let $(\mathcal V,\mathcal E,\mathbf L)$ and $(\mathcal V,\mathcal E',\mathbf L')$ be edge-adjacent labeled networks. Fix any measurable set \rev{$S$ of possible outputs}. By conditioning on the output $\hat{\mathbf L}$ of $\mathcal M_1$,
\[
\Pr[\mathcal M(\mathcal V,\mathcal E,\mathbf L)\in S]
=
\E_{\hat{\mathbf L}\sim\mathcal M_1(\mathbf L)}
\left[
\Pr[\mathcal M_2(\mathcal V,\mathcal E,\hat{\mathbf L})\in S]
\right].
\]
Since $\mathbf L$ and $\mathbf L'$ differ in at most one node attribute and $\mathcal M_1$ is $(\varepsilon_\ell,\delta_\ell)$-DP with respect to node attributes, we have
\[
\Pr[\mathcal M(\mathcal V,\mathcal E,\mathbf L)\in S]
\leq
e^{\varepsilon_\ell}
\E_{\hat{\mathbf L}\sim\mathcal M_1(\mathbf L')}
\left[
\Pr[\mathcal M_2(\mathcal V,\mathcal E,\hat{\mathbf L})\in S]
\right]
+\delta_\ell.
\]

Since $\mathcal E$ and $\mathcal E'$ differ in at most one edge and $\mathcal M_2$ is $(\varepsilon_e,\delta_e)$-DP with respect to edges for any fixed $\hat{\mathbf L}$,
\[
\Pr[\mathcal M_2(\mathcal V,\mathcal E,\hat{\mathbf L})\in S]
\leq
e^{\varepsilon_e}
\Pr[\mathcal M_2(\mathcal V,\mathcal E',\hat{\mathbf L})\in S]
+\delta_e.
\]
Substituting this bound into the previous inequality gives
\begin{align*}
\Pr[\mathcal M(\mathcal V,\mathcal E,\mathbf L)\in S]
&\leq
e^{\varepsilon_\ell+\varepsilon_e}
\E_{\hat{\mathbf L}\sim\mathcal M_1(\mathbf L')}
\left[
\Pr[\mathcal M_2(\mathcal V,\mathcal E',\hat{\mathbf L})\in S]
\right]
+\delta_\ell+e^{\varepsilon_\ell}\delta_e\\
&=
e^{\varepsilon_\ell+\varepsilon_e}
\Pr[\mathcal M(\mathcal V,\mathcal E',\mathbf L')\in S]
+\delta_\ell+e^{\varepsilon_\ell}\delta_e,
\end{align*}
which proves $(\varepsilon_\ell+\varepsilon_e,\delta_\ell+e^{\varepsilon_\ell}\delta_e)$ edge-adjacent differential privacy.
\end{proof}

\subsection{Proof of \texorpdfstring{\Cref{prop:mvue}}{propmvue}: MVUE for individual connectedness}\label{app:proof:prop:mvue}
\begin{proof}
We first establish unbiasedness.  
Under randomized response,
\[
\E[\1\{\hat l_j=b\}\mid l_j]
=
(1-p)\1\{l_j=b\}+p\1\{l_j=a\}
=
(1-2p)\1\{l_j=b\}+p.
\]

Therefore,
\begin{align*}
\E[\hat\rho_i\mid \mathcal E,\mathbf L]
&=\sum_{j\in\mathcal V}a_{ij}\E[\1\{\hat l_j=b\}\mid \mathcal E,\mathbf L]\\
&=\sum_{j\in\mathcal V}a_{ij}\left((1-2p)\1\{l_j=b\}+p\right)\\
&=(1-2p)\rho_i+p\rev{\1\{d_i>0\}},
\end{align*}
\rev{where we use $\sum_j a_{ij}=\1\{d_i>0\}$.} Rearranging gives
\[
\rho_i=\frac{\E[\hat\rho_i\mid \mathcal E,\mathbf L]-p\rev{\1\{d_i>0\}}}{1-2p},
\]
which implies that $\tilde\rho_i=(\hat\rho_i-p\rev{\1\{d_i>0\}})/(1-2p)$ is unbiased.

\medskip

Next, we establish uniqueness via Lehmann--Scheff\'e.  
The privatized labels $\{\hat l_j\}$ are generated independently with likelihood
\[
f(\hat l_j\mid l_j)
=
\big((1-p)\1\{l_j=\hat l_j\}+p\1\{l_j\neq \hat l_j\}\big),
\]
so the joint likelihood factorizes as
\[
f(\hat{\mathbf{L}}\mid \mathbf{L})=\prod_{j\in\mathcal V} f(\hat l_j\mid l_j),
\]
which implies that $\hat{\mathbf{L}}$ is a sufficient statistic for $\mathbf{L}$ by Neyman’s factorization theorem.

\rev{Completeness follows from the invertibility of randomized response. The graph is fixed, so the parameter is the label vector $\mathbf L$, which ranges over all of $\{a,b\}^{\mathcal V}$. Let $h$ be a real-valued function of the privatized label vector such that $\E[h(\hat{\mathbf L})\mid\mathbf L]=0$ for every $\mathbf L$; we show by induction on the number of nodes that $h=0$. For a single node, the two conditions read
\[
(1-p)\,h(a)+p\,h(b)=0
\quad\text{(true label $a$)},
\qquad
p\,h(a)+(1-p)\,h(b)=0
\quad\text{(true label $b$)}.
\]
This linear system has determinant $(1-p)^2-p^2=1-2p\neq0$, since $p<1/2$, so $h(a)=h(b)=0$. For more than one node, single out a node $j$, write $\mathbf L_{-j}$ and $\hat{\mathbf L}_{-j}$ for the true and privatized labels of the other nodes, and for $\ell\in\{a,b\}$ let $h_\ell(\hat{\mathbf L}_{-j})$ denote the value of $h$ when the privatized label of $j$ equals $\ell$. Because the labels are privatized independently,
\[
\E[h(\hat{\mathbf L})\mid\mathbf L]
=(1-p)\,\E\bigl[h_{l_j}(\hat{\mathbf L}_{-j})\mid\mathbf L_{-j}\bigr]
+p\,\E\bigl[h_{\bar l_j}(\hat{\mathbf L}_{-j})\mid\mathbf L_{-j}\bigr],
\]
where $\bar l_j$ is the label opposite to $l_j$. Requiring this to vanish for both values of $l_j$ gives the single-node system for the two conditional expectations $\E[h_a(\hat{\mathbf L}_{-j})\mid\mathbf L_{-j}]$ and $\E[h_b(\hat{\mathbf L}_{-j})\mid\mathbf L_{-j}]$, so both vanish for every $\mathbf L_{-j}$. By the induction hypothesis, $h_a=h_b=0$, that is, $h=0$.}

Therefore, $\hat{\mathbf{L}}$ is complete and sufficient for $\mathbf{L}$, and since $\tilde{\rho}_i$ is an unbiased estimator depending on the data only through $\hat{\mathbf{L}}$, the Lehmann--Scheff\'e theorem implies that $\tilde{\rho}_i$ is the unique MVUE of $\rho_i$.
\end{proof}

\subsection{Proof of Theorem~\ref{thm:consistency_dp}: Consistency of the debiased private estimator}\label{app:proof:thm:consistency_dp}

Before establishing the consistency of our private H\'ajek ratio estimator in Theorem~\ref{thm:consistency_dp}, 
we first present the proof of \rev{\cref{lem:S1_edge_sens}} on Edge-sensitivity of the weighted sum in the numerator, followed by three additional lemmas.
\begin{proof}[Proof of Lemma~\ref{lem:S1_edge_sens}]
Toggling a single edge $(u,v)$ changes only the normalized weight rows $a_{u\cdot}$ and $a_{v\cdot}$, hence only
$\hat\rho_u,\hat\rho_v$ (and therefore only $\tilde\rho_u,\tilde\rho_v$) can change; all other $\tilde\rho_i$ remain the same.
Moreover $w_i$ depends only on $\hat l_i$, so it is unchanged.

Thus
\[
\begin{aligned}
|S_{1,n}(\mathcal E)-S_{1,n}(\mathcal E')|
&=|w_u(\tilde\rho_u(\mathcal E)-\tilde\rho_u(\mathcal E'))+w_v(\tilde\rho_v(\mathcal E)-\tilde\rho_v(\mathcal E'))|\\
&\rev{\le |w_u|\,|\tilde\rho_u(\mathcal E)-\tilde\rho_u(\mathcal E')|+|w_v|\,|\tilde\rho_v(\mathcal E)-\tilde\rho_v(\mathcal E')|.}
\end{aligned}
\]
\rev{By \eqref{eq:rho_tilde_node}, if $d_i>0$, then $\hat\rho_i\in[0,1]$ and
\[
\tilde\rho_i=\frac{\hat\rho_i-p}{1-2p}\in\left[-\frac{p}{1-2p},\frac{1-p}{1-2p}\right],
\]
while if $d_i=0$, then $\tilde\rho_i=0$, which lies in the same interval. Hence $|\tilde\rho_i(\mathcal E)-\tilde\rho_i(\mathcal E')|\le \frac{1}{1-2p}$ for $i\in\{u,v\}$, also when $u$ or $v$ is isolated in one of the two networks.}
\rev{Moreover, since $w_i=(\mathbf{1}\{\hat l_i=a\}-p)/(1-2p)$ and $p<1/2$, we have $|w_i|\leq (1-p)/(1-2p)$. Therefore,
\[
|S_{1,n}(\mathcal E)-S_{1,n}(\mathcal E')|
\leq \frac{2(1-p)}{1-2p}\cdot\frac{1}{1-2p}
=\frac{2(1-p)}{(1-2p)^2}.
\]}
\end{proof}

\begin{lemma}[Debiasing identities]\label{lem:debiasing}
Conditioning on the network $\mathcal E$ and the true label vector $\mathbf{L}=(l_i)_{i\in\mathcal V}$:
\begin{align}
\E[w_i \mid \mathcal E, \mathbf{L}] &= \1\{l_i=a\}, \label{eq:Eb_w}\\
\E[\tilde\rho_i \mid \mathcal E, \mathbf{L}] &= \rho_i. \label{eq:Eb_rhotilde}
\end{align}
Moreover, for each fixed $i$, $w_i$ depends only on $\hat l_i$ while $\tilde\rho_i$ depends on $(\hat l_j)_{j\in\mathcal V}$; under independent perturbations, conditional on $(\mathcal E,\mathbf{L})$, $w_i$ is independent of $(\hat l_j)_{j\neq i}$.
\end{lemma}

\begin{proof}
For \eqref{eq:Eb_w}, when $l_i=a$, $\Pr(\hat l_i=a)=1-p$, so
\[
\E[w_i\mid \mathcal E,l_i=a]=\frac{(1-p)-p}{1-2p}=1.
\]
When $l_i=b$, $\Pr(\hat l_i=a)=p$, so
\[
\E[w_i\mid \mathcal E,l_i=b]=\frac{p-p}{1-2p}=0.
\]

For \eqref{eq:Eb_rhotilde}, note that for any $j$,\rev{
\[
\begin{aligned}
\E[\1\{\hat l_j=b\}\mid \mathcal E,l_j]
&=
\begin{cases}
1-p, & l_j=b,\\
p, & l_j=a,
\end{cases}\\
&=p+(1-2p)\1\{l_j=b\}.
\end{aligned}
\]
}
Taking conditional expectations and using linearity,
\begin{align*}
\E[\hat\rho_i\mid \mathcal E,\mathbf L]
&=\sum_j a_{ij}\left(p+(1-2p)\1\{l_j=b\}\right)\\
&=p\1\{d_i>0\}+(1-2p)\rho_i,
\end{align*}
\rev{where we use $\sum_j a_{ij}=\1\{d_i>0\}$.} Thus,
\[
\E[\tilde\rho_i\mid \mathcal E,\mathbf L]
=\frac{\E[\hat\rho_i\mid \mathcal E,\mathbf L]-p\1\{d_i>0\}}{1-2p}
=\rho_i.
\]
\end{proof}

\begin{lemma}[Concentration of $S_{0,n}$]
Let $W:=\max\{\frac{p}{1-2p},\frac{1-p}{1-2p}\}$. If $p<\frac{1}{2}$, then $|w_i|\le W$ almost surely and
\[
\Pr\!\left(\left|\frac{S_{0,n}}{n}-\frac{\E[S_{0,n}]}{n}\right|>t\right)
\le 2\exp\!\left(-\frac{2n t^2}{(2W)^2}\right),\qquad t>0.
\]
In particular, $\frac{S_{0,n}}{n}-\frac{\E[S_{0,n}]}{n}\toP 0$.
\label{lem:S0}
\end{lemma}

\begin{proof}
Each $w_i$ is a function of $\hat l_i$ only, so $(w_i)_{i\in\mathcal V}$ are independent by the randomization in Algorithm~1. They are bounded by $W$. Apply Hoeffding's inequality to the sum $S_{0,n}=\sum_i w_i$.
\end{proof}

\begin{lemma}[Bounded differences for $S_{1,n}$]
\label{lem:bdiff}  
Let $\hat{\mathbf{L}}$ and $\hat{\mathbf{L}}'$ differ only at node $k$ (i.e., $\hat l_j=\hat l'_j$ for all $j\neq k$). \rev{Suppose $p \in (0, 1/2)$, which holds for any positive privacy parameter.} \rev{Under \cref{ass:bounded-degree} with degree bound $\Delta$,}
\[
\bigl|S_{1,n}(\hat{\mathbf{L}})-S_{1,n}(\hat{\mathbf{L}}')\bigr| \le c_k,
\]
where one may take
\begin{equation}
c_k\leq \frac{2(1-p)}{(1-2p)^2}\left(1+\sum_{i\in\mathcal V}a_{ik}\right)
\leq \frac{2(1-p)}{(1-2p)^2}(1+\Delta)=:C. \label{eq:ck}
\end{equation}
\end{lemma}

\begin{proof}
Changing $\hat l_k$ affects:
(i) the weight $w_k$ in the $k$-th summand of $S_{1,n}$; and
(ii) the indicators $\1\{\hat l_k=b\}$ inside each $\hat\rho_i$.

\medskip

First, since $|w_k|\le W$ and $|\tilde\rho_k|\le \frac{1}{1-2p}$ (because $\hat\rho_k\in[0,1]$ and $\tilde\rho_k=0$ when $d_k=0$), the change in the $k$-th summand is at most
\[
\left|w_k(\hat{\mathbf L})\tilde\rho_k(\hat{\mathbf L})
-w_k(\hat{\mathbf L}')\tilde\rho_k(\hat{\mathbf L}')\right|
\leq\frac{2(1-p)}{(1-2p)^2}.
\]

Second, for $i\neq k$, $w_i$ does not change, and only $\hat\rho_i$ changes through the single indicator at index $k$:
\[
\bigl|\hat\rho_i(\hat{\mathbf{L}})-\hat\rho_i(\hat{\mathbf{L}}')\bigr|
=
\left|\sum_{j}a_{ij}\bigl(\1\{\hat l_j=b\}-\1\{\hat l'_j=b\}\bigr)\right|
\le a_{ik}.
\]
\rev{Since the degree $d_i$ is unchanged, the correction term $p\1\{d_i>0\}$ in $\tilde\rho_i$ is also unchanged.} Therefore,
\[
\bigl|\tilde\rho_i(\hat{\mathbf{L}})-\tilde\rho_i(\hat{\mathbf{L}}')\bigr|
\le \frac{a_{ik}}{1-2p},
\quad\Rightarrow\quad
\bigl|w_i\tilde\rho_i(\hat{\mathbf{L}})-w_i\tilde\rho_i(\hat{\mathbf{L}}')\bigr|
\le \frac{Wa_{ik}}{1-2p}.
\]

Summing over $i\neq k$ gives an additional change bounded by
\[
\frac{W}{1-2p}\sum_{i\in\mathcal V}a_{ik}.
\]
Since $\sum_{i\in\mathcal V}a_{ik}\le d_k\le\Delta$ by \cref{ass:bounded-degree}, adding the $k$-th term bound yields \eqref{eq:ck}.
\end{proof}

\begin{proof}[Proof of Theorem~\ref{thm:consistency_dp}]
By Equation~\eqref{eq:Eb_w},
\[
\E[S_{0,n}\mid \mathcal E,\mathbf{L}]=\sum_i \1\{l_i=a\}=\#(\mathcal{A}),
\]
hence $\E[S_{0,n}]=\#(\mathcal{A})$ (treating $\mathcal E,\mathbf{L}$ as fixed, or taking an outer expectation if random).
Moreover,
\[
\E[S_{1,n}\mid \mathcal E,\mathbf{L}]=\sum_{i}\E[w_i\tilde\rho_i\mid \mathcal E,\mathbf{L}].
\]
\rev{Since the network has no self-loops, $w_i$ depends only on $\hat l_i$, whereas $\tilde\rho_i$ depends only on the privatized labels of other nodes. The independent randomized responses therefore make $w_i$ and $\tilde\rho_i$ conditionally independent given $(\mathcal E,\mathbf L)$. By Lemma~\ref{lem:debiasing},
\[
\begin{aligned}
\E[w_i\tilde\rho_i\mid \mathcal E,\mathbf L]
&=\E[w_i\mid \mathcal E,\mathbf L]\,
  \E[\tilde\rho_i\mid \mathcal E,\mathbf L]\\
&=\1\{l_i=a\}\rho_i.
\end{aligned}
\]}
and therefore
\[
\E[S_{1,n}\mid \mathcal E,\mathbf{L}]=\sum_{i\in\cA}\rho_i,
\qquad\Rightarrow\qquad
\E[S_{1,n}]=\sum_{i\in\cA}\rho_i.
\]

\smallskip

By Lemma~\ref{lem:S0},
\[
\frac{S_{0,n}}{n}-\frac{\E[S_{0,n}]}{n}\toP 0.
\]
Moreover, Lemma~\ref{lem:bdiff} gives a bounded-differences constant (uniformly in $k$) for the map
$\hat{\mathbf{L}}\mapsto S_{1,n}(\hat{\mathbf{L}})$, so McDiarmid's inequality implies
\[
\frac{S_{1,n}}{n}-\frac{\E[S_{1,n}]}{n}\toP 0.
\]

\smallskip
 
By Assumption~\ref{ass:piA}, $\E[S_{0,n}]/n=\#(\cA)/n\to \pi_\cA>0$.  
Together with $S_{0,n}/n\toP \E[S_{0,n}]/n$, we have $S_{0,n}/n\toP \pi_\cA$, hence
\[
\Pr\!\left(S_{0,n}/n>\pi_\cA/2\right)\to 1.
\]
\rev{Since $B_n=\max\{S_{0,n},1\}$,
\[
\frac{B_n}{n}
=\max\left\{\frac{S_{0,n}}{n},\frac{1}{n}\right\}
\toP\pi_\cA.
\]
In particular, $B_n>0$ for every realization and $B_n=S_{0,n}$ with probability tending to one.}

\smallskip

Consider the vector
\[
\left(
\frac{S_{1,n}}{n}-\frac{\E[S_{1,n}]}{n},
\frac{S_{0,n}}{n}-\frac{\E[S_{0,n}]}{n}
\right)\toP(0,0).
\]
\rev{Since $\E[S_{1,n}]=\sum_{i\in\cA}\rho_i$ and $\E[S_{0,n}]=\#(\mathcal A)$, the true connectedness index satisfies
\[
C^{\mathcal A\to\mathcal B}
=\frac{\E[S_{1,n}]}{\E[S_{0,n}]}.
\]
Moreover,
\[
\begin{aligned}
\widetilde C^{\mathcal A\to\mathcal B}
-C^{\mathcal A\to\mathcal B}
&=\frac{S_{1,n}}{B_n}
-\frac{\E[S_{1,n}]}{\E[S_{0,n}]}\\
&=\frac{S_{1,n}/n-\E[S_{1,n}]/n}{B_n/n}\\
&\quad+\frac{\E[S_{1,n}]}{n}
\left(
\frac{1}{B_n/n}
-\frac{1}{\E[S_{0,n}]/n}
\right)
\toP 0.
\end{aligned}
\]
Here we use the concentration results above, $B_n/n\toP\pi_\cA>0$, and $0\leq\E[S_{1,n}]/n\leq1$.}

\smallskip

\rev{Let $Z_n\mid S_{0,n} \sim \mathrm{Lap}\!\left(0,\frac{2(1-p)}{(1-2p)^2\varepsilon_e B_n}\right)$.}
Fix any $t>0$ and define
\[
\mathcal G_n:=\left\{S_{0,n}/n>\pi_\cA/2\right\}.
\]
Then $\Pr(\mathcal G_n)\to 1$. On $\mathcal G_n$, \rev{for all sufficiently large $n$, $B_n=S_{0,n}$, so}
\[
\rev{\frac{2(1-p)}{(1-2p)^2\varepsilon_e B_n}}
\le \frac{4(1-p)}{(1-2p)^2\varepsilon_e \pi_\cA n}.
\]
Using the Laplace tail bound $\Pr(|\mathrm{Lap}(0,b)|>t)=\exp(-t/b)$, we obtain
\begin{align*}
\Pr(|Z_n|>t)
&\le \Pr(\mathcal G_n^c) + \Pr(|Z_n|>t,\mathcal G_n) \\
&\rev{\le \Pr(\mathcal G_n^c)
    + \E\left[
      \1\{\mathcal G_n\}
      \exp\!\left(-\frac{(1-2p)^2\varepsilon_e B_n}{2(1-p)}t\right)
      \right]} \\
&\le \Pr(\mathcal G_n^c)
    + \exp\!\left(-\frac{(1-2p)^2\varepsilon_e \pi_\cA n}{4(1-p)}\,t\right)
\;\longrightarrow\; 0,
\end{align*}
and hence $Z_n\toP 0$.

\rev{
Finally, since $\widetilde C^{\mathcal A\to\mathcal B}-C^{\mathcal A\to\mathcal B}\xrightarrow{p}0$ and $Z_n\xrightarrow{p}0$, the triangle inequality gives
\[
\left|\widehat C_{\mathrm{DP}}^{\mathcal A\to\mathcal B}-C^{\mathcal A\to\mathcal B}\right|
\leq
\left|\widetilde C^{\mathcal A\to\mathcal B}-C^{\mathcal A\to\mathcal B}\right|+|Z_n|
\xrightarrow{p}0.
\]
}

\end{proof}

\subsection{Proof of Theorem~\ref{thm:privacy-guarantee}: Edge-Adjacent DP Guarantee for Binary Labels}\label{app:proof:thm:privacy-guarantee}

\begin{proof}%

Consider the input labeled network $(\mathcal{V},\mathcal{E},\mathbf{L})$.
We analyze the privacy of Algorithm~\ref{alg:dp-connectedness} by decomposing it
into two stages.

In the first stage, the algorithm applies randomized response independently
to each node label in $\mathbf{L}$ with flip probability
$p=\frac{1}{1+e^{\varepsilon_\ell}}$, producing the perturbed labels
$\hat{\mathbf{L}}$.
By \rev{Fact~\ref{fact:randomized-response}}, the first-stage mechanism $\mathcal M_1:\mathbf L\mapsto\hat{\mathbf L} $ satisfies $\varepsilon_\ell$-differential privacy with respect to changing a single node label. \rev{The mechanism depends only on the label vector and not on the edge set.}

Fix $\hat{\mathbf{L}}$, so that the weights $(w_i)_{i\in\mathcal{V}}$
are deterministic.
Let $\mathcal{E}$ and $\mathcal{E}'$ be two edge-adjacent edge sets that differ
by the addition or removal of a single undirected edge.
By \rev{Lemma~\ref{lem:S1_edge_sens}}, the statistic
\[
S_{1,n}(\mathcal{E})=\sum_{i\in\mathcal{V}} w_i\,\tilde{\rho}_i(\mathcal{E})
\]
has bounded edge sensitivity, namely
\[
\bigl|S_{1,n}(\mathcal{E})-S_{1,n}(\mathcal{E}')\bigr|
\;\le\;
\frac{2(1-p)}{(1-2p)^2}.
\]
Therefore, by the Laplace mechanism, releasing the Hájek estimator (where $B_n=\max\{S_{0,n},1\}$)
\[
\frac{\widehat S_{1,n}}{B_n}
=
\frac{S_{1,n}(\mathcal{E})}{B_n}+Z_n,
\qquad
Z_n\mid S_{0,n}\sim
\mathrm{Lap}\left(
0,\frac{2(1-p)}{(1-2p)^2\varepsilon_e B_n}
\right),
\]
is $\varepsilon_e$ edge-adjacent differentially private conditional on
$\hat{\mathbf{L}}$.
Any deterministic post-processing used to compute
$w_i$, $\tilde{\rho}_i$, $S_{0,n}$, and $S_{1,n}$ preserves this guarantee.

Finally, Algorithm~\ref{alg:dp-connectedness} is the composition of the
$\varepsilon_\ell$-DP label randomization mechanism and the
$\varepsilon_e$ edge-private release.
By \rev{Theorem~\ref{thm:edgeadj-composition-epsdelta}} (the composition theorem), the overall mechanism satisfies
$(\varepsilon_\ell+\varepsilon_e)$ edge-adjacent differential privacy
with respect to the input $(\mathcal{V},\mathcal{E},\mathbf{L})$.
\end{proof}

\subsection{Proof of Theorem~\ref{thm:relative-variance-orders}: 
Asymptotic normality and perturbation order}
\label{app:proof:thm:ratio-asymp}

\leavevmode\\

\rev{Throughout this subsection, the networks and their true labels are fixed, as in the setting described before \cref{thm:relative-variance-orders}; all probabilities and moments are taken with respect to the randomized response and the Laplace noise only. We index the nodes of the $n$th network by $\mathcal{V}_n=\{1,\ldots,n\}$.}

Let
\begin{align}
I_i&:=\1\{\hat l_i=a\},
\quad
w_i:=\frac{I_i-p}{1-2p},
\quad
r_i:=\1\{d_i>0\}, \label{eq:w_i} \\
\quad
\tilde\rho_i
&:=\frac{\hat\rho_i-pr_i}{1-2p}
=\frac{1}{1-2p}\left(\sum_{j\in N(i)}a_{ij}(1-I_j)-pr_i\right), \label{eq:rho_tilde_i}
\end{align}
where
\[
a_{ij}:=
\begin{cases}
e_{ij}/d_i,&d_i>0,\\
0,&d_i=0,
\end{cases}
\qquad
d_i:=\sum_{j\in \mathcal V_n}e_{ij}.
\]

\rev{We also write $q_i:=\E[I_i]$, so that $q_i=1-p$ if $l_i=a$ and $q_i=p$ if $l_i=b$; and in either case, $\Var(I_i)=p(1-p)$.}

\rev{For each node $i\in \mathcal V_n$, define}
\[
\rev{\xi_i}:=
\begin{pmatrix}
w_i\\
y_i
\end{pmatrix},
\qquad
y_i:=w_i\tilde\rho_i,
\qquad
T_n:=
\begin{pmatrix}
S_{0,n}\\
S_{1,n}
\end{pmatrix}
=
\sum_{i=1}^n \rev{\xi_i}.
\]

\rev{The following lemma gives the variances and covariances of the vectors $\xi_i$, with their closed-form expressions derived in the proof.
}

\begin{lemma}[\rev{Local covariance structure}]\label{lem:local-covariance}
\rev{For all nodes $i,j\in \mathcal V_n$: (i) $\Cov(\xi_i,\xi_j)=0$ whenever $\mathrm{dist}(i,j)\geq3$; and (ii) each entry of $\Cov(\xi_i,\xi_j)$ is bounded in absolute value by a constant that depends only on $p$, and is determined by the labels of the nodes within distance three of $i$ and the edges among them. Consequently, with $T_n=\sum_{i\in \mathcal V_n}\xi_i$, the matrix $\Sigma_n=n^{-1}\Cov(T_n)$ satisfies
\[
n\,\Sigma_n=\sum_{i\in \mathcal V_n}\ \sum_{j\in \mathcal V_n:\ \mathrm{dist}(i,j)\le2}\Cov(\xi_i,\xi_j),
\]
and, under \cref{ass:bounded-degree}, the entries of $\Sigma_n$ are all bounded by a constant that depends only on $p$ and $\Delta$.}
\end{lemma}

\begin{proof}
Because randomized response is applied independently across nodes, the variables
$\{I_i\}_{i=1}^n$ are independent. Note that $w_i$ depends only on $I_i$, while
$\tilde\rho_i$ depends only on $\{I_j:j\in N(i)\}$. Hence $\rev{\xi_i}$ is measurable with respect to
\[
\sigma\bigl(I_i,\{I_j:j\in N(i)\}\bigr).
\]

Therefore, if $\mathrm{dist}(i,j)\ge 3$, then the collections of randomized-response variables
used by $\rev{\xi_i}$ and $\rev{\xi_j}$ are disjoint, hence independent, and
\[
\Cov(\rev{\xi_i},\rev{\xi_j})=0,
\qquad
\mathrm{dist}(i,j)\ge 3.
\]

Consequently, we have the exact $2$-hop covariance decomposition
\begin{equation}\label{eq:2hop-cov-decomp}
\frac{1}{n}\Cov(T_n)
=
\frac{1}{n}\sum_{i=1}^n \Var(\rev{\xi_i})
+
\frac{1}{n}\sum_{i=1}^n
\sum_{\substack{j\neq i:\\ \mathrm{dist}(i,j)\le 2}}
\Cov(\rev{\xi_i},\rev{\xi_j}).
\end{equation}

\rev{We next compute the variances and covariances appearing in \eqref{eq:2hop-cov-decomp}.} For these calculations, we use the defining equation for $w_i$ and $\tilde{\rho}_i$ in \eqref{eq:w_i} and \eqref{eq:rho_tilde_i}, along with the facts that \rev{$\sum_{k\in N(i)}a_{ik}=r_i$} and  $a_{ii}=0$. 
Since the network has no self-loops and randomized response is applied independently across nodes, $I_i$ is independent of $\{I_k:k\in N(i)\}$, and thus
\[
w_i \perp\!\!\!\perp \tilde\rho_i.
\]

We write
\[
\Var(\rev{\xi_i})=
\begin{pmatrix}
\Var(w_i) & \Cov(w_i,y_i)\\
\Cov(y_i,w_i) & \Var(y_i)
\end{pmatrix}.
\]

\medskip
\noindent\textbf{(i) $\Var(w_i)$.}\par
By \eqref{eq:w_i}, %
and recalling $\Var(I_i)=p(1-p)$,
\begin{equation}\label{eq:var-wi}
\Var(w_i)=\frac{\Var(I_i)}{(1-2p)^2}=\frac{p(1-p)}{(1-2p)^2}.
\end{equation}

\medskip
\noindent\textbf{(ii) $\Cov(w_i,y_i)$.}\par
Using $y_i=w_i\tilde\rho_i$ and independence $w_i \perp\!\!\!\perp \tilde\rho_i$,
\[
\Cov(w_i,y_i)=\Cov(w_i, w_i\tilde\rho_i)
=\E[w_i^2\tilde\rho_i]-\E[w_i]\E[w_i\tilde\rho_i]
=\E[\tilde\rho_i]\Big(\E[w_i^2]-\E[w_i]^2\Big)
=\Var(w_i)\,\E[\tilde\rho_i].
\]
Next, by randomized-response debiasing, for each $k$,
\[
\E\!\left[\frac{(1-I_k)-p}{1-2p}\right]=\1\{l_k=b\},
\]
hence
\begin{equation}\label{eq:mean-rhotilde}
\E[\tilde\rho_i]
=\sum_{k\in N(i)} a_{ik}\,\1\{l_k=b\}
= \rho_i.
\end{equation}
Combining \eqref{eq:var-wi} and \eqref{eq:mean-rhotilde},
\begin{equation}\label{eq:cov-wi-yi}
\Cov(w_i,y_i)=\frac{p(1-p)}{(1-2p)^2}\,\rho_i.
\end{equation}

\medskip
\noindent\textbf{(iii) $\Var(y_i)=\Var(w_i\tilde\rho_i)$.}\par
Using again $w_i \perp\!\!\!\perp \tilde\rho_i$,
\begin{equation}\label{eq:var-product}
\Var(y_i)=\Var(w_i\tilde\rho_i)=\E[w_i^2]\Var(\tilde\rho_i)+\Var(w_i)\big(\E[\tilde\rho_i]\big)^2.
\end{equation}
We first compute $\Var(\tilde\rho_i)$. Since
\[
\tilde{\rho}_i
=\frac{1}{1-2p}
\left(
\sum_{k\in N(i)}a_{ik}(1-I_k)-pr_i
\right),
\]
the constants do not affect variance, and independence of $\{I_k:k\in N(i)\}$ yields
\begin{align}\label{eq:var-rhotilde}
\Var(\tilde\rho_i)
=\frac{1}{(1-2p)^2}\Var\!\Big(\sum_{k\in N(i)} a_{ik}(1-I_k)\Big)
&=\frac{1}{(1-2p)^2}\sum_{k\in N(i)} a_{ik}^2\,\Var(I_k)\\
&=\frac{p(1-p)}{(1-2p)^2}\sum_{k\in N(i)} a_{ik}^2. \nonumber
\end{align}
Next, compute $\E[w_i^2]$. By $\E[w_i^2]=\Var(w_i)+(\E[w_i])^2$ and randomized-response unbiasedness,
\[
\E[w_i]=\1\{l_i=a\}=:A_i,
\]
so
\begin{equation}\label{eq:second-moment-wi}
\E[w_i^2]=\frac{p(1-p)}{(1-2p)^2}+A_i.
\end{equation}
Plugging \eqref{eq:var-rhotilde}, \eqref{eq:second-moment-wi}, and \eqref{eq:mean-rhotilde}
into \eqref{eq:var-product} gives
\begin{equation}\label{eq:var-yi-closed}
\Var(y_i)
=
\left(\frac{p(1-p)}{(1-2p)^2}+A_i\right)\cdot
\frac{p(1-p)}{(1-2p)^2}\sum_{k\in N(i)} a_{ik}^2
+\frac{p(1-p)}{(1-2p)^2}\rho_i^2.
\end{equation}

From \eqref{eq:var-wi}, \eqref{eq:cov-wi-yi}, and \eqref{eq:var-yi-closed},
\[
\Var(\rev{\xi_i})=
\begin{pmatrix}
\displaystyle \frac{p(1-p)}{(1-2p)^2}
&
\displaystyle \frac{p(1-p)}{(1-2p)^2}\,\rho_i
\\[1.2ex]
\displaystyle \frac{p(1-p)}{(1-2p)^2}\,\rho_i
&
\displaystyle
\left(\frac{p(1-p)}{(1-2p)^2}+A_i\right)\cdot
\frac{p(1-p)}{(1-2p)^2}\sum_{k\in N(i)} a_{ik}^2
+\frac{p(1-p)}{(1-2p)^2}\rho_i^2
\end{pmatrix}.
\]

In particular, all entries of $\Var(\rev{\xi_i})$ are determined by the local quantities $p$, $A_i$, $\rho_i$, and $\sum_{k\in N(i)} a_{ik}^2$, each of which is bounded. Hence $\Var(\rev{\xi_i})$ is uniformly bounded in $i$ and $n$.

For $i\neq j$,
\[
\Cov(\rev{\xi_i},\rev{\xi_j})=
\begin{pmatrix}
\Cov(w_i,w_j) & \Cov(w_i,y_j)\\
\Cov(y_i,w_j) & \Cov(y_i,y_j)
\end{pmatrix}.
\]

\medskip
\noindent\textbf{(a) $\Cov(w_i,w_j)=0$ for $i\neq j$.}\par
Because $w_i$ and $w_j$ are independent,
\begin{equation}\label{eq:cov-wi-wj}
\Cov(w_i,w_j)=0,\qquad (i\neq j).
\end{equation}

\medskip
\noindent\textbf{(b) $\Cov(w_i,y_j)$ for $i\neq j$.}\par
Fix $i\neq j$. Since $y_j=w_j\tilde\rho_j$, and for $i\neq j$ the random variable $w_j$
is independent of $(w_i,\tilde\rho_j)$, we have
\[
\Cov(w_i,y_j)=\Cov(w_i,w_j\tilde\rho_j)=\E[w_j]\Cov(w_i,\tilde\rho_j).
\]
\rev{Because $\tilde\rho_j$ in \eqref{eq:rho_tilde_i} depends on $I_i$ only through the term $a_{ji}(1-I_i)/(1-2p)$, we have} $\Cov(w_i,\tilde\rho_j)\neq 0$ only if $i\in N(j)$, and in that case
\[
\Cov(w_i,\tilde\rho_j)
=
\frac{a_{ji}}{1-2p}\Cov(w_i,1-I_i)
=
-\frac{a_{ji}}{1-2p}\Cov(w_i,I_i).
\]
\rev{Moreover, since $\Cov(w_i,I_i)=\Var(I_i)/(1-2p)=p(1-p)/(1-2p)$,} it follows that
\[
\Cov(w_i,\tilde\rho_j)
=
-\frac{a_{ji}\,p(1-p)}{(1-2p)^2}.
\] From $\E[w_j]=A_j=\1\{l_j=a\}$ and because $w_j$ is independent of $(w_i,\tilde\rho_j)$, we obtain
\begin{equation}\label{eq:cov-wi-yj}
\Cov(w_i,y_j) =\E[w_j]\,\Cov(w_i,\tilde\rho_j)
=
-\,A_j\,\frac{a_{ji}\,p(1-p)}{(1-2p)^2}\,\1\{i\in N(j)\},
\end{equation}and exchanging the roles of $i$ and $j$ gives,
\begin{equation}\label{eq:cov-yi-wj}
\Cov(y_i,w_j)
=
-\,A_i\,\frac{a_{ij}\,p(1-p)}{(1-2p)^2}\,\1\{j\in N(i)\}.
\end{equation}

\medskip
\noindent\textbf{(c) $\Cov(y_i,y_j)$.}\par
To compute $\Cov(y_i,y_j)$, we use an exact polynomial expansion in the independent Bernoulli
variables $\{I_u\}$. %
\rev{Since $\sum_{k\in N(i)}a_{ik}=r_i$, the $\tilde\rho_i$ in \eqref{eq:rho_tilde_i} can be written as
\[
\tilde\rho_i=\frac{(1-p)r_i-\sum_{k\in N(i)}a_{ik}I_k}{1-2p}.
\]
Define the constants}
\[
c_0:=\frac{1}{(1-2p)^2},
\qquad
\beta:=(1-p)c_0,
\qquad
\gamma_{ik}:=pc_0a_{ik},
\qquad
\delta_{ik}:=-c_0a_{ik}.
\]
Because $a_{ii}=0$, we have $i\notin N(i)$, and expanding
$y_i=w_i\tilde{\rho}_i$ yields the identity
\begin{equation}\label{eq:y-expansion-again}
y_i
=-p(1-p)c_0r_i
+\beta r_i I_i
+\sum_{k\in N(i)}
(\gamma_{ik}+\delta_{ik}I_i)I_k.
\end{equation}
Since $\{I_u\}$ are independent, $\Cov(y_i,y_j)$ can only come from monomials that share at least
one common Bernoulli factor. On an undirected graph, for $\mathrm{dist}(i,j)\le 2$ there are two
possibilities.

\medskip
\noindent\textbf{Case A: $\mathrm{dist}(i,j)=2$.}

Assume $i\notin N(j)$ and $j\notin N(i)$, but
\[
C_{ij}:=N(i)\cap N(j)\neq\emptyset .
\]

Fix $k\in C_{ij}$.  
The terms in $y_i$ involving $I_k$ can be grouped as
\[
(\gamma_{ik}+\delta_{ik}I_i)I_k,
\]
and similarly the terms in $y_j$ involving $I_k$ equal
\[
(\gamma_{jk}+\delta_{jk}I_j)I_k .
\]

\rev{Terms with no shared Bernoulli variable contribute zero covariance.}

Since $I_i,I_j,I_k$ are mutually independent and
$(\gamma_{ik}+\delta_{ik}I_i)$,
$(\gamma_{jk}+\delta_{jk}I_j)$
are independent of $I_k$, we obtain
\[
\Cov\!\big(
(\gamma_{ik}+\delta_{ik}I_i)I_k,
(\gamma_{jk}+\delta_{jk}I_j)I_k
\big)
=
\E[\gamma_{ik}+\delta_{ik}I_i]\,
\E[\gamma_{jk}+\delta_{jk}I_j]\,
\Var(I_k).
\]

Summing over all shared neighbors gives
\[
\Cov(y_i,y_j)
=
\sum_{k\in C_{ij}}
(\gamma_{ik}+\delta_{ik}q_i)\,
(\gamma_{jk}+\delta_{jk}q_j)\,
\Var(I_k).
\]
\rev{Since $\Var(I_k)=p(1-p)$, this gives}
\begin{equation}\label{eq:cov-y-dist2-correct}
\Cov(y_i,y_j)
=
p(1-p)
\sum_{k\in N(i)\cap N(j)}
(\gamma_{ik}+\delta_{ik}q_i)\,
(\gamma_{jk}+\delta_{jk}q_j).
\end{equation}

\medskip
\noindent\textbf{Case B: $\mathrm{dist}(i,j)=1$.}

\rev{Assume $j\in N(i)$ (equivalently $i\in N(j)$). We use the same polynomial expansion \eqref{eq:y-expansion-again}, now written in centered variables. Recall $q_k=\E[I_k]$ and $\Var(I_k)=p(1-p)$, let $J_k:=I_k-q_k$, and put
\[
h_{ii}:=\beta r_i+\sum_{k\in N(i)}\delta_{ik}q_k,
\qquad
h_{ik}:=\gamma_{ik}+\delta_{ik}q_i\quad (k\in N(i)).
\]
Expanding \eqref{eq:y-expansion-again} gives
\[
y_i-\E[y_i]
=h_{ii}J_i+\sum_{k\in N(i)}h_{ik}J_k
+\sum_{k\in N(i)}\delta_{ik}J_iJ_k.
\]
Because the $J_k$ are independent and mean zero, distinct linear and quadratic monomials are orthogonal. For adjacent $i,j$, the common linear variables are $J_i$, $J_j$, and the $J_k$ for $k\in N(i)\cap N(j)$; the only common quadratic monomial is $J_iJ_j$. Thus,
\begin{equation}\label{eq:cov-y-dist1-correct-nocenter}
\begin{aligned}
\Cov(y_i,y_j)
&=h_{ii}h_{ji}\Var(I_i)+h_{ij}h_{jj}\Var(I_j)\\
&\quad+\sum_{k\in N(i)\cap N(j)}h_{ik}h_{jk}\Var(I_k)
+\delta_{ij}\delta_{ji}\Var(I_i)\Var(I_j).
\end{aligned}
\end{equation}}

\rev{The formulas \eqref{eq:var-wi}--\eqref{eq:cov-y-dist1-correct-nocenter} express every entry of $\Var(\xi_i)$ and $\Cov(\xi_i,\xi_j)$ in terms of $p$, the labels of $i$, $j$ and their neighbors, and the weights $a_{ik}$ and $a_{jk}$. When $\mathrm{dist}(i,j)\le2$, all of these are determined by the labels of the nodes within distance three of $i$ and the edges among them. Since $0\le a_{ik}\le1$ and $\sum_ka_{ik}\le1$, each entry is bounded in absolute value by a constant that depends only on $p$. Together with $\Cov(\xi_i,\xi_j)=0$ for $\mathrm{dist}(i,j)\geq3$, this proves (i) and (ii), and multiplying \eqref{eq:2hop-cov-decomp} by $n$ gives the stated formula for $n\Sigma_n$. Finally, by \cref{ass:bounded-degree}, each node has at most $1+\Delta^2$ nodes within distance two, so every entry of $\Sigma_n$ is bounded in absolute value by $1+\Delta^2$ times that constant.}

\end{proof}

\begin{proof}[Proof of \cref{thm:relative-variance-orders}]

\leavevmode\\

\rev{The proof has five steps. We write $T_n-\E[T_n]$ as a sum of martingale differences, one for each node's privatized label, and bound them (\hyperlink{step1}{Step~1}); show that their conditional covariance converges to $\Sigma$ (\hyperlink{step2}{Step~2}); deduce a central limit theorem for $T_n$ (\hyperlink{step3}{Step~3}); carry the normal limit over to the ratio $\hat\theta_n=S_{1,n}/B_n$, centered at the true index $\theta_n$  (\hyperlink{step4}{Step~4}); and show that adding the edge-private Laplace noise $Z_n$ does not change the limit, so that the released estimate $\hat\theta_n^{\mathrm{DP}}=\hat\theta_n+Z_n$ has the same normal limit (\hyperlink{step5}{Step~5}).}

\medskip
\noindent\rev{\hypertarget{step1}{}\textbf{Step 1: Martingale differences and their bound.}}

Define the filtration
\[
\mathcal F_k:=\sigma(I_1,\dots,I_k),\qquad k=0,1,\dots,n,
\]
with $\mathcal F_0$ trivial.
Define the $\mathbb R^2$-valued Doob martingale
\[
M_{n,k}:=\E[T_n\mid \mathcal F_k],\qquad k=0,1,\dots,n,
\]
and the martingale differences
\[
D_{n,k}:=M_{n,k}-M_{n,k-1},\qquad k=1,\dots,n.
\]
Then $\{(M_{n,k},\mathcal F_k)\}$ is a bivariate martingale, $M_{n,0}=\E[T_n]$, $M_{n,n}=T_n$, and
\begin{equation}\label{eq:doob-sum}
T_n-\E[T_n]=\sum_{k=1}^n D_{n,k},\qquad
\E[D_{n,k}\mid \mathcal F_{k-1}]=0.
\end{equation}

We show that the martingale differences $D_{n,k}$ are uniformly bounded given $\max_i d_i\le \Delta$. Changing the value of $I_k$ can only affect random variables that depend on $I_k$.
Since $w_i$ depends only on $I_i$, only $w_k$ is directly affected.
Moreover, $\tilde\rho_i$ depends on $\{I_j:j\in N(i)\}$,
so $I_k$ affects $\tilde\rho_i$ only if $k\in N(i)$,
that is, $i\in N(k)$.

Hence changing $I_k$ can affect $y_i=w_i\tilde\rho_i$
only for indices
\[
i\in\{k\}\cup N(k),
\]
whose cardinality is at most $1+\Delta$ by the bounded-degree assumption.

Next observe that
\[
|w_k|
=\left|\frac{I_k-p}{1-2p}\right|
\leq\frac{1}{1-2p},
\qquad
|\tilde{\rho}_i|
=\left|\frac{\hat{\rho}_i-pr_i}{1-2p}\right|
\leq\frac{1}{1-2p},
\]
since $0\leq\hat{\rho}_i\leq 1$ and
$r_i\in\{0,1\}$.

Therefore each affected summand in $(S_{0,n},S_{1,n})$ changes by at most a constant
depending only on $p$, and at most $1+\Delta$ summands can change.
Consequently there exists a constant $C=C(p,\Delta)<\infty$ such that
\[
\|T_n(I_k=1)-T_n(I_k=0)\|_\infty \le C.
\]

\rev{We now show that the martingale differences satisfy the same bound. For each possible value $z\in\{0,1\}$ of $I_k$, let
\[
H_{n,k}(z):=\E[T_n\mid\mathcal F_{k-1},I_k=z]
\]
be the best prediction of $T_n$ given $I_1,\ldots,I_{k-1}$ and $I_k=z$. 
Both $H_{n,k}(0)$ and $H_{n,k}(1)$ are functions of $I_1,\ldots,I_{k-1}$. By independence of $I_1,\ldots,I_n$, the conditional distribution of the remaining indicators $I_{k+1},\ldots,I_n$ is the same for either value of $I_k$. Taking conditional expectations in the preceding bounded-differences inequality, which holds for every value of $(I_j)_{j\neq k}$, gives
\[
\|H_{n,k}(1)-H_{n,k}(0)\|_\infty\le C.
\]
Since $\mathcal F_k$ is generated by $\mathcal F_{k-1}$ and $I_k$, the definition of $H_{n,k}$ shows that $M_{n,k}$ equals $H_{n,k}(1)$ when $I_k=1$ and $H_{n,k}(0)$ when $I_k=0$; that is, $M_{n,k}=I_kH_{n,k}(1)+(1-I_k)H_{n,k}(0)$. Taking $\E[\,\cdot\mid\mathcal F_{k-1}]$ of this identity, and using $\E[I_k\mid\mathcal F_{k-1}]=\E[I_k]=q_k$, which holds because $I_k$ is independent of $(I_1,\ldots,I_{k-1})$, gives $M_{n,k-1}=q_kH_{n,k}(1)+(1-q_k)H_{n,k}(0)$. Subtracting,
\[
D_{n,k}=(I_k-q_k)\bigl(H_{n,k}(1)-H_{n,k}(0)\bigr).
\]
Since $|I_k-q_k|\le 1$, it follows that
\[
\|D_{n,k}\|_\infty\le C\qquad\text{a.s., for all $n$ and $k$.}
\]}

\medskip
\noindent\rev{\hypertarget{step2}{}\textbf{Step 2: Convergence of the conditional covariance.}}

\rev{Write
\[
\Psi_{n,k}:=\E\!\big[D_{n,k}D_{n,k}^\top\mid \mathcal F_{k-1}\big],
\qquad
\bar\Psi_n:=\frac{1}{n}\sum_{k=1}^n \Psi_{n,k}
\]
for the conditional covariance matrix of the $k$th increment and their average, the predictable quadratic variation matrix.}

For the Doob martingale decomposition \eqref{eq:doob-sum}, the orthogonality of martingale differences implies
\begin{equation}\label{eq:var-sum-differences}
\Var(T_n-\E[T_n])
=\sum_{k=1}^n \E\!\big[D_{n,k}D_{n,k}^\top\big]
=\E\!\left[\sum_{k=1}^n \E\!\big[D_{n,k}D_{n,k}^\top\mid \mathcal F_{k-1}\big]\right].
\end{equation}
\rev{Dividing by $n$, and using $\Sigma_n=n^{-1}\Cov(T_n)$ and $\Cov(T_n)=\Var(T_n-\E[T_n])$, gives}
\[
\E[\bar\Psi_n]=\Sigma_n
\qquad\text{for every }n.
\]

To verify the conditional variance convergence required by the martingale CLT, \rev{we show that $\bar\Psi_n$ concentrates around its mean.}

\rev{Fix one entry of these $2\times2$ matrices, and denote by $\psi_{n,k}$ the corresponding entry of $\Psi_{n,k}$. By the definition of $\bar\Psi_n$, the corresponding entry of $\bar\Psi_n$ is their average, denoted by
\[
\bar\psi_n:=\frac1n\sum_{k=1}^n\psi_{n,k}.
\]} Since $\|D_{n,k}\|_\infty\le C$, each entry of $\Psi_{n,k}$ is uniformly bounded by a constant depending only on $C$.

\rev{More precisely, write $\Gamma_{n,k}:=H_{n,k}(1)-H_{n,k}(0)$. By the factorization $D_{n,k}=(I_k-q_k)\Gamma_{n,k}$ in Step~1, and because $\Gamma_{n,k}$ depends only on $I_1,\ldots,I_{k-1}$,
\[
\Psi_{n,k}=\E\bigl[(I_k-q_k)^2\mid\mathcal F_{k-1}\bigr]\,\Gamma_{n,k}\Gamma_{n,k}^\top=q_k(1-q_k)\,\Gamma_{n,k}\Gamma_{n,k}^\top,
\]
where the last equality holds because $I_k$ is independent of $I_1,\ldots,I_{k-1}$ and has variance $q_k(1-q_k)$. By \eqref{eq:y-expansion-again}, $T_n$ is a polynomial of degree at most two in the indicators, with products $I_iI_j$ only for adjacent nodes. So the change in $T_n$ when $I_k$ changes involves only the indicators of the neighbors of $k$, and averaging over the indicators not yet revealed shows that $\Gamma_{n,k}$ depends only on the earlier indicators $I_r$ with $r<k$ and $r\in N(k)$. Changing one indicator $I_r$ therefore changes $\Psi_{n,k}$ only for $k\in N(r)$, that is, for at most $\Delta$ values of $k$, by \cref{ass:bounded-degree}.}
\rev{Since each $\Psi_{n,k}$ is $\mathcal F_{k-1}$-measurable, $\bar\psi_n=n^{-1}\sum_{k=1}^n\psi_{n,k}$ is a function of $I:=(I_1,\ldots,I_n)$; we write $\bar\psi_n(I)$, and $I^{(r)}$ for the vector $I$ with $I_r$ replaced by $1-I_r$, that is, with node $r$'s privatized label switched. Each entry of $\Psi_{n,k}$ is bounded in absolute value by a constant that depends only on $C$, and switching node $r$'s label changes at most $\Delta$ of the terms $\psi_{n,k}$. Hence there exists a constant $C'<\infty$ such that, for every $r$,}
\[
\big|\rev{\bar\psi_n}(I)- \rev{\bar\psi_n}(I^{(r)})\big|\le \frac{C'}{n}.
\]

By McDiarmid's inequality, for each entry and each $t>0$,
\[
\Pr\!\left(\big|\rev{\bar\psi_n}-\E[\rev{\bar\psi_n}]\big|>t\right)
\le 2\exp(-c n t^2)
\]
for some constant $c>0$. Hence
\[
\rev{\bar\psi_n}-\E[\rev{\bar\psi_n}] \xrightarrow{p} 0.
\]
\rev{Since this holds for each of the four entries, we obtain}
\[
\bar\Psi_n-\E[\bar\Psi_n]\xrightarrow{p}0.
\]

Under the \rev{additional covariance-limit condition $\Sigma_n\to\Sigma$ in the theorem statement}, we obtain
\[
\bar\Psi_n \xrightarrow{p} \Sigma.
\]

\medskip
\noindent\rev{\hypertarget{step3}{}\textbf{Step 3: Central limit theorem for $T_n$.}}

\rev{We use the central limit theorem for martingale difference arrays \citep[][Corollary~3.1]{hall1980martingale}. It applies to increments that have mean zero given the preceding ones and that may change with $n$, as ours do because the network changes with $n$. The theorem states that if the increments satisfy the Lindeberg’s condition and their conditional variances sum to a non-random constant in probability, %
then their sum converges in distribution to a centered normal law with that variance. Of note, Corollary~3.1 of \citet{hall1980martingale} is stated for $\sigma$-fields that are nested across $n$. In our case, the limiting variance is non-random and the nested structure is not needed.\footnote{\rev{Alternatively, we may generate the privatized labels of every network from one sequence $U_1,U_2,\ldots$ of independent random variables, uniform on $(0,1)$, by setting $I_k:=\1\{U_k\le q_k\}$ for node $k$ of the $n$th network. This leaves the distribution of $T_n$ unchanged for every $n$, so the conclusion of the theorem is unaffected. The $\sigma$-fields $\mathcal G_k:=\sigma(U_1,\ldots,U_k)$ are then the same for all $n$, hence nested. Since $U_1,\ldots,U_k$ determine $I_1,\ldots,I_k$ and are independent of $I_{k+1},\ldots,I_n$, conditioning on $\mathcal G_k$ instead of $\mathcal F_k$ gives the same martingale, the same increments $D_{n,k}$ and the same conditional variances.}}}

\rev{Fix $\omega\in\mathbb R^2$. The scaled increments $\omega^\top D_{n,k}/\sqrt n$, $k=1,\ldots,n$, form a martingale difference array. By Step~1, $|\omega^\top D_{n,k}|\le\|\omega\|_1C$. Hence, for any $\varepsilon>0$ and all $n>(\|\omega\|_1C/\varepsilon)^2$, we have $|\omega^\top D_{n,k}|/\sqrt n\le\varepsilon$ for every $k$, so every term of the following Lindeberg sum
\[
\sum_{k=1}^n\E\!\left[\frac{(\omega^\top D_{n,k})^2}{n}\,\1\!\left\{\frac{|\omega^\top D_{n,k}|}{\sqrt n}>\varepsilon\right\}\,\middle|\,\mathcal F_{k-1}\right]
\]
is zero. The sum is therefore zero for all such $n$, and the conditional Lindeberg condition holds. Moreover, by Step~2, the conditional variances of the scaled increments sum to $\omega^\top\bar\Psi_n\omega\xrightarrow{p}\omega^\top\Sigma\omega$. Hence, by the martingale central limit theorem \citep[][Corollary~3.1]{hall1980martingale}, stated above,
\[
\frac{\omega^\top(T_n-\E[T_n])}{\sqrt n}\Rightarrow\mathcal N(0,\omega^\top\Sigma\omega).
\]
By the Cram\'er--Wold theorem \citep{cramer1936some}, a sequence of random vectors converges in distribution to a random vector $X$ if and only if every fixed linear combination of its coordinates converges in distribution to the same linear combination of the coordinates of $X$. If $X\sim\mathcal N(0,\Sigma)$, then $\omega^\top X\sim\mathcal N(0,\omega^\top\Sigma\omega)$, so the limit above is the distribution of $\omega^\top X$ for every $\omega$. Hence the vector itself converges:} %
\[
\frac{T_n-\E[T_n]}{\sqrt n}
=\sqrt n\Big(\frac{T_n}{n}-\E\Big[\frac{T_n}{n}\Big]\Big)
\ \Rightarrow\ \mathcal N(0,\Sigma).
\]

\medskip
\noindent\rev{\hypertarget{step4}{}\textbf{Step 4: Carrying the normal limit over from the sums $T_n=(S_{0,n},S_{1,n})^\top$ to their ratio, the stabilized estimator  $\hat\theta_n=S_{1,n}/B_n$.}}

\rev{By \cref{lem:debiasing} and the proof of \cref{thm:consistency_dp}, $\E[S_{0,n}]=\#(\cA)$ and $\E[S_{1,n}]=\sum_{i\in\cA}\rho_i=\theta_n\E[S_{0,n}]$. Let $\zeta_n:=(-\theta_n,1)^\top$, so that $\zeta_n^\top\E[T_n]=0$. Since $B_n=\max\{S_{0,n},1\}\geq1$, the identity
\begin{align}\label{eq:ratio-identity}
\sqrt n(\hat\theta_n-\theta_n)
=\frac{n}{B_n}\left(\frac{\zeta_n^\top(T_n-\E[T_n])}{\sqrt n}+\theta_n\,\frac{S_{0,n}-B_n}{\sqrt n}\right)
\end{align}
holds for every realization. By \cref{lem:S0} and \cref{ass:piA}, $S_{0,n}/n\xrightarrow{p}\pi_\cA>0$. Hence $n/B_n\xrightarrow{p}1/\pi_\cA$ and $\Pr(B_n\neq S_{0,n})=\Pr(S_{0,n}<1)\to0$, so the second term in parentheses converges to zero in probability. Moreover, $\zeta_n\to\zeta$, where $\zeta=(-\theta,1)^\top$ is as in \cref{thm:relative-variance-orders}, and $(T_n-\E[T_n])/\sqrt n$ converges in distribution by Step~3, so $(\zeta_n-\zeta)^\top(T_n-\E[T_n])/\sqrt n\xrightarrow{p}0$. Slutsky's theorem \citep[][Lemma~2.8]{vandervaart1998asymptotic} allows terms that converge in probability to constants to be replaced by those constants in a limit in distribution. Applying it to \eqref{eq:ratio-identity}, with $n/B_n$ replaced by $1/\pi_\cA$, $\zeta_n$ by $\zeta$, and the second term in parentheses by zero, and using Step~3, gives
\[
\sqrt n(\hat\theta_n-\theta_n)\Rightarrow
\mathcal N\!\left(0,\frac{\zeta^\top\Sigma\zeta}{\pi_\cA^2}\right)
=\mathcal N(0,\sigma^2),
\]
by the definition of $\sigma^2$ in \cref{thm:relative-variance-orders}.\footnote{\rev{Alternatively, we can obtain the same limit by the multivariate delta method \citep[][Theorem~3.1]{vandervaart1998asymptotic}, which linearizes a smooth function around a fixed point: if $\sqrt n(X_n-\mu)\Rightarrow\mathcal N(0,\Sigma)$ for a fixed vector $\mu$ and $g$ is differentiable at $\mu$, then $\sqrt n\bigl(g(X_n)-g(\mu)\bigr)\Rightarrow\mathcal N\bigl(0,\nabla g(\mu)^\top\Sigma\nabla g(\mu)\bigr)$. Here $X_n=T_n/n$ and $g(x,y)=y/x$, so that $g(T_n/n)=S_{1,n}/S_{0,n}$, and $\mu=(\pi_\cA,\pi_\cA\theta)^\top$ is the limit of $\E[T_n]/n$. Then $\nabla g(\mu)=\zeta/\pi_\cA$, and the delta-method variance is $\zeta^\top\Sigma\zeta/\pi_\cA^2=\sigma^2$. We use the exact identity \eqref{eq:ratio-identity} instead of the delta method because Step~3 centers $T_n/n$ at $\E[T_n]/n$, which changes with $n$, rather than at the fixed vector $\mu$.}}}

\medskip
\noindent\rev{\hypertarget{step5}{}\textbf{Step 5: Adding the edge-private Laplace noise $Z_n$ and carrying the limit over to the released estimate $\hat\theta_n^{\mathrm{DP}}=\hat\theta_n+Z_n$.}}

\rev{To complete the variance calculation, let
\[
\kappa_e:=\frac{2(1-p)}{(1-2p)^2\varepsilon_e}
\]
and use the stabilized Laplace scale $\kappa_e/B_n$, so that $Z_n\mid S_{0,n}\sim\mathrm{Lap}(0,\kappa_e/B_n)$. Since a Laplace random variable with scale parameter $b$ has variance $2b^2$, and $\E[Z_n\mid S_{0,n}]=0$, the law of total variance gives
\[
\Var(Z_n)=2\kappa_e^2\E[B_n^{-2}].
\]}

\rev{To bound $\E[B_n^{-2}]$, we separate the outcomes in which $S_{0,n}$, the debiased count of the nodes in $\cA$, is close to its typical size $n\pi_\cA$ from the others. Let}
\[
G_n:=\left\{\left|\frac{S_{0,n}}{n}-\pi_\cA\right|\le \frac{\pi_\cA}{2}\right\}
\]
\rev{be the event that $S_{0,n}/n$ is within $\pi_\cA/2$ of $\pi_\cA$. We call $G_n$ the good event: on it, the denominator $B_n$ is of order $n$. On the complementary, bad event $G_n^c$, the sum $S_{0,n}$ may be small or even negative, and $B_n$ may be as small as $1$.

We first show that the bad event is rare. By \cref{ass:piA}, $|\#(\cA)/n-\pi_\cA|\le\pi_\cA/4$ for all large $n$, and for such $n$ the triangle inequality shows that $G_n^c$ implies $|S_{0,n}/n-\#(\cA)/n|>\pi_\cA/4$. Hence \cref{lem:S0}, applied with $t=\pi_\cA/4$, gives
\[
\Pr(G_n^c)\le 2\exp\!\left(-\frac{n\pi_\cA^2}{32W^2}\right),
\]
where $W=(1-p)/(1-2p)$ is the largest possible absolute value of a debiased membership weight $w_i$ in \eqref{eq:w_i}: each $w_i$ equals either $(1-p)/(1-2p)$ or $-p/(1-2p)$ (see \cref{lem:S0}). This bound enters because Hoeffding's inequality, on which \cref{lem:S0} rests, depends on how large the summands of $S_{0,n}=\sum_iw_i$ can be. The right-hand side decays exponentially in $n$; in particular, $\Pr(G_n^c)=o(n^{-2})$.}

On the \rev{good} event $G_n$,
\[
\frac{\pi_\cA}{2}n \le S_{0,n}\le \frac{3\pi_\cA}{2}n.
\]
\rev{For sufficiently large $n$, $B_n=S_{0,n}$ on $G_n$, and therefore
\[
\frac{4}{9\pi_\cA^2 n^2}
\le \frac{1}{B_n^2}
\le \frac{4}{\pi_\cA^2 n^2}.
\]
Thus,
\[
\frac{4}{9\pi_\cA^2n^2}\Pr(G_n)
\le\E\!\left[B_n^{-2}\1(G_n)\right]
\le\frac{4}{\pi_\cA^2n^2}.
\]
Since $\Pr(G_n)\to1$, this contribution is $\Theta(n^{-2})$.}

\rev{On the bad event $G_n^c$, the denominator convention gives $B_n\ge1$, so
\[
0\le\E\!\left[B_n^{-2}\1(G_n^c)\right]
\le\Pr(G_n^c)=o(n^{-2}).
\]
Combining the two events yields
\[
\Var(Z_n)=2\kappa_e^2\E[B_n^{-2}]=\Theta(n^{-2}).
\]
Consequently, since $\E[Z_n]=0$ and $\Var(\sqrt n Z_n)=n\Var(Z_n)\to0$, Chebyshev's inequality gives $\sqrt n Z_n\xrightarrow{p}0$, and the additional edge-private Laplace perturbation is asymptotically negligible at the root-$n$ scale.}

\rev{Combining this with Step~4, Slutsky's theorem yields
\[
\sqrt n\bigl(\hat\theta_n^{\mathrm{DP}}-\theta_n\bigr)
=\sqrt n\bigl(\hat\theta_n-\theta_n\bigr)+\sqrt n Z_n
\Rightarrow\mathcal N(0,\sigma^2).
\]}

\end{proof}

\subsection{Proof of Proposition~\ref{prop:noisystats}: Consistency of Private Regression Estimators}\label{app:prop:debias-reg-coefficient-consistent}

\begin{lemma}[Errors-in-variables correction]
\label{lem:eiv-correction}
Suppose that \rev{$\{(x_i,u_i,\eta_i,\nu_i)\}_{i=1}^n$ are i.i.d. across $i$ with finite second moments, that $\Var (x_i) > 0$}, and
\[
y_i=\alpha+\beta x_i+\rev{\nu_i},
\qquad
\E[\rev{\nu_i}\mid x_i]=0.
\]
Let
\[
\hat y_i = y_i+\eta_i,
\qquad
\hat x_i = x_i+u_i,
\]
where
\[
\E[u_i]=0,\qquad
\E[\eta_i]=0,\qquad
\E[u_i^2]=\sigma^2.
\]
Assume that $x_i$, $u_i$, $\eta_i$, and $\rev{\nu_i}$ are mutually independent.

If $\beta^*$ denotes the slope coefficient from the regression of $\hat y$ on $\hat x$, then the standard errors-in-variables correction is
\[
\tilde\beta
=
\beta^*
\frac{\left(\frac{1}{n-1}\right)\sum_i (\hat x_i-\bar{\hat x})^2}
{\left(\frac{1}{n-1}\right)\sum_i (\hat x_i-\bar{\hat x})^2-\sigma^2},
\]
which is consistent for the true regression coefficient $\beta$.

The corresponding intercept estimator is
\[
\tilde\alpha=\bar{\hat y}-\tilde\beta\,\bar{\hat x},
\]
which is consistent for $\alpha$.
\end{lemma}

\begin{proof}
We first show the consistency of $\tilde{\beta}$.

\rev{The i.i.d.\ assumption and finite second moments justify the laws of large numbers used below. In particular, products such as $u_ix_i$ and $\hat x_i\nu_i$ are integrable by the Cauchy--Schwarz inequality.}

Let $\beta^*$ denote the coefficient from the regression of $\hat y$ on $\hat x$, namely
\begin{align*}
\beta^*
=
\frac{\sum_i (\hat x_i-\bar{\hat x})(\hat y_i-\bar{\hat y})}
{\sum_i (\hat x_i-\bar{\hat x})^2}.
\end{align*}

We begin with the numerator. Since
\[
y_i = \alpha + \beta x_i + \nu_i,\qquad
\hat y_i = y_i + \eta_i,\qquad
\hat x_i = x_i + u_i,
\]
we have
\[
\hat y_i = \alpha + \beta x_i + \nu_i + \eta_i.
\]
Therefore,
\begin{align*}
\sum_i (\hat x_i-\bar{\hat x})(\hat y_i-\bar{\hat y})
&=
\sum_i (\hat x_i-\bar{\hat x})(\alpha + \beta x_i + \nu_i + \eta_i - \bar{\hat y}).
\end{align*}
\begin{align*}
\sum_i (\hat x_i-\bar{\hat x})(\hat y_i-\bar{\hat y})
=
\beta \sum_i (\hat x_i-\bar{\hat x})x_i
+
\sum_i (\hat x_i-\bar{\hat x})\nu_i
+
\sum_i (\hat x_i-\bar{\hat x})\eta_i.
\end{align*}

Now note that
\[
\hat x_i-\bar{\hat x}
=
(x_i+u_i)-(\bar x+\bar u)
=
(x_i-\bar x)+(u_i-\bar u).
\]
Hence,
\begin{align*}
\sum_i (\hat x_i-\bar{\hat x})x_i
&=
\sum_i (x_i-\bar x)x_i + \sum_i (u_i-\bar u)x_i \\
&=
\sum_i (x_i-\bar x)^2 + \sum_i (u_i-\bar u)x_i.
\end{align*}
We have
\[
\frac{1}{n}\sum_i (u_i-\bar u)x_i
=
\frac{1}{n}\sum_i u_i x_i - \bar u\,\bar x.
\]
\rev{Since $u_i$ and $x_i$ are independent and $\E[u_i]=0$, we have $\E[u_ix_i]=0$.} 
By the law of large numbers,
\[
\rev{\frac{1}{n}\sum_i u_i x_i \xrightarrow{p} 0,
\qquad
\bar u \xrightarrow{p} 0,
\qquad
\bar x \xrightarrow{p} \E[x_i].}
\]
Therefore,
\[
\frac{1}{n}\sum_i (u_i-\bar u)x_i \xrightarrow{p} 0.
\]
Similarly, since $\E[\nu_i\mid x_i]=0$, we have
$\E[\nu_i]=\E[x_i\nu_i]=0$. \rev{Moreover, since $u_i$ and $\nu_i$ are independent, $\E[\nu_i\mid u_i]=\E[\nu_i]=0$, and therefore $\E[u_i\nu_i]=0$.} Thus,
\[
\rev{\E[\hat x_i\nu_i]
=\E[x_i\nu_i]+\E[u_i\nu_i]=0.}
\]
By the law of large numbers,
\[
\frac{1}{n}\sum_i(\hat x_i-\bar{\hat x})\nu_i
=
\frac{1}{n}\sum_i\hat x_i\nu_i
-\bar{\hat x}\bar\nu
\xrightarrow{p}0.
\]
\rev{Also, since $\eta_i$ is independent of $(x_i,u_i,\nu_i)$ and $\E[\eta_i]=0$, we have $\E[\hat x_i\eta_i]=0$. Moreover, $\E[\hat x_i^2]\leq 2\E[x_i^2]+2\E[u_i^2]<\infty$ by the assumed finite second moments. Therefore, by the law of large numbers,}
\begin{align*}
\frac{1}{n}\sum_i (\hat x_i-\bar{\hat x})\eta_i \rev{=\frac{1}{n}\sum_i\hat x_i\eta_i-\bar{\hat x}\bar\eta} \xrightarrow{p} 0.
\end{align*}
Therefore,
\begin{align*}
\frac{1}{n}\sum_i (\hat x_i-\bar{\hat x})(\hat y_i-\bar{\hat y})
\xrightarrow{p}
\beta\,\Var(x_i).
\end{align*}

Now consider the denominator:
\begin{align*}
\frac{1}{n}\sum_i (\hat x_i-\bar{\hat x})^2
&=
\frac{1}{n}\sum_i \bigl[(x_i-\bar x)+(u_i-\bar u)\bigr]^2 \\
&=
\frac{1}{n}\sum_i (x_i-\bar x)^2
+
\frac{2}{n}\sum_i (x_i-\bar x)(u_i-\bar u)
+
\frac{1}{n}\sum_i (u_i-\bar u)^2.
\end{align*}
By the law of large numbers, \rev{using the independence of $x_i$ and $u_i$ and the assumption $\E[u_i]=0$,}
\begin{align*}
\frac{1}{n}\sum_i (x_i-\bar x)^2 &\xrightarrow{p} \Var(x_i), \\
\frac{1}{n}\sum_i (x_i-\bar x)(u_i-\bar u) &\xrightarrow{p} 0, \\
\frac{1}{n}\sum_i (u_i-\bar u)^2 &\xrightarrow{p} \E[u_i^2]=\sigma^2.
\end{align*}
Therefore,
\begin{align*}
\frac{1}{n}\sum_i (\hat x_i-\bar{\hat x})^2
\xrightarrow{p}
\Var(x_i)+\sigma^2.
\end{align*}

Combining the numerator and denominator and applying the continuous mapping theorem, we obtain
\begin{align*}
\beta^*
\xrightarrow{p}
\frac{\beta\,\Var(x_i)}{\Var(x_i)+\sigma^2}.
\end{align*}
Moreover,
\begin{align*}
\frac{1}{n-1}\sum_i (\hat x_i-\bar{\hat x})^2
\xrightarrow{p}
\Var(x_i)+\sigma^2.
\end{align*}
Hence, applying the continuous mapping theorem again,
\begin{align*}
\tilde{\beta}
&=
\beta^*
\frac{\left(\frac{1}{n-1}\right)\sum_i (\hat x_i-\bar{\hat x})^2}
{\left(\frac{1}{n-1}\right)\sum_i (\hat x_i-\bar{\hat x})^2-\sigma^2} \\
&\xrightarrow{p}
\frac{\beta\,\Var(x_i)}{\Var(x_i)+\sigma^2}
\cdot
\frac{\Var(x_i)+\sigma^2}{\Var(x_i)}
=
\beta.
\end{align*}
This establishes the consistency of $\tilde{\beta}$.

Next we show the consistency of $\tilde{\alpha}$. Since
\[
\hat y_i = y_i+\eta_i = \alpha + \beta x_i + \nu_i + \eta_i,
\]
we have
\begin{align*}
\bar{\hat y}
=
\alpha + \beta \bar x + \bar \nu + \bar \eta.
\end{align*}
Also,
\[
\bar{\hat x}=\bar x+\bar u.
\]
Thus,
\begin{align*}
\tilde{\alpha}
&=
\bar{\hat y}-\tilde{\beta}\bar{\hat x} \\
&=
\alpha + \beta \bar x + \bar \nu + \bar \eta - \tilde{\beta}(\bar x+\bar u) \\
&=
\alpha + (\beta-\tilde{\beta})\bar x - \tilde{\beta}\bar u + \bar \nu + \bar \eta.
\end{align*}
By the law of large numbers, \rev{using $\E[u_i]=\E[\eta_i]=0$ and $\E[\nu_i]=0$ (the last following from $\E[\nu_i\mid x_i]=0$)},
\begin{align*}
\bar u \xrightarrow{p} 0,\qquad
\bar \nu \xrightarrow{p} 0,\qquad
\bar \eta \xrightarrow{p} 0.
\end{align*}
Since $\tilde{\beta}\xrightarrow{p}\beta$, it follows that
\[
\tilde{\alpha}\xrightarrow{p}\alpha.
\]
Therefore, $\tilde{\alpha}$ is consistent for $\alpha$.
\end{proof}

\rev{\Cref{lem:eiv-correction} assumes independent observations, which the network observations are not, because the friend averages of neighboring nodes share privatized inputs. The proof of \cref{prop:noisystats} therefore uses the same moment identities and the same correction as \cref{lem:eiv-correction}, but replaces the law of large numbers by a variance bound that holds under \cref{ass:bounded-degree}.}

\begin{proof}[Proof of \cref{prop:noisystats}]
We show that the slope estimate produced by Algorithm~\ref{alg:noisystats-bounded} is consistent.

\rev{We use $v_n$ and $\beta_n$ as defined before \cref{prop:noisystats}, and $v$ and $\beta$ as in its statement.}

First, \rev{write
\[
\operatorname{ncov}_n:=\operatorname{ncov}(\hat x,\hat y),
\qquad
\operatorname{nvar}_n:=\operatorname{nvar}(\hat x),
\]
for the unnormalized sample covariance and variance computed from the privatized observations $(\hat x_i,\hat y_i)$.} Let
\[
\beta^*:=\frac{\operatorname{ncov}_n}{\operatorname{nvar}_n}
\]
denote the regression coefficient computed from the bounded privatized sample
\[
\{(\hat x_i,\hat y_i)\}_{i=1}^n
\]
produced by Algorithm~\ref{alg:noisystats-bounded}. The slope estimator before debiasing is
\[
\hat\beta
=
\frac{\operatorname{ncov}_n+L_1}{\operatorname{nvar}_n},
\]
whenever $\operatorname{nvar}_n>(n-1)\sigma_z^2$, where
\[
L_1\sim\operatorname{Lap}\!\left(0,\frac{2\Delta_1}{\varepsilon_e}\right).
\]
\rev{Here
\[
\Delta_1=\Bigl(\max_i|\hat x_i|\Bigr)\max_{j\neq k}\bigl(|\hat x_j-\bar{\hat x}|+|\hat x_k-\bar{\hat x}|\bigr)
\]
is the covariance sensitivity bound set in Algorithm~\ref{alg:noisystats-bounded}, where $\bar{\hat x}$ is the mean of the $\hat x_i$.}

Since the input pairs $(\hat x_i,\hat y_i)$ are bounded, the sensitivity bound $\Delta_1$ is bounded uniformly in $n$, and hence
\[
L_1=O_p(1).
\]
\rev{We write $X_n=O_p(a_n)$ if $X_n/a_n$ is bounded in probability, that is, if for every $\varepsilon>0$ there is $M<\infty$ with $\Pr(|X_n/a_n|>M)<\varepsilon$ for all $n$, and $X_n=o_p(a_n)$ if $X_n/a_n\xrightarrow{p}0$.}

Now,
\[
\hat\beta-\beta^*
=
\frac{\operatorname{ncov}_n+L_1}{\operatorname{nvar}_n}
-
\frac{\operatorname{ncov}_n}{\operatorname{nvar}_n}
=
\frac{L_1}{\operatorname{nvar}_n}.
\]
\rev{Under \cref{ass:bounded-degree}, the bounded variables $\hat x_i$, $\hat y_i$, $\hat x_i^2$ and $\hat x_i\hat y_i$ are independent at graph distances greater than two, so each empirical moment has variance $O(n^{-1})$. Since $\E[\hat x_i]=x_i$, $\E[\hat y_i]=y_i$, $\E[\hat x_i^2]=x_i^2+\sigma_z^2$ and $\E[\hat x_i\hat y_i]=x_i y_i$ (there are no self-loops), Chebyshev's inequality gives
\[
\frac{1}{n}\operatorname{nvar}_n
\xrightarrow{p}v+\sigma_z^2>0,
\qquad
\frac{1}{n}\operatorname{ncov}_n-\beta_n v_n\xrightarrow{p}0.
\]
Consequently, in $\hat\beta-\beta^*=L_1/\operatorname{nvar}_n$ the numerator is $O_p(1)$, while $\operatorname{nvar}_n/n$ is bounded away from zero with probability tending to one, by the first limit above. It follows that
\[
\hat\beta-\beta^*=O_p(n^{-1})=o_p(1).
\]}
Moreover, since $\operatorname{nvar}_n/n\xrightarrow{p}v+\sigma_z^2$ with $v>0$,
\[
\Pr\bigl(\operatorname{nvar}_n>(n-1)\sigma_z^2\bigr)\to1,
\]
\rev{so the probability that Algorithm~\ref{alg:noisystats-bounded} returns $\perp$ tends to zero. Thus the Laplace noise $L_1$, added for edge privacy, changes the slope only by $o_p(1)$.}

\rev{Finally, since $v_n\to v$ and $\beta_n\to\beta$, the two limits above give $\beta^*=\operatorname{ncov}_n/\operatorname{nvar}_n\xrightarrow{p}\beta v/(v+\sigma_z^2)$, and hence also $\hat\beta\xrightarrow{p}\beta v/(v+\sigma_z^2)$: the noise in the privatized labels $\hat x_i$ shrinks the slope toward zero by the factor $v/(v+\sigma_z^2)$. To undo this shrinkage, \cref{alg:noisystats-bounded} applies the errors-in-variables correction, which multiplies $\hat\beta$ by $\operatorname{nvar}_n/(\operatorname{nvar}_n-(n-1)\sigma_z^2)$, an estimate of $(v+\sigma_z^2)/v$: by the first limit above, this factor converges in probability to $(v+\sigma_z^2)/v$. By the continuous mapping theorem,
\[
\tilde\beta
=\hat\beta\,\frac{\operatorname{nvar}_n}{\operatorname{nvar}_n-(n-1)\sigma_z^2}
\xrightarrow{p}\frac{\beta v}{v+\sigma_z^2}\cdot\frac{v+\sigma_z^2}{v}=\beta.
\]}

Hence $\tilde\beta$ is consistent for $\beta$.

\end{proof}

\begin{lemma}[Global sensitivity of variance and covariance on general intervals]
\label{lem:gs_nvar_ncov_general_intervals}
Let $n\ge 2$. Suppose $x,x'\in[a,b]^n$ and $y,y'\in[a',b']^n$ are neighboring databases in the sense that there exists $k\in\{1,\dots,n\}$ such that
\[
(x_i,y_i)=(x_i',y_i')\qquad \text{for all } i\neq k.
\]
Define
\[
\mathrm{nvar}(x)=\sum_{i=1}^n x_i^2-\frac{1}{n}\Big(\sum_{i=1}^n x_i\Big)^2,
\qquad
\mathrm{ncov}(x,y)=\sum_{i=1}^n x_i y_i-\frac{1}{n}\Big(\sum_{i=1}^n x_i\Big)\Big(\sum_{i=1}^n y_i\Big).
\]
Then
\[
\big|\mathrm{nvar}(x)-\mathrm{nvar}(x')\big|
\le \Big(1-\frac{1}{n}\Big)(b-a)^2,
\]
and
\[
\big|\mathrm{ncov}(x,y)-\mathrm{ncov}(x',y')\big|
\le \Big(1-\frac{1}{n}\Big)(b-a)(b'-a').
\]
Consequently,
\[
\mathrm{GS}_{\mathrm{nvar}}=\Big(1-\frac{1}{n}\Big)(b-a)^2,
\qquad
\mathrm{GS}_{\mathrm{ncov}}=\Big(1-\frac{1}{n}\Big)(b-a)(b'-a').
\]
\end{lemma}

\begin{proof}
For $x\in[a,b]^n$ and $y\in[a',b']^n$, define the affine rescalings
\[
\tilde x=\frac{x-a\1}{b-a}\in[0,1]^n,
\qquad
\tilde y=\frac{y-a'\1}{b'-a'}\in[0,1]^n,
\]
and define $\tilde x',\tilde y'$ analogously from $x',y'$. The neighboring relation is preserved under these coordinatewise affine maps.

Let $\bar x$ (resp.\ $\overline{\tilde x}$) denote the mean of $x$ (resp.\ $\tilde x$), and similarly for $y$. Since
\[
x-\bar x\,\1=(b-a)\big(\tilde x-\overline{\tilde x}\,\1\big),
\qquad
y-\bar y\,\1=(b'-a')\big(\tilde y-\overline{\tilde y}\,\1\big),
\]
we obtain the scaling identities
\[
\mathrm{nvar}(x)
=(b-a)^2\,\mathrm{nvar}(\tilde x),
\qquad
\mathrm{ncov}(x,y)
=(b-a)(b'-a')\,\mathrm{ncov}(\tilde x,\tilde y).
\]
Therefore,
\[
\big|\mathrm{nvar}(x)-\mathrm{nvar}(x')\big|
=(b-a)^2\big|\mathrm{nvar}(\tilde x)-\mathrm{nvar}(\tilde x')\big|,
\]
and
\[
\big|\mathrm{ncov}(x,y)-\mathrm{ncov}(x',y')\big|
=(b-a)(b'-a')\big|\mathrm{ncov}(\tilde x,\tilde y)-\mathrm{ncov}(\tilde x',\tilde y')\big|.
\]

By \cite[Lemma~\rev{18}]{alabi2022differentially}, applied to vectors in $[0,1]^n$, we have
\[
\big|\mathrm{nvar}(\tilde x)-\mathrm{nvar}(\tilde x')\big|
\le 1-\frac{1}{n},
\qquad
\big|\mathrm{ncov}(\tilde x,\tilde y)-\mathrm{ncov}(\tilde x',\tilde y')\big|
\le 1-\frac{1}{n}.
\]
Multiplying by the corresponding scale factors yields the stated bounds.
\end{proof}

Based on the above lemma, we give the formal proof of \rev{\Cref{thm:privacy-guarantee-continuous}}.

\begin{proof}
Let $(\mathcal{V},\mathcal{E},\mathbf{L})$ and $(\mathcal{V},\mathcal{E}',\mathbf{L}')$ be two edge-adjacent labeled networks. We verify that Algorithm~\ref{alg:noisystats-bounded} satisfies $(\varepsilon_\ell+\varepsilon_e,\delta_\ell)$ edge-adjacent differential privacy.

The first stage of Algorithm~\ref{alg:noisystats-bounded} perturbs the node attributes using truncated Laplace noise. By \citet{geng2018truncated}, the truncated Laplace mechanism achieves $(\varepsilon_\ell,\delta_\ell)$-differential privacy with respect to the node attributes. Hence, the released perturbed labels $\hat{\mathbf{L}}$ satisfy $(\varepsilon_\ell,\delta_\ell)$-DP.

\rev{Conditioning on the perturbed labels $\hat{\mathbf{L}}$, the covariates $\widehat x_i$ are fixed. Changing one undirected edge $(u,v)$ can alter only $\widehat y_u$ and $\widehat y_v$. Let $\widehat y'$ denote the vector of outcomes after this edge change, and write}
\[
\rev{m_x:=\frac{1}{n}\sum_{i=1}^n\widehat x_i,
\qquad
m_y:=\frac{1}{n}\sum_{i=1}^n\widehat y_i,
\qquad
m_y':=\frac{1}{n}\sum_{i=1}^n\widehat y_i'.}
\]
\rev{Consequently, $\operatorname{nvar}(\widehat x)$ is unchanged, while}
\[
\rev{\begin{aligned}
&\left|\operatorname{ncov}(\widehat x,\widehat y)
-\operatorname{ncov}(\widehat x,\widehat y')\right|\\
&\quad=
\left|\sum_{i\in\{u,v\}}(\widehat x_i-m_x)
(\widehat y_i-\widehat y_i')\right|\\
&\quad\leq \left(\max_i|\widehat x_i|\right)
\bigl(|\widehat x_u-m_x|+|\widehat x_v-m_x|\bigr)
\leq\Delta_1.
\end{aligned}}
\]
\rev{Here $\Delta_1$ is defined in Algorithm~\ref{alg:noisystats-bounded}. Adding an edge changes a nonisolated endpoint's friend average by at most $2\max_j|\widehat x_j|/(d_i+1)\le\max_j|\widehat x_j|$, and an isolated endpoint's by at most $\max_j|\widehat x_j|$; deletion is the reverse operation.}

\rev{Moreover,}
\[
\rev{|m_y-m_y'|
\leq \frac{2}{n}\max_i|\widehat x_i|=\Delta_2.}
\]
\rev{With the corresponding sensitivity calibration in Algorithm~\ref{alg:noisystats-bounded}, adding $L_1$ to the covariance and $L_2$ to $m_y$ gives two $\varepsilon_e/2$-DP releases. The coefficients are post-processing of these releases and the fixed privatized labels, and the failure check depends only on those labels. Sequential composition therefore implies that the second-stage mechanism satisfies $(\varepsilon_e,0)$-differential privacy with respect to changing one edge, conditional on the privatized labels.}

Therefore, the overall mechanism is the sequential composition of an $(\varepsilon_\ell,\delta_\ell)$-DP mechanism (on node attributes) and an $(\varepsilon_e,0)$ edge-adjacent DP mechanism (on edges). Applying \Cref{thm:edgeadj-composition-epsdelta} yields that the composed mechanism satisfies $(\varepsilon_\ell+\varepsilon_e,\delta_\ell)$ edge-adjacent differential privacy with respect to $(\mathcal{V},\mathcal{E},\mathbf{L})$.
\end{proof}

\subsection{Proof of Theorem~\ref{thm:regression-asymp-normal}: Asymptotic normality of the private debiased regression estimator}\label{app:thm:regression-asymp-normal}
\begin{proof}
\rev{We use the notation introduced before \cref{thm:regression-asymp-normal}.} \rev{The randomness in $M_n$ comes from the independent noise variables $z_1,\dots,z_n$ of \cref{alg:noisystats-bounded}, since $\hat x_k=x_k+z_k$.} Define the filtration
\[
\mathcal F_{n,k}:=\sigma(\rev{z_1,\dots,z_k}),
\qquad
k=0,1,\dots,n,
\]
with \(\mathcal F_{n,0}\) trivial. Define the \(\mathbb{R}^4\)-valued Doob martingale
\[
U_{n,k}:=\E[M_n\mid \mathcal F_{n,k}],
\qquad
k=0,1,\dots,n,
\]
and martingale differences
\[
D_{n,k}:=U_{n,k}-U_{n,k-1},
\qquad
k=1,\dots,n.
\]
Then
\[
U_{n,0}=\E[M_n],\qquad U_{n,n}=M_n,
\]
and hence
\[
M_n-\E[M_n]=\sum_{k=1}^n D_{n,k},
\qquad
\E[D_{n,k}\mid \mathcal F_{n,k-1}]=0.
\]

By construction, \(\hat x_i\) depends only on \rev{the noise variable \(z_i\)}, while \(\hat y_i\) depends only on \rev{the noise variables of the neighbors of \(i\)}. Therefore each \(\xi_{i,n}\) is a function of \rev{the noise variables \(z_j\) with \(j\in\{i\}\cup N(i)\). By \cref{ass:bounded-degree}, changing a single noise variable \(z_k\) can therefore affect only the summands \(\xi_{i,n}\) with \(i\in\{k\}\cup N(k)\), which number at most \(1+\Delta\).} Moreover, \(\hat x_i\) and \(\hat y_i\) are uniformly bounded, \rev{because the labels and the truncated-Laplace noise are bounded;} hence each coordinate of \(\xi_{i,n}\) is uniformly bounded. \rev{Therefore there exists a constant $C<\infty$ such that changing only $z_k$ changes $M_n$ by at most $C$ in the supremum norm.}

\rev{To justify the corresponding bound on the martingale differences, define
\[
H_{n,k}(s):=\E[M_n\mid\mathcal F_{n,k-1},z_k=s].
\]
Since $z_1,\ldots,z_n$ are independent, the conditional distribution of $z_{k+1},\ldots,z_n$ does not depend on the value of $z_k$. Thus, taking conditional expectations in the bounded-differences inequality gives
\[
\|H_{n,k}(s)-H_{n,k}(s')\|_\infty\le C,
\]
for any possible values $s,s'$ of $z_k$. Moreover,
\[
D_{n,k}
=H_{n,k}(z_k)
-\E[H_{n,k}(z_k)\mid\mathcal F_{n,k-1}].
\]
The preceding bound therefore implies
\[
\|D_{n,k}\|_\infty\le C
\qquad\text{a.s. for all }n,k.
\]}

Then define
\[
\Psi_{n,k}:=\E[D_{n,k}D_{n,k}^\top\mid \mathcal F_{n,k-1}],
\qquad
\bar\Psi_n:=\frac1n\sum_{k=1}^n \Psi_{n,k}.
\]
We claim that
\[
\bar\Psi_n-\E[\bar\Psi_n]\xrightarrow{p}0.
\]
\rev{Fix one entry of these $4\times4$ matrices, and denote by $\psi_{n,k}$ the corresponding entry of $\Psi_{n,k}$; the corresponding entry of $\bar\Psi_n$ is their average $\bar\psi_n:=n^{-1}\sum_{k=1}^n\psi_{n,k}$.} Since \(\|D_{n,k}\|_\infty\le C\), each \(\psi_{n,k}\) is uniformly bounded. \rev{Moreover, $\psi_{n,k}$ is a function of $z_1,\ldots,z_{k-1}$, and it depends on $z_r$ only if $r\in N(k)$: by the preceding paragraph, $z_k$ enters $M_n$ only through the summands $\xi_{i,n}$ with $i\in\{k\}\cup N(k)$, and in these summands $z_k$ is multiplied only by noise variables $z_r$ with $r\in N(k)$, through $\hat y_k$ and the products $\hat x_i\hat y_i$. Hence changing one noise variable $z_r$ changes $\psi_{n,k}$ only for $k\in N(r)$, that is, for at most $\Delta$ values of $k$, by \cref{ass:bounded-degree}. Consequently, there exists $C'<\infty$ such that changing only $z_r$ changes $\bar\psi_n$ by at most $C'/n$.} Applying McDiarmid's inequality yields, for every \(t>0\),
\[
\Pr\Bigl(\bigl|\rev{\bar\psi_n}-\E[\rev{\bar\psi_n}]\bigr|>t\Bigr)
\le 2\exp(-c n t^2)
\]
for some constant \(c>0\). Hence
\[
\rev{\bar\psi_n}-\E[\rev{\bar\psi_n}]\xrightarrow{p}0.
\]
\rev{Since this holds for each of the $16$ entries,} this implies
\[
\bar\Psi_n-\E[\bar\Psi_n]\xrightarrow{p}0.
\]

Next, by orthogonality of martingale differences,
\[
\E[\bar\Psi_n]
=
\frac1n\sum_{k=1}^n \E[D_{n,k}D_{n,k}^\top]
=
\frac1n\rev{\Cov}(M_n).
\]
\rev{By the covariance-limit assumption in the theorem,}
\[
\frac1n\Cov(M_n)\to \Sigma.
\]
Therefore
\[
\bar\Psi_n\xrightarrow{p}\Sigma.
\]

Consider the normalized martingale difference array
\[
X_{n,k}:=\frac{1}{\sqrt n}D_{n,k}.
\]
Since $D_{n,k}$ is uniformly bounded,
\[
\|X_{n,k}\|_\infty\le \frac{C}{\sqrt n}\to 0
\]
uniformly in \(k\). \rev{Hence, for every fixed $\omega\in\mathbb R^4$, the scalar increments $\omega^\top X_{n,k}$ are bounded by $\|\omega\|_1C/\sqrt n$, so they satisfy Lindeberg's condition, as in Step~3 of the proof of \cref{thm:relative-variance-orders}.}
\[
\sum_{k=1}^n \E[X_{n,k}X_{n,k}^\top\mid \mathcal F_{n,k-1}]
=
\frac1n\sum_{k=1}^n \Psi_{n,k}
=
\bar\Psi_n
\xrightarrow{p}\Sigma.
\]
\rev{Therefore, by the argument of Step~3 in the proof of \cref{thm:relative-variance-orders}, which applies the martingale central limit theorem \citep[][Corollary~3.1]{hall1980martingale} to each linear combination $\omega^\top M_n$ and then the Cram\'er--Wold theorem \citep{cramer1936some},}
\[
\frac{1}{\sqrt n}\bigl(M_n-\E[M_n]\bigr)
\Rightarrow
\mathcal N(0,\Sigma).
\]
\rev{Equivalently, since $\bar M_n=M_n/n$ and $\mu_n=\E[\bar M_n]$,}
\[
\sqrt n(\bar M_n-\mu_n)\Rightarrow \mathcal N(0,\Sigma).
\]

Then we \rev{establish} the asymptotic normality of the variance-corrected estimator \rev{$\tilde\beta_n=h(\bar M_n)$. Since $\mu_3-\mu_1^2-\sigma_z^2>0$, the map $h$ is continuously differentiable on a ball $U$ around $\mu$. With probability tending to one, both $\bar M_n$ and $\mu_n$ lie in $U$, and then the mean value theorem gives a point $m_n$ on the segment between them such that
\[
\sqrt n\bigl(\tilde\beta_n-h(\mu_n)\bigr)=\nabla h(m_n)^\top\sqrt n\,(\bar M_n-\mu_n).
\]
Since $\bar M_n-\mu_n\xrightarrow{p}0$ and $\mu_n\to\mu$, we have $m_n\xrightarrow{p}\mu$, and hence $\nabla h(m_n)\xrightarrow{p}\nabla h(\mu)$. By Slutsky's theorem,}\footnote{\rev{Equivalently, this step is an instance of the uniform delta method \citep[][Theorem~3.8]{vandervaart1998asymptotic}, which extends the delta method to a centering that changes with $n$.}}
\[
\sqrt n\bigl(\tilde\beta_n-h(\mu_n)\bigr)
\Rightarrow
\mathcal N\!\left(0,\nabla h(\mu)^\top\Sigma\nabla h(\mu)\right).
\]
\rev{Finally, $h(\mu_n)=\beta_n$. Indeed, the moment identities in the proof of \cref{prop:noisystats} give $\mu_n=n^{-1}\sum_{i=1}^n(x_i,\,y_i,\,x_i^2+\sigma_z^2,\,x_iy_i)^\top$, so $h(\mu_n)$ is the empirical covariance of the pairs $(x_i,y_i)$ divided by the empirical variance of the $x_i$, which is $\beta_n$. Therefore
\[
\sqrt n(\tilde\beta_n-\beta_n)\Rightarrow \mathcal N(0,\tau^2),
\qquad
\tau^2=\nabla h(\mu)^\top\Sigma\nabla h(\mu).
\]}

\rev{Next, we show that \cref{alg:noisystats-bounded} returns $\perp$ with vanishing probability and that the Laplace noise $L_1$ is negligible.} Let
\[
\operatorname{ncov}_n:=\operatorname{ncov}(\hat x,\hat y),
\qquad
\operatorname{nvar}_n:=\operatorname{nvar}(\hat x).
\]
\rev{Since $\bar M_n\xrightarrow{p}\mu$ and $\operatorname{nvar}_n/n$ is a continuous function of $\bar M_n$, we have $\operatorname{nvar}_n/n\xrightarrow{p}\mu_3-\mu_1^2$, and $\mu_3-\mu_1^2>\sigma_z^2$ by assumption. Hence $\Pr\bigl(\operatorname{nvar}_n>(n-1)\sigma_z^2\bigr)\to1$, so the probability that \cref{alg:noisystats-bounded} returns $\perp$ tends to zero. On the complementary event, \cref{alg:noisystats-bounded} returns
\[
\tilde\beta_n^{\mathrm{DP}}=\frac{\operatorname{ncov}_n+L_1}{\operatorname{nvar}_n-(n-1)\sigma_z^2},
\qquad\text{while}\qquad
\tilde\beta_n=h(\bar M_n)=\frac{\operatorname{ncov}_n}{\operatorname{nvar}_n-n\sigma_z^2}.
\]}
\rev{Hence
\[
\tilde\beta_n^{\mathrm{DP}}-\tilde\beta_n
=\frac{L_1}{\operatorname{nvar}_n-(n-1)\sigma_z^2}
-\frac{\sigma_z^2\operatorname{ncov}_n}{\bigl(\operatorname{nvar}_n-(n-1)\sigma_z^2\bigr)\bigl(\operatorname{nvar}_n-n\sigma_z^2\bigr)}.
\]
By the limit above, both denominators $\operatorname{nvar}_n-(n-1)\sigma_z^2$ and $\operatorname{nvar}_n-n\sigma_z^2$, divided by $n$, converge in probability to $\mu_3-\mu_1^2-\sigma_z^2>0$. Moreover, $|\operatorname{ncov}_n|\le Cn$ for a constant $C$, because the privatized observations are bounded, and $L_1=O_p(1)$, because $L_1\sim\operatorname{Lap}(0,2\Delta_1/\varepsilon_e)$ and the sensitivity bound $\Delta_1$ is bounded uniformly in $n$.} Therefore
\[
\tilde{\beta}_n^{\mathrm{DP}}-\tilde\beta_n=O_p(n^{-1}),
\]
and thus
\[
\sqrt n\,(\tilde{\beta}_n^{\mathrm{DP}}-\tilde\beta_n)\xrightarrow{p}0.
\]

Finally,
\[
\sqrt n(\tilde{\beta}_n^{\mathrm{DP}}-\rev{\beta_n})
=
\sqrt n(\tilde\beta_n-\rev{\beta_n})
+
\sqrt n(\tilde{\beta}_n^{\mathrm{DP}}-\tilde\beta_n).
\]
\rev{Since
\[
\sqrt n(\tilde\beta_n-\rev{\beta_n})\Rightarrow \mathcal N(0,\tau^2)
\]
and
\[
\sqrt n(\tilde{\beta}_n^{\mathrm{DP}}-\tilde\beta_n)\xrightarrow{p}0,
\]
Slutsky's theorem gives
\[
\sqrt n(\tilde{\beta}_n^{\mathrm{DP}}-\rev{\beta_n})\Rightarrow \mathcal N(0,\tau^2).
\]}
This completes the proof.
\end{proof}

\section{Additional Material}

\subsection{Products frequently purchased together on Amazon}\label{sec:appx-amazon}

    This section demonstrates an empirical application of our method with continuous labels on data collected by \cite{leskovec2007dynamics} and \cite{yang2012defining} on the characteristics of Amazon products and whether or not they are frequently purchased together. The patterns of complementarity between different goods is a classic question in economics with implications for business pricing, market structure, and antitrust regulation (see, for example, \cite{gentzkow2007valuing}). Specifically, we consider a labeled graph where the label of each node corresponds to that node's percentile rank in the sales distribution (scaled to be in $[0,1]$) and an edge exists between two products if Amazon's website says they are frequently purchased together. In Figure~\ref{fig:amazon-reg} we show that high sales products are frequently co-purchased with other high-sales products. We run 1,000 simulations of our method for continuous labels and plot the distribution of regression slopes in Figure~\ref{fig:amazon-reg-private}.

\begin{figure}[htbp]
    \centering
    \begin{subfigure}[t]{0.48\linewidth}
        \centering
        \includegraphics[width=\linewidth]{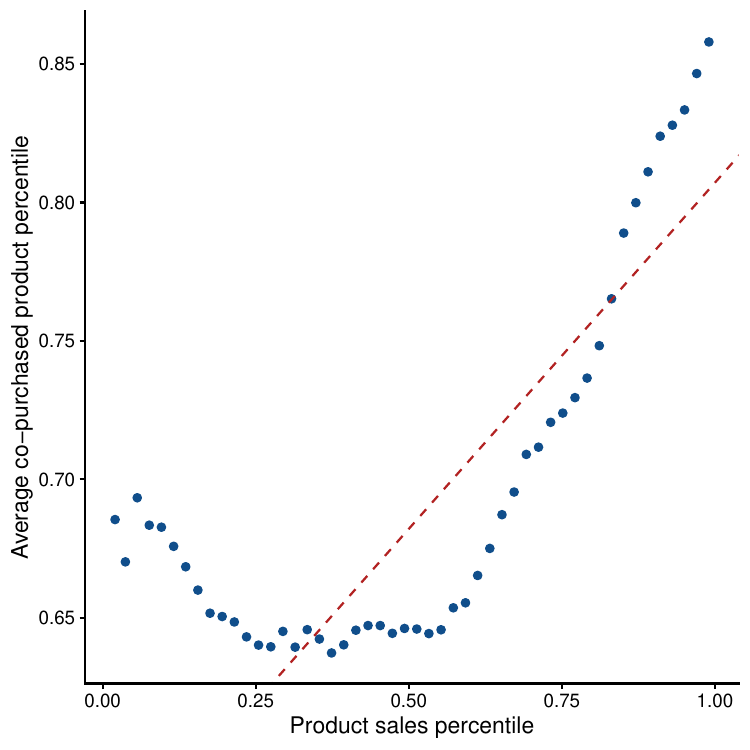}
        \caption{Sales rank of frequently co-purchased products vs own sales rank.}
        \label{fig:amazon-reg}
    \end{subfigure}\hfill
    \begin{subfigure}[t]{0.48\linewidth}
        \centering
        \includegraphics[width=\linewidth]{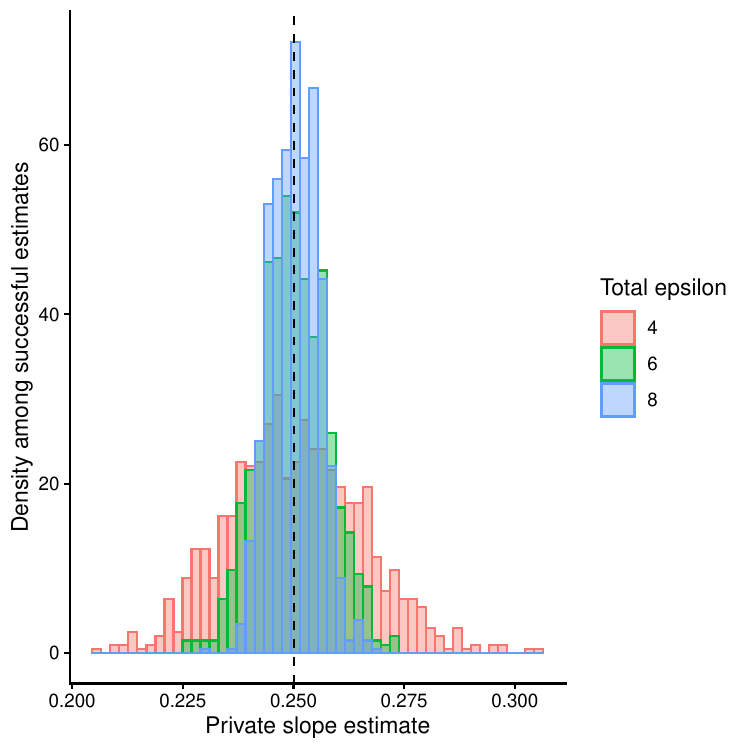}
        \caption{Distribution of private regression slopes for various values of $\varepsilon$.}
        \label{fig:amazon-reg-private}
    \end{subfigure}

    \caption{\textbf{Amazon products frequently purchased together.} Panel (a) plots the average sales rank (percentile) of co-purchased products against a product’s own sales rank (percentile). Panel (b) shows the distribution of private slope estimates across 1,000 simulations; the dotted red line marks the true slope. \rev{We divide the privacy budget equally between $\varepsilon_\ell$ and $\varepsilon_e$.}}
    \label{fig:amazon-combined}
\end{figure}

\clearpage

\begin{figure}
        \centering
        \includegraphics[width=0.8\linewidth]{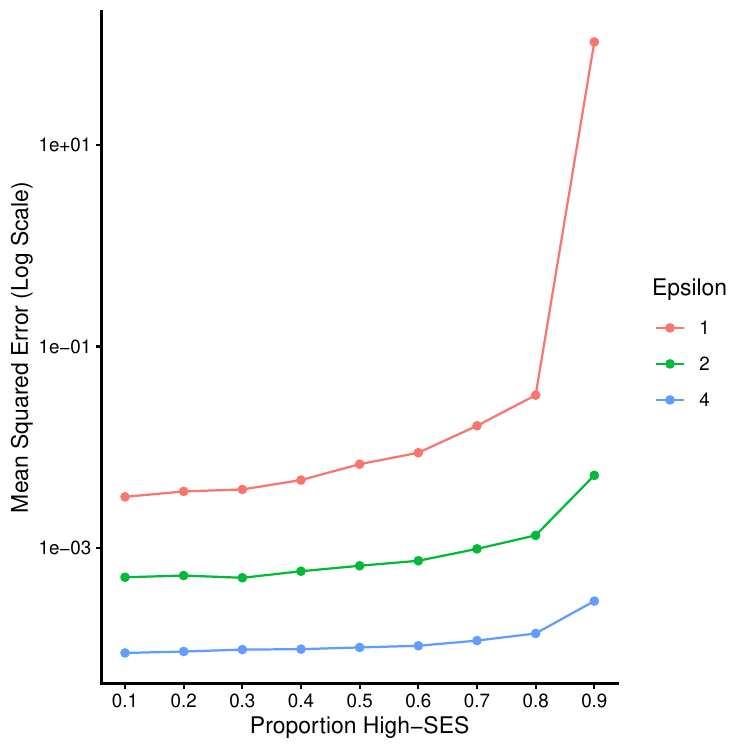}
        \caption{\textbf{Mean squared error vs cell composition for simulated networks by cell composition.} The figure illustrates the relationship between the proportion of High-SES individuals in a network and the resulting MSE of the private estimator across three privacy budgets ($\varepsilon \in \{1, 2, 4\}$), with the privacy budget split equally between $\varepsilon_e$ and $\varepsilon_\ell$. Results are averaged over 1,125 simulations per data point ($75$ graphs $\times 15$ noise seeds). The network is generated as an Erd\H{o}s-R\'{e}nyi graph with a connection probability of 0.04 on 2,000 nodes. On the $x$-axis, we vary the proportion of nodes in the ``high-SES'' set (corresponding to the application in \cite{chettyetal2022I}). The $y$-axis uses a log scale. }
        \label{fig:accuracy-cellcomp}
    \end{figure}

\begin{figure}
        \centering
        \includegraphics[width=0.8\linewidth]{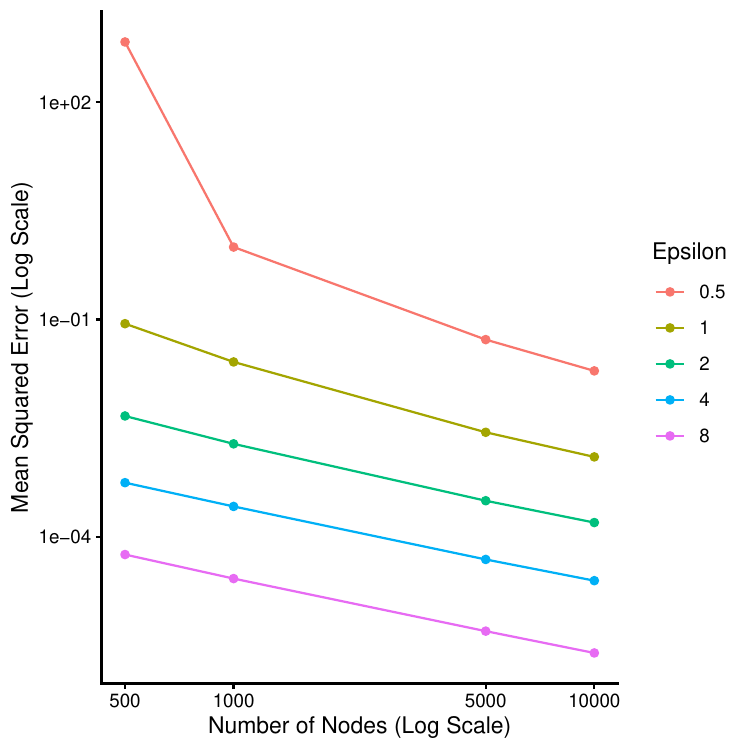}
        \caption{\textbf{Mean squared error vs network size for simulated networks.} This figure displays the Mean Squared Error (MSE) of our differentially private estimator as a function of network size, under varying privacy budgets ($\varepsilon$). The results are based on 10,000 Monte Carlo simulations on simulated Erd\H{o}s-R\'{e}nyi graphs with nodes randomly partitioned into two equal-sized groups, and varying the connection probability as the graph grows to maintain an \rev{expected} degree of 20. We split our privacy budget equally between $\varepsilon_e$ and $\varepsilon_\ell$.}
        \label{fig:accuracy-networksize}
    \end{figure}

\begin{table}[ht]
    \centering
    \caption{Summary statistics for villages from \cite{banerjee_diffusion_2013}.}
    \footnotesize
    \setlength{\tabcolsep}{4pt}
    \begin{tabular}{ccccccc}
    \hline
    Village \# & \# HH & \# Hist. Dis. & \# Hist. Non-Dis. & Avg. Deg. & Cross-Caste Conn. & SD ($\varepsilon$=8) \\
    \hline
    28 & 315 & 93 & 222 & 8.7 & 0.26 & \rev{0.02} \\
    29 & 272 & 155 & 117 & 7.3 & 0.21 & \rev{0.01} \\
    30 & 138 & 64 & 74 & 8.7 & 0.19 & \rev{0.02} \\
    31 & 151 & 141 & 10 & 7.9 & 0.03 & \rev{0.02} \\
    32 & 241 & 116 & 125 & 9.7 & 0.22 & \rev{0.02} \\
    33 & 204 & 26 & 178 & 7.4 & 0.35 & \rev{0.06} \\
    34 & 165 & 103 & 62 & 6.1 & 0.22 & \rev{0.02} \\
    35 & 206 & 107 & 99 & 6.7 & 0.22 & \rev{0.02} \\
    36 & 289 & 92 & 197 & 10.4 & 0.41 & \rev{0.02} \\
    38 & 157 & 31 & 126 & 6.7 & 0.32 & \rev{0.04} \\
    39 & 287 & 64 & 223 & 8.4 & 0.52 & \rev{0.02} \\
    40 & 240 & 62 & 178 & 8.1 & 0.21 & \rev{0.03} \\
    42 & 192 & 126 & 66 & 7.4 & 0.04 & \rev{0.02} \\
    44 & 227 & 32 & 195 & 9.5 & 0.27 & \rev{0.06} \\
    45 & 219 & 59 & 160 & 7.7 & 0.37 & \rev{0.02} \\
    46 & 261 & 169 & 92 & 8.3 & 0.13 & \rev{0.01} \\
    47 & 137 & 26 & 111 & 9.3 & 0.36 & \rev{0.05} \\
    48 & 182 & 86 & 96 & 9.8 & 0.26 & \rev{0.02} \\
    49 & 193 & 44 & 149 & 9.3 & 0.32 & \rev{0.03} \\
    50 & 244 & 60 & 184 & 10.3 & 0.26 & \rev{0.03} \\
    51 & 248 & 103 & 145 & 12.9 & 0.36 & \rev{0.02} \\
    52 & 327 & 79 & 248 & 12.3 & 0.36 & \rev{0.03} \\
    53 & 151 & 48 & 103 & 11.5 & 0.41 & \rev{0.03} \\
    54 & 99 & 57 & 42 & 11.1 & 0.15 & \rev{0.02} \\
    55 & 257 & 52 & 205 & 6.9 & 0.11 & \rev{0.04} \\
    57 & 208 & 201 & 7 & 9.1 & 0.02 & \rev{0.01} \\
    58 & 177 & 51 & 126 & 9.2 & 0.41 & \rev{0.03} \\
    59 & 328 & 110 & 218 & 8.6 & 0.29 & \rev{0.02} \\
    60 & 354 & 101 & 253 & 8.0 & 0.21 & \rev{0.02} \\
    61 & 121 & 31 & 90 & 8.3 & 0.42 & \rev{0.04} \\
    62 & 189 & 67 & 122 & 8.7 & 0.25 & \rev{0.02} \\
    63 & 161 & 29 & 132 & 6.6 & 0.15 & \rev{0.06} \\
    64 & 257 & 78 & 179 & 7.5 & 0.28 & \rev{0.02} \\
    65 & 285 & 86 & 199 & 10.7 & 0.19 & \rev{0.02} \\
    66 & 183 & 40 & 143 & 8.5 & 0.37 & \rev{0.04} \\
    67 & 193 & 92 & 101 & 10.5 & 0.13 & \rev{0.02} \\
    68 & 153 & 35 & 118 & 9.7 & 0.28 & \rev{0.04} \\
    69 & 180 & 94 & 86 & 13.4 & 0.19 & \rev{0.02} \\
    70 & 205 & 45 & 160 & 12.6 & 0.57 & \rev{0.03} \\
    71 & 297 & 63 & 234 & 10.2 & 0.49 & \rev{0.02} \\
    72 & 223 & 123 & 100 & 11.0 & 0.19 & \rev{0.02} \\
    73 & 164 & 78 & 86 & 10.5 & 0.28 & \rev{0.02} \\
    74 & 170 & 32 & 138 & 7.4 & 0.17 & \rev{0.05} \\
    75 & 166 & 55 & 111 & 11.3 & 0.36 & \rev{0.02} \\
    76 & 251 & 62 & 189 & 7.6 & 0.30 & \rev{0.03} \\
    77 & 153 & 83 & 70 & 7.6 & 0.14 & \rev{0.02} \\
    \hline
    \end{tabular}
    \label{tab:summary-caste}
    \caption*{\small \\\textit{Notes for Table~\ref{tab:summary-caste}: }This table summarizes the network and demographic characteristics for the subset of villages from \cite{banerjee_diffusion_2013} with \rev{caste coverage above 50\%, more than six historically disadvantaged households, and at least two households in each caste group after removing observations with missing caste.}}
\end{table}

\clearpage
\FloatBarrier
\subsection{List of Notation}

\vspace{5pt}

    \begin{center}
    \captionof{table}{List of Notation}
        \label{tab:notation-full}
        \begin{tabular}{@{}ll@{}}
        \toprule
        \textbf{Symbol} & \textbf{Description} \\ \midrule
        $\mathcal{V}$ & Set of vertices (nodes) in the network \\
        $\mathcal{E}$ & Set of edges (friendships/connections) between nodes \\
       $\mathbf{L}=(l_i)_{i\in\mathcal V}$ & Node label vector with $l_i\in\{a,b\}$ for each node $i$ \\
        $e_{ij}$ & Binary indicator; $e_{ij}=1$ if an edge exists between $i$ and $j$, else $0$ \\
        $d_i$ & Degree of node $i$, $\#(j \in \mathcal{V}:e_{ij} = 1)$. \\ 
        $a_{ij}$ & $e_{ij} / d_i$, with $a_{ij} = 0$ when $d_i = 0$. \\ 
        $N(i)$ & Neighborhood of node $i$, $\{j \in \mathcal{V}:e_{ij} = 1\}$. \\ 
        $\mathcal{A}, \mathcal{B}$ & Partitions of $\mathcal{V}$ based on labels ($a, b$) \\
        $\pi_A$ & Fraction of nodes in set $\mathcal{A}$. \\
        $\#( \cdot )$ & Cardinality operator (number of elements in a set) \\
        $\rho_i$ & Individual $i$'s fraction of friends belonging to the target group \\
        $C^{\mathcal{A}\to\mathcal{B}}$ & Cross-type connectedness index (average $\rho_i$ for individuals in group $\mathcal{A}$) \\
        $C^{\mathcal{A}\to\mathcal{A}}$ & Same-type connectedness index \\
        $s$ & A specific ``cell" or subset of users (e.g., a county or school) \\
        $C_s^{\mathcal{A}\to\mathcal{B}}$ & Connectedness index calculated specifically for cell $s$ \\
        $\varepsilon$ & Privacy-loss budget (differential privacy parameter) \\
        $\varepsilon_\ell$ & Privacy-loss budget for labels \\
        $\varepsilon_e$ & Privacy-loss budget for edges \\ 
        $\mathcal{M}_1, \mathcal{M}_2$ & Privacy-preserving mechanisms for protecting labels $\cM_1$ and edges $\cM_2$ \\
        $D, D'$ & Edge-adjacent labeled networks \\ \bottomrule
        \end{tabular}
    \end{center}

\end{document}